%% file: main.tex
\documentclass[a4paper,11pt]{article}
\pdfoutput=1
\usepackage[english]{babel}
\usepackage[T1]{fontenc}

\usepackage{microtype}
\usepackage{mdframed}
\usepackage{lineno}
\usepackage{csquotes}
\usepackage[dvipsnames,svgnames,x11names]{xcolor}
\usepackage{soul}
\usepackage{mathtools, amsthm, amssymb}
\usepackage{thm-restate}
\usepackage{subcaption}
\usepackage{listings}
\usepackage{mathrsfs}
\usepackage[mathscr]{euscript}
\usepackage{amsfonts,amsthm,amsmath}
\usepackage{amssymb}
\usepackage{amstext}
\usepackage{graphicx,tikz-cd}
\RequirePackage{fancyhdr}
\usepackage{boxedminipage}
\usepackage{algorithm}
\usepackage[noend]{algpseudocode}
\usepackage{graphics}
\usepackage{pgf}
\usepackage{paralist} 
\usepackage{framed}
\usepackage{thm-restate}
\usepackage{thmtools}
\usepackage{enumerate}
\usepackage{float}
\usepackage{xspace}
\usepackage{graphicx}
\usepackage{caption}
\usepackage{subcaption}
\usepackage{enumitem}
\usepackage{makecell}
\usepackage{thmtools, thm-restate}
\usepackage{thm-autoref}
 \usepackage{mathdots}
\usepackage{multirow}
\usepackage{chngcntr}
\usepackage{apptools}
\usepackage{textpos}
\usepackage{accents}

\usepackage{vmargin}
\setmarginsrb{1in}{1in}{1in}{1in}{0pt}{0pt}{0pt}{6mm}

\usepackage{booktabs}

\numberwithin{equation}{section}
\newtheorem{theorem}[equation]{Theorem}

\newtheorem{lemma}[equation]{Lemma}
\newtheorem{claim}[equation]{Claim}

\newtheorem{observation}[equation]{Observation}
\newtheorem{invariant}[equation]{Invariant}

\AtAppendix{\counterwithin{lemma}{section}}
\AtAppendix{\counterwithin{claim}{section}}
\AtAppendix{\counterwithin{proposition}{section}}

\newcommand{\cqed}{\ensuremath{\lhd}}
\makeatletter
\newenvironment{claimproof}{\par
  \pushQED{\cqed}%
  \normalfont \topsep6\p@\@plus6\p@\relax
  \trivlist
  \item\relax
  {\itshape
    Proof of the claim\@addpunct{.}}\hspace\labelsep\ignorespaces
}{%
  \hfill\popQED\endtrivlist\@endpefalse
}
\makeatother

\theoremstyle{definition}
\newtheorem{definition}[theorem]{Definition}

\usepackage[pdftex, plainpages = false, pdfpagelabels, 
                 bookmarks=false,
                 bookmarksopen = true,
                 bookmarksnumbered = true,
                 breaklinks = true,
                 linktocpage,
                 pagebackref,
                 colorlinks = true,  
                 linkcolor = blue,
                 urlcolor  = blue,
                 citecolor = red,
                 anchorcolor = green,
                 hyperindex = true,
                 hypertexnames = false,
                 hyperfigures
                 ]{hyperref}

\usepackage[nameinlink]{cleveref}
\crefname{observation}{Observation}{Observations}
\crefname{claim}{Claim}{Claims}
\crefname{hypothesis}{Hypothesis}{Hypotheses}
\crefname{mainlemma}{Main lemma}{Main lemmas}
\crefname{invariant}{Invariant}{Invariants}

\input{macros}

\begin{document}

\pagenumbering{Roman}

\hypersetup{pageanchor=false}
\title{Dynamic Meta-Kernelization\thanks{This work was supported by the VILLUM Foundation, Grant Number 54451, Basic Algorithms Research Copenhagen (BARC). T.K was supported by the
European Union under Marie Skłodowska-Curie Actions (MSCA), project no. 101206430.}}
\author{
Christian Bertram\thanks{University of Copenhagen, Denmark. \texttt{\{chbe@di.ku.dk, mvje@di.ku.dk, tuko@di.ku.dk\}}}
\and
Deborah Haun\thanks{Karlsruhe Institute of Technology, Germany. \texttt{deborah.haun@student.kit.edu}}
\and
Mads Vestergaard Jensen\addtocounter{footnote}{-2}\footnotemark
\and
Tuukka Korhonen\addtocounter{footnote}{-1}\footnotemark}
\date{}

\maketitle
\thispagestyle{empty}
\input{abstract}
 \begin{textblock}{20}(0, 4.8)
 \includegraphics[width=100px]{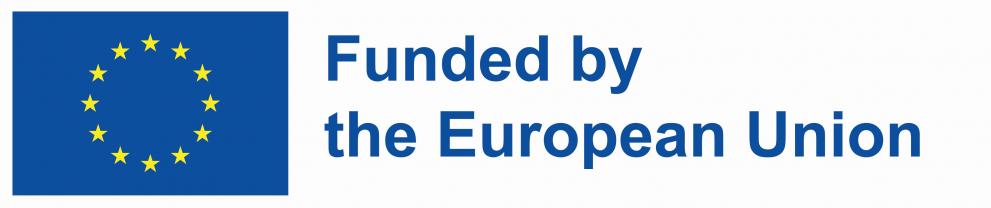}%
 \end{textblock}
\thispagestyle{empty}
\newpage
\pagenumbering{roman}
\setcounter{page}{1}
\setcounter{tocdepth}{2}
\tableofcontents
\newpage
\pagenumbering{arabic}
\hypersetup{pageanchor=true}
\clearpage
\setcounter{page}{1}
\input{introduction}
\input{overview}
\input{preliminaries}

\input{existence}
\input{local-search}
\input{balance-protrusions}
\input{main-data-structure}
\input{kernelization}
\input{conclusion}

\bibliographystyle{alpha}
\bibliography{refs}
\appendix
\input{appendix.tex}

\end{document}

%% file: macros.tex
\newcommand{\mc}{\mathcal}

\newcommand{\incidences}{N^{\mathsf{inc}}}

\newcommand{\poly}{\mathsf{poly}}

\newcommand{\tw}{\mathsf{tw}}
\newcommand{\itw}{\mathsf{itw}}

\newcommand{\bd}{\mathsf{bd}}
\newcommand{\wl}{\mathsf{wl}}
\newcommand{\torso}{\mathsf{torso}}
\newcommand{\inter}{\mathsf{int}}
\newcommand{\adh}{\mathsf{adh}}
\newcommand{\adhsize}{\mathsf{adhsize}}
\newcommand{\chd}{\mathsf{chd}}
\newcommand{\rank}{\mathsf{rk}}
\newcommand{\cl}{\mathsf{cl}}
\newcommand{\desc}{\mathsf{desc}}
\newcommand{\anc}{\mathsf{anc}}
\newcommand{\trace}{\mathsf{trace}}
\newcommand{\depth}{\mathsf{depth}}
\newcommand{\Vint}{{V_{\mathsf{int}}}}
\newcommand{\expand}{\vartriangleright}

\newcommand{\chips}{\mathsf{chips}}
\newcommand{\vol}{\mathsf{vol}}

\newcommand{\bbst}{\mathcal{S}}

\newcommand{\opt}{\mathrm{OPT}}
\newcommand{\optTwMod}{\mathsf{tw}\text{-}\mathsf{mod}_{\eta}}
\newcommand{\bag}{\mathsf{bag}}
\newcommand{\edges}{\mathsf{edges}}
\newcommand{\EL}{\mathsf{EL}}
\newcommand{\cmso}{\mathrm{CMSO}_2}
\newcommand{\mso}{\mathrm{MSO}_2}

\newcommand{\init}{\mathsf{Init}}
\newcommand{\addEdge}{\mathsf{AddEdge}}
\newcommand{\deleteEdge}{\mathsf{DeleteEdge}}
\newcommand{\addHyperedge}{\mathsf{AddHyperedge}}
\newcommand{\deleteHyperedge}{\mathsf{DeleteHyperedge}}

\newcommand{\query}{\mathsf{Query}}
\newcommand{\addVertex}{\mathsf{AddVertex}}
\newcommand{\deleteVertex}{\mathsf{DeleteVertex}}

\newcommand{\merge}{\mathsf{Merge}}

\newcommand{\true}{\top}
\newcommand{\false}{\bot}

\newcommand{\OO}{\mathcal{O}} 
\newcommand{\Tc}{\mathcal{T}} 
\newcommand{\Lc}{\mathcal{L}} 
\newcommand{\Sc}{\mathcal{S}} 
\newcommand{\Hc}{\mathcal{H}} 
\newcommand{\Pc}{\mathcal{P}} 
\newcommand{\oracle}{\mathsf{O}}
\newcommand{\leaves}{\mathsf{L}}

\newcommand{\compl}[1]{\overline{#1}}

\newcommand{\N}{\mathbb{N}}

\newcommand{\ZO}{\mathbb{Z}_{\geq 0}}
\newcommand{\ZI}{\mathbb{Z}_{\geq 1}}

\newcommand{\labelSet}{\mathbf{\Lambda}}

\newcommand{\oplusb}{\oplus_b}
\newcommand{\oplusu}{\oplus_u}

\newcommand{\cTwMod}{\eta}

\newcommand{\protboundary}{\partial} 

\newcommand{\cSbdProt}{\alpha} 

\newcommand{\cSbdRootDeg}{\delta} 

\newcommand{\cTorsoVertices}{\delta_V} 

\newcommand{\cOldPdRoot}{d_{\ref{lem:modulator_implies_decomposition}}}

\newcommand{\cPdProt}{\omega} 

\newcommand{\cTorsoIntersect}{\sigma} 
\newcommand{\cNewSbdProt}{\omega} 

\newcommand{\run}{\rho}
\newcommand{\autom}{\mathcal{A}}

\newcommand{\cFinalProt}{\gamma}
\newcommand{\cUpperBoundSize}{\sigma}

\newcommand{\cUBTorsoOps}{\zeta}

\newcommand{\tuukkain}[1]{\todo[size=\normalsize,inline,color=green!40]{Tuukka: #1}}

\newcounter{thmcond}
\renewcommand{\thethmcond}{\arabic{thmcond}}
\crefformat{thmcond}{#2(#1)#3}        
\Crefformat{thmcond}{#2(#1)#3}        
\crefrangeformat{thmcond}{#3(#1)#4--#5(#2)#6} 

%% file: abstract.tex
\begin{abstract}
Kernelization studies polynomial-time preprocessing algorithms.
Over the last 20 years, the most celebrated positive results of the field have been linear kernels for classical NP-hard graph problems on sparse graph classes.
In this paper, we lift these results to the dynamic setting.

As the canonical example, Alber, Fellows, and Niedermeier~[J. ACM~2004] gave a linear kernel for dominating set on planar graphs.
We provide the following dynamic version of their kernel:
Our data structure is initialized with an $n$-vertex planar graph $G$ in $\OO(n \log n)$ amortized time, and, at initialization, outputs a planar graph $K$ with $\opt(K) = \opt(G)$ and $|K| = \OO(\opt(G))$, where $\opt(\cdot)$ denotes the size of a minimum dominating set.
The graph $G$ can be updated by insertions and deletions of edges and isolated vertices in $\OO(\log n)$ amortized time per update, under the promise that it remains planar.
After each update to $G$, the data structure outputs $\OO(1)$ updates to $K$, maintaining $\opt(K) = \opt(G)$, $|K| = \OO(\opt(G))$, and planarity of $K$.

Furthermore, we obtain similar dynamic kernelization algorithms for all problems satisfying certain conditions on (topological-)minor-free graph classes.
Besides kernelization, this directly implies new dynamic constant-approximation algorithms and improvements to dynamic FPT algorithms for such problems.

Our main technical contribution is a dynamic data structure for maintaining an approximately optimal \emph{protrusion decomposition} of a dynamic topological-minor-free graph.
Protrusion decompositions were introduced by Bodlaender, Fomin, Lokshtanov, Penninkx, Saurabh, and Thilikos~[J.~ACM~2016], and have since developed into a part of the core toolbox in kernelization and parameterized algorithms.
\end{abstract}

%% file: introduction.tex
\section{Introduction}
\label{sec:intro}
The field of kernelization studies polynomial-time preprocessing algorithms for NP-hard problems.
In order to achieve worst-case performance guarantees in this setting, one must consider parameterized problems where the hardness of an instance is captured by a parameter.
A \emph{kernelization algorithm}, also called simply a \emph{kernel}, for a parameterized problem takes an input $(I,k)$, runs in time $\poly(|I|,k)$, and outputs an instance $(K, k')$ of the problem so that (1) $(K,k')$ is a yes-instance if and only if $(I,k)$ is a yes-instance, and (2) $|K|+k' \le f(k)$ for a computable function $f$~\cite{Downey_Fellows_1999}.

It is a classic observation that a problem has a kernel if and only if it is fixed-parameter tractable~\cite{DBLP:journals/apal/CaiCDF97}.
This motivates the notions of a \emph{polynomial kernel} and a \emph{linear kernel}, which restrict the function $f$ in the definition above to be a polynomial (resp. linear) function.
Lower bounds for kernels are known under the assumption that $\mathsf{coNP} \not\subseteq \mathsf{NP}/\mathsf{poly}$, showing that some natural problems, such as \textsc{Longest Path} parameterized by the length, do not (likely) have polynomial kernels despite being fixed-parameter tractable~\cite{Bodlaender_Downey_Fellows_Hermelin_2009,DBLP:journals/jcss/FortnowS11}.
Similarly, \textsc{Feedback Vertex Set} has a kernel with $\OO(k^2)$ edges~\cite{DBLP:journals/talg/Thomasse10}, but not with $\OO(k^{2-\varepsilon})$ edges for any $\varepsilon > 0$, unless $\mathsf{coNP} \subseteq \mathsf{NP}/\mathsf{poly}$~\cite{DBLP:journals/jacm/DellM14}. 
We refer the reader to~\cite{kernelization-book} for a recent textbook on kernelization.

In this paper, we design \emph{dynamic} linear kernels for graph problems on sparse graph classes.
This means designing data structures that support updating the input graph while simultaneously efficiently maintaining a kernel.
To be concrete about our contribution before diving deeper into the literature, let us start by stating a representative example of a result that we obtain.
It provides a dynamic linear kernel for \textsc{Dominating Set} (parameterized by the solution size) on planar graphs.

\begin{theorem}
\label{thm:introdomset}
There is a data structure that is initialized with a planar graph $G$ in $\OO(|G| \log |G|)$\footnote{We denote by $|G| = |V(G)|+|E(G)|$ the total number of vertices and edges of a graph $G$.} amortized time, supports updating $G$ via insertions and deletions of edges and isolated vertices in amortized $\OO(\log |G|)$ time per update under the promise that $G$ remains planar, and throughout maintains a graph $K$ so that 
\begin{enumerate}
\item $K$ is planar,
\item $\opt(K) = \opt(G)$, where $\opt(\cdot)$ denotes the size of a minimum dominating set, and
\item $|K| \le \OO(\opt(G))$.
\end{enumerate}
The data structure outputs $K$ at the initialization, and after each update it outputs at most $\OO(1)$ updates to $K$,  which consist of insertions and deletions of edges and isolated vertices.
\end{theorem}

\Cref{thm:introdomset} follows from our meta-theorem \Cref{thm:kernelization}, which provides similar results for any pair of a problem and a graph class that satisfy certain conditions.

The reader may notice that the kernel of \Cref{thm:introdomset} does not quite match the definition we provided a couple of paragraphs earlier, but is in fact slightly stronger, working for all values $k$ of the parameter simultaneously.
We also highlight the fact that the graph $K$ is updated by (worst-case) $\OO(1)$ updates per update to $G$, which is a non-trivial result under any polynomial running time.
It enables efficient chaining of \Cref{thm:introdomset} with other dynamic data structures.
Another observation is that the data structure of \Cref{thm:introdomset} also maintains a constant approximation to the minimum dominating set size due to the inequalities $\opt(G) \le |K| \le \OO(\opt(G))$.

Let us then discuss the literature of linear kernels on sparse graph classes before stating our results in the full generality.

\paragraph{Linear kernels on sparse graph classes.}
Among the most influential results in kernelization are linear kernels for NP-hard graph problems parameterized by solution size on sparse graph classes.
The first such result was by Alber, Fellows, and Niedermeier~\cite{Alber_Fellows_Niedermeier_2004}, who gave a linear kernel for \textsc{Dominating Set} on planar graphs with running time $\OO(n^3)$.
We note that even the existence of any kernel for \textsc{Dominating Set} on planar graphs is non-trivial and perhaps surprising, as on general graphs \textsc{Dominating Set} is $\mathsf{W[2]}$-complete~\cite{DBLP:journals/siamcomp/DowneyF95} and thus has no kernel under the standard assumption $\mathsf{FPT} \neq \mathsf{W[2]}$.

The result of Alber et al. led to a flurry of linear kernels on planar graphs, for example, for \textsc{Feedback Vertex Set}~\cite{DBLP:conf/iwpec/BodlaenderP08}, \textsc{Cycle Packing}~\cite{DBLP:conf/isaac/BodlaenderPT08}, \textsc{Induced Matching}~\cite{DBLP:conf/faw/MoserS07,DBLP:journals/jcss/KanjPSX11}, \textsc{Full-Degree Spanning Tree}~\cite{DBLP:conf/iwpec/GuoNW06}, and \textsc{Connected Dominating Set}~\cite{DBLP:journals/tcs/LokshtanovMS11} (see also~\cite{DBLP:journals/siamcomp/ChenFKX07}).
Guo and Niedermeier~\cite{DBLP:conf/icalp/GuoN07} gave a framework capturing some of the techniques for kernels on planar graphs, obtaining linear kernels also for \textsc{Connected Vertex Cover}, \textsc{Edge Dominating Set}, \textsc{Triangle Packing}, and \textsc{Efficient Dominating Set}.
Fomin and Thilikos generalized the planar \textsc{Dominating Set} kernel to graphs of bounded Euler genus~\cite{DBLP:conf/icalp/FominT04}.

In 2009, all of the aforementioned results were subsumed by a general meta-theorem of Bodlaender, Fomin, Lokshtanov, Penninkx, Saurabh, and Thilikos~\cite{DBLP:conf/focs/BodlaenderFLPST09,Bodlaender_Fomin_Lokshtanov_Penninkx_Saurabh_Thilikos_2016}.
The theorem states that all problems on graphs of bounded Euler genus that (\refstepcounter{thmcond}\label{intro:quasi-coverable}\thethmcond) are ``quasi-coverable'', and (\refstepcounter{thmcond}\label{intro:fii}\thethmcond) have ``finite integer index'' (FII) have a linear kernel.
The appendix of~\cite{Bodlaender_Fomin_Lokshtanov_Penninkx_Saurabh_Thilikos_2016} mentions~31 such problems.

The condition~\Cref{intro:quasi-coverable} is a technical statement that uses distances on a surface.
Its modern equivalent~\cite{Fomin_Lokshtanov_Saurabh_Thilikos_2020} formulation would be that all instances $G$ of the problem admit a treewidth-$\cTwMod$-modulator $X \subseteq V(G)$ of size $|X| \le \OO(\opt(G))$ for a constant $\cTwMod$, i.e., a set $X$ such that $\tw(G \setminus X) \le \cTwMod$.\footnote{We denote the treewidth of a graph $G$ by $\tw(G)$. For the definition of treewidth we refer the reader to \Cref{sec:def_decompositions}, and for a comprehensive introduction, see~\cite[Chapter~14]{kernelization-book}.}
The condition~\Cref{intro:fii}, i.e. FII, means roughly speaking that the problem has a dynamic programming algorithm on graphs of bounded treewidth, such that the integer values on the dynamic programming table are well-behaved.

The main technique introduced by~\cite{Bodlaender_Fomin_Lokshtanov_Penninkx_Saurabh_Thilikos_2016} is \emph{protrusion replacement}.
A \emph{$c$-protrusion} in a graph $G$ is a set $P \subseteq V(G)$ so that $\tw(G[P]) \le c$ and $|\protboundary P| \le c$, where $\protboundary P$ denotes the vertices in $P$ that have a neighbor outside of $P$.
The parameter $c$ should be seen as a constant that depends on the problem and the graph class.
Protrusion replacement means the operation of replacing a protrusion $P$ by a small gadget, while maintaining that the optimum value of the problem stays the same, up to a shifting constant that can be computed while replacing the protrusion.
Having FII implies that protrusions can be replaced by constant-size gadgets~\cite{Bodlaender_Fomin_Lokshtanov_Penninkx_Saurabh_Thilikos_2016}.

The algorithm of~\cite{Bodlaender_Fomin_Lokshtanov_Penninkx_Saurabh_Thilikos_2016} works by replacing protrusions by constant-size gadgets until the graph has no more large enough protrusions to make progress.
To argue that this results in a kernel, they use \emph{protrusion decompositions}.
A $(k,c)$-protrusion decomposition of a graph $G$ is a pair $(T,\bag)$, where $T$ is a tree rooted at a node $r$ and $\bag \colon V(T) \to 2^{V(G)}$ is a function that satisfies
\begin{enumerate}
\item for all $uv \in E(G)$ there is $t \in V(T)$ with $u,v \in \bag(t)$,
\item for all $v \in V(G)$, the set $\{t \in V(T) \mid v \in \bag(t)\}$ is non-empty and connected in $T$,
\item $|\bag(r)| \le k$ and the degree of $r$ is $\le k$, and
\item $|\bag(t)| \le c$ for all $t \in V(T) \setminus \{r\}$.
\end{enumerate}
In other words, a $(k,c)$-protrusion decomposition is a tree decomposition $(T,\bag)$, where the root-bag has size and degree $\le k$ and other bags size $\le c$.
The subtrees rooted at the children of the root form $c$-protrusions, and thus, a graph that has a $(k,c)$-protrusion decomposition, but whose $c$-protrusions have size $\OO(1)$, has size $\OO(k)$.
Bodlaender et al. showed that the condition~\Cref{intro:quasi-coverable} on graphs of bounded genus implies the existence of an $(\OO(\opt(G)), \OO(1))$-protrusion decomposition, implying that replacing protrusions results in a linear kernel.

The machinery of protrusion replacement and protrusion decompositions was later used for even more general kernelization meta-theorems~\cite{Kim_Langer_Paul_Reidl_Rossmanith_Sau_Sikdar_2015,Fomin_Lokshtanov_Saurabh_Thilikos_2020}, for other applications in kernelization~\cite{Fomin_Lokshtanov_Misra_Saurabh_2012,DBLP:conf/soda/FominLST12,Fomin_Lokshtanov_Misra_Ramanujan_Saurabh_2015,Fomin_Lokshtanov_Misra_Philip_Saurabh_2016,Dabrowski_Golovach_vantHof_Paulusma_Thilikos_2017,Gajarský_Hliněný_Obdržálek_2017,Fomin_Lokshtanov_Saurabh_Thilikos_2018,Kim_Serna_Thilikos_2019,DBLP:journals/siamdm/GiannopoulouPRT21,DBLP:journals/siamcomp/JansenW25,Lokshtanov_Ramanujan_Saurabh_Sharma_Zehavi_2025,DBLP:conf/icalp/GahlawatRZ25}, and for applications in parameterized and approximation algorithms outside of kernelization~\cite{Fomin_Lokshtanov_Misra_Saurabh_2012,Joret_Paul_Sau_Saurabh_Thomassé_2014,Kim_Langer_Paul_Reidl_Rossmanith_Sau_Sikdar_2015,Kim_Paul_Philip_2015,Chatzidimitriou_Raymond_Sau_Thilikos_2017,DBLP:journals/siamcomp/BasteST23,Golovach_Stamoulis_Thilikos_2023}.

The meta-theorem of Fomin, Lokshtanov, Saurabh, and Thilikos~\cite{Fomin_Lokshtanov_Saurabh_Thilikos_2020} states that (1) all problems that are ``$\cmso$-definable'', ``linear-separable'', and ``minor-bidimensional'' admit linear kernels on minor-free graphs, and (2) all problems that are ``$\cmso$-definable'', ``linear-separable'', and ``contraction-bidimensional'' admit linear kernels on apex-minor-free graphs.
The problems for which (1) applies include for example \textsc{Cycle Packing} and \textsc{Feedback Vertex Set}.
The problems for which (2) applies include additionally \textsc{($r$-)Dominating Set}, \textsc{Connected Vertex Cover}, and \textsc{$r$-Scattered Set}.
The meta-theorem of Kim, Langer, Paul, Reidl, Rossmanith, Sau, and Sikdar~\cite{Kim_Langer_Paul_Reidl_Rossmanith_Sau_Sikdar_2015} states that all problems that are ``linearly treewidth-bounding'' and have FII have linear kernels on topological-minor-free graphs.
These include for example \textsc{Feedback Vertex Set}, \textsc{Chordal Vertex Deletion}, and \textsc{Edge Dominating Set}.

The algorithms of both~\cite{Fomin_Lokshtanov_Saurabh_Thilikos_2020} and~\cite{Kim_Langer_Paul_Reidl_Rossmanith_Sau_Sikdar_2015} follow the same idea as the algorithm of~\cite{Bodlaender_Fomin_Lokshtanov_Penninkx_Saurabh_Thilikos_2016}: They use FII and employ protrusion replacement, and then argue via treewidth modulators that each $G$ has an $(\OO(\opt(G)), \OO(1))$-protrusion decomposition, certifying that protrusion replacement results in a linear kernel.
These algorithms, as well as other purely protrusion-replacement based kernels, can be implemented in linear time by the linear-time protrusion-replacement routine of Fomin, Lokshtanov, Misra, Ramanujan, and Saurabh~\cite{Fomin_Lokshtanov_Misra_Ramanujan_Saurabh_2015} (see also~\cite{Fomin_Lokshtanov_Misra_Saurabh_2012,kernelization-book}).

\paragraph{Our contribution.}
The main technical contribution of this paper is to provide a dynamic version of the protrusion replacement technique, resulting in dynamic versions of protrusion-based kernels.
While there are some previous works on dynamic kernelization~\cite{Iwata_Oka_2014,DBLP:conf/soda/ChitnisCHM15,Alman_Mnich_Williams_2020,Bannach_Heinrich_Reischuk_Tantau_2022,DBLP:conf/isaac/AnCJJL0SSS24}, there are no previous dynamic algorithms based on protrusions nor dynamic versions of the above discussed linear kernels for sparse graph classes.

We provide a dynamic algorithm that maintains a protrusion decomposition of a topological-minor-free graph, so that the first parameter of the decomposition is linear in the minimum size of a treewidth modulator.
We use $\OO_{\bar{p}}(\cdot)$ to denote the $\OO$-notation that ignores factors that depend on a tuple of parameters $\bar{p}$ and are computable given $\bar{p}$.
We denote by $\optTwMod(G)$ the size of a smallest set $X \subseteq V(G)$ so that $\tw(G \setminus X) \le \cTwMod$.
We state here a version of our main theorem omitting certain technical details.
The version with the details is stated as \Cref{thm:technical_main}.

\begin{restatable}{theorem}{maintheorem}\label{theo:main}
There is a data structure that is initialized with a graph $H$, integer $\cTwMod$, and an $H$-topological-minor-free graph $G$.
It supports updating $G$ via insertions and deletions of edges and isolated vertices, under the promise that $G$ remains $H$-topological-minor-free.
It maintains an $(\OO_{H,\cTwMod}(\optTwMod(G)), \OO_{H,\cTwMod}(1))$-protrusion decomposition $\Tc$ of $G$.
The amortized running time of the initialization is $\OO_{H,\cTwMod}(|G| \log |G|)$, and the amortized update time is $\OO_{H,\cTwMod}(\log |G|)$.

Furthermore, the data structure provides infrastructure for maintaining dynamic programming procedures on the subtrees of $\Tc$ rooted at the children of the root.
Each update to $G$ causes updates to only $\OO_{H,\cTwMod}(1)$ such subtrees, and changes the root-bag only by at most $\OO_{H,\cTwMod}(1)$.
\end{restatable}

The second paragraph of the statement of \Cref{theo:main} is an informal explanation of the technical details that are needed for the application to dynamic kernelization.

The parameters of the protrusion decomposition of \Cref{theo:main} are optimal by $\OO_{H,\cTwMod}(1)$-factors, since having a $(k,\cTwMod)$-protrusion decomposition implies that $\optTwMod(G) \le k$.
Furthermore, topological-minor-free graph classes are the most general subgraph-closed graph classes where a linear relation between treewidth modulators and protrusion decompositions holds (see \Cref{thm:topoltight}), so the restriction to topological-minor-free graphs is justified.

To the best of our knowledge, the initialization algorithm\footnote{The initialization algorithm simply inserts vertices and edges one-by-one into the data structure.} of \Cref{theo:main} is the first algorithm to explicitly compute a protrusion decomposition with approximately optimal parameters in near-linear time, without a given treewidth-modulator.
Previously, an $\OO_{H,\cTwMod}(|G|^2)$ time algorithm was given by Kim, Serna, and Thilikos~\cite{Kim_Serna_Thilikos_2019} (in the same setting of $H$-topological-minor-free graphs and $\OO_{H,\cTwMod}(1)$-approximation).
However, most of the previous kernels based on protrusion replacement do not use the decomposition explicitly, but instead employ iterative protrusion replacement, which can be implemented in linear time~\cite{Fomin_Lokshtanov_Misra_Ramanujan_Saurabh_2015}.

The data structure of \Cref{theo:main} can be seen as a generalization of the recent dynamic treewidth data structure of Korhonen~\cite{Korhonen_2025}.
Applying it to graphs of treewidth $\le \cTwMod$ (which exclude $K_{\cTwMod+2}$ as a topological-minor and have $\optTwMod(G) = 0$) recovers the result of~\cite{Korhonen_2025}, but with a significantly higher dependency on $\cTwMod$ in both the running time and the width.
This implies that the factor $\log |G|$ in the amortized update time is optimal, as it is required (unconditionally) even for maintaining dynamic forests~\cite{DBLP:journals/siamcomp/PatrascuD06}.

\paragraph{Dynamic meta-kernelization.}
We then present our kernelization meta-theorem, which provides dynamic kernelization algorithms similar to \Cref{thm:introdomset} for a large class of problems.
For that, we need a couple of definitions.

By $\cmso$ we mean the Counting Monadic Second-Order logic on graphs (see e.g.~\cite[Section~14.5]{kernelization-book}).
We say that a graph class $\mc G$ is $\cmso$-definable if there is a $\cmso$-sentence $\Phi$ so that $G \models \Phi$ if and only if $G \in \mc G$.
Most of the natural graph classes are $\cmso$-definable.
A graph class $\mc G$ excludes a topological minor if there exists a graph $H$ so that no graph in $\mc G$ contains $H$ as a topological minor.
Classes excluding a topological minor include classes excluding a minor, such as the planar graphs, but also the graphs of bounded degree.

We consider parameterized graph problems $\Pi$ that are either minimization or maximization problems, and denote by $\opt_{\Pi}(G)$ the smallest (resp. largest) value $k$ so that $(G,k)$ is a yes-instance.\footnote{For the purpose of stating \Cref{thm:kernelization}, let us assume that such $k$ exists. A more precise form of \Cref{thm:kernelization} is given as the pair of \Cref{lem:kernelization_tw_mod,lem:kernelization_tw_mod:v2}.}
A problem $\Pi$ is \emph{linearly treewidth-bounding} on a graph class $\mc G$ if there is a constant $\cTwMod$ so that for all $G \in \mc G$ we have $\optTwMod(G) \le \OO(\opt_{\Pi}(G))$.
For example, for \textsc{Dominating Set} on planar graphs we can take $\cTwMod = 2$~\cite{kernelization-book}.
The statement of \Cref{thm:kernelization} uses also the definition of FII, which we already introduced informally, and which will be defined in \Cref{sec:def_fpt}.

\begin{restatable}{theorem}{kernelization}\label{thm:kernelization}
Let $\mc G$ be a $\cmso$-definable graph class that excludes a topological minor, and $\Pi$ a parameterized graph problem that is linearly treewidth-bounding on $\mc G$ and has FII.

There is a data structure that is initialized with a graph $G \in \mc G$ in $\OO(|G| \log |G|)$ time, supports updating $G$ via insertions and deletions of edges and isolated vertices in amortized $\OO(\log |G|)$ time per update under the promise that $G$ remains in $\mc G$, and throughout maintains a graph $K$ and a non-negative integer $\Delta$ so that
\begin{enumerate}
\item $K \in \mc G$,
\item $\opt_{\Pi}(K) + \Delta = \opt_{\Pi}(G)$, and
\item $|K| \le \OO(\opt_{\Pi}(G))$.
\end{enumerate}
The data structure outputs $(K,\Delta)$ at the initialization, and after each update it outputs at most $\OO(1)$ updates to $(K,\Delta)$, which consist of insertions and deletions of edges and isolated vertices to/from $K$, and (unbounded) changes of $\Delta$.
\end{restatable}

Now, for example \Cref{thm:introdomset} follows by taking $\mc G$ to be the class of planar graphs, $\Pi$ the \textsc{Dominating Set} problem, and using the known fact that \textsc{Dominating Set} is linearly treewidth-bounding on planar graphs~\cite{Fomin_Lokshtanov_Saurabh_Thilikos_2020}.
More generally, we could replace $\mc G$ by any apex-minor-free $\cmso$-definable graph class~\cite{Fomin_Lokshtanov_Saurabh_Thilikos_2020}, for example, the graphs of Euler genus at most $10$.

For most of the natural problems such as \textsc{Dominating Set}, we can get rid of the shifting-constant $\Delta$ by encoding it in $K$ via gadgets.
In particular, from the fact that $K$ changes by at most $\OO(1)$ updates per update to $G$, it follows that the value of $\Delta$ also changes by at most $\OO(1)$, so we can within the same bounds maintain a set of $\Delta$ additional isolated vertices in $K$.

The algorithm of \Cref{thm:kernelization} applies to all problems to which the meta-theorems of Fomin et al.~\cite[Theorem~1.1]{Fomin_Lokshtanov_Saurabh_Thilikos_2020} and Kim et al.~\cite[Theorem~1]{Kim_Langer_Paul_Reidl_Rossmanith_Sau_Sikdar_2015} apply.
This also encompasses the results of Bodlaender et al. about linear kernels~\cite[Theorem~1.3]{Bodlaender_Fomin_Lokshtanov_Penninkx_Saurabh_Thilikos_2016}.
We give a list of concrete problems to which \Cref{thm:kernelization} applies in \Cref{sec:concl}.

The amortized update time $\OO(\log |G|)$ in \Cref{thm:kernelization} is optimal for some problems captured by it.
In particular, the result of P{\u{a}}tra{\c{s}}cu and Demaine~\cite{DBLP:journals/siamcomp/PatrascuD06} implies that the problem of maintaining whether a planar graph is a forest or contains a cycle requires $\OO(\log |G|)$ (amortized and randomized) update time, unconditionally.
This implies that $\OO(\log |G|)$ update time is required for any problem $\Pi$ on planar graphs whose optimum value is a constant $c$ if and only if $G$ is a forest, for example, \textsc{Cycle Packing} or \textsc{Feedback Vertex Set}.

The statement of \Cref{thm:kernelization} is non-constructive, in particular, it is not clear how the problem $\Pi$ would even be described.
The proof is also inherently non-constructive because of the non-constructive nature of the definition of FII and the related protrusion-replacement machinery.
However, it can be made constructive for large classes of concrete problems by using the techniques introduced by Garnero, Paul, Sau, and Thilikos~\cite{Garnero_Paul_Sau_Thilikos_2015,Garnero_Paul_Sau_Thilikos_2019}.

\paragraph{Applications outside of kernelization.}
\Cref{thm:kernelization} has some direct applications outside of kernelization.
As the first example, we observe that it improves the update times of many dynamic parameterized algorithms from $2^{\OO(\sqrt{k})} \log n$ to $2^{\OO(\sqrt{k})} + \OO(\log n)$.
In particular, Korhonen~\cite{Korhonen_2025} observed that his dynamic treewidth data structure gives a dynamic algorithm for maintaining whether an $n$-vertex planar graph contains a dominating set of size $\le k$ in $2^{\OO(\sqrt{k})} \log n$ amortized update time.
Plugging in \Cref{thm:kernelization} (in fact, \Cref{thm:introdomset}), directly improves this to an amortized update time of $2^{\OO(\sqrt{k})} + \OO(\log n)$.
Similar effect happens to a large class of parameterized problems on (topological-)minor-free graph classes.

Another direct application is to dynamic approximation algorithms.
We observe that if a problem $\Pi$ satisfies that $\opt_{\Pi}(G) \le \OO(|G|)$ for all $G$, then the data structure of \Cref{thm:kernelization} directly maintains a constant-factor approximation of $\opt_{\Pi}(G)$, because of the inequalities $\opt_{\Pi}(G) \le \OO(|K|)+\Delta \le \OO(\opt_{\Pi}(G))$.
To the best of our knowledge, this is the first dynamic constant-factor approximation algorithm for many problems captured by \Cref{thm:kernelization}, for example, for \textsc{Feedback Vertex Set} on planar graphs.

Previously, Korhonen, Nadara, Pilipczuk, and Sokołowski~\cite{Korhonen_Nadara_Pilipczuk_Sokołowski_2024} gave dynamic $(1+\varepsilon)$-approximation algorithms for \textsc{Weighted Independent Set} on apex-minor-free graphs and \textsc{Weighted Dominating Set} on bounded-degree minor-free graphs.
They achieve $f(\varepsilon) \cdot n^{o(1)}$ amortized update time by implementing a dynamic version of the Baker's scheme~\cite{DBLP:journals/jacm/Baker94}.
Dynamic approximation algorithms are well-studied for \textsc{Matching} and \textsc{Vertex Cover}, see e.g.~\cite{DBLP:conf/focs/Solomon16,Bhattacharya_Henzinger_Italiano_2018,DBLP:journals/jacm/BhattacharyaKSW24}.

\paragraph{Our techniques.}
Our main technical contribution is \Cref{theo:main}, which implies \Cref{thm:kernelization} via a dynamic implementation of the known protrusion replacement machinery.
The basic blueprint behind the data structure of \Cref{theo:main} is inspired by the dynamic treewidth data structure of Korhonen~\cite{Korhonen_2025}.
In particular, the subtrees of bounded treewidth are maintained by the same routine as in~\cite{Korhonen_2025}, and the whole protrusion decomposition satisfies the same key invariant of ``downwards well-linkedness'' as the tree decomposition maintained by~\cite{Korhonen_2025}.
The difference is that now the root-bag of the decomposition can have large size, and in particular, its size should be maintained to be approximately the same as the minimum treewidth-$\cTwMod$-modulator.

It is easy to incorporate edge insertions/deletions to this structure so that they increase the size of the root-bag by a constant in each operation.
Therefore, the main challenges are to:
\begin{enumerate}
\item show that if the root-bag grows too large compared to the optimum treewidth-$\cTwMod$-modulator, then it can be reduced by chopping off a small part to the bounded-treewidth subtrees, and\label{intro:ov:enum11}
\item implement this chopping within $\OO(\log n)$ amortized update time and $\OO(1)$ changes to the root per update.\label{intro:ov:enum12}
\end{enumerate}
Both of the parts (\ref{intro:ov:enum11}) and (\ref{intro:ov:enum12}) are non-trivial.
The part (\ref{intro:ov:enum11}) yields an essentially new type of an algorithm for constructing a protrusion decomposition, which constructs a protrusion decomposition by iteratively making the root-bag smaller, unlike the previous algorithms, which proceed by starting with a treewidth-$\cTwMod$-modulator and then making the root-bag a superset of it.
The proof of (\ref{intro:ov:enum11}) uses graph-theoretical techniques, requiring topological-minor-freeness and heavily relying on the ``downwards well-linkedness'' of the decomposition.
The part (\ref{intro:ov:enum12}) requires a careful dynamic implementation of a local search procedure for finding large sets of vertices with small neighborhoods.

\paragraph{Related work on dynamic kernelization.}
So far we omitted the discussion on previous dynamic kernelization algorithms, so let us review them here.
Let us focus only on dynamic polynomial kernels, and omit the larger body of work on dynamic FPT algorithms (see e.g.~\cite{DBLP:conf/wg/Bodlaender93a,DBLP:conf/wads/DvorakT13,Dvořák_Kupec_Tůma_2014,Alman_Mnich_Williams_2020,DBLP:conf/soda/ChenCDFHNPPSWZ21,Olkowski_Pilipczuk_Rychlicki_Węgrzycki_Zych-Pawlewicz_2023,Korhonen_Majewski_Nadara_Pilipczuk_Sokołowski_2023,Majewski_Pilipczuk_Zych-Pawlewicz_2024,Korhonen_2025,DBLP:journals/toct/MajewskiPS25}).

Iwata and Oka~\cite{Iwata_Oka_2014} gave dynamic kernels for \textsc{Vertex Cover} and \textsc{Cluster Vertex Deletion}.
The kernel for \textsc{Vertex Cover} has $\OO(k^2)$ update time and size $\OO(k^2)$, while the kernel for \textsc{Cluster Vertex Deletion} has $\OO(k^8 \log n)$ update time and size $\OO(k^5)$.

Alman, Mnich, and Vassilevska Williams~\cite{Alman_Mnich_Williams_2020} improved the update time of the dynamic vertex cover kernelization to $\OO(k)$ worst-case and $\OO(1)$ amortized.
They also gave dynamic polynomial kernels for \textsc{$d$-Hitting Set}, \textsc{Edge Dominating Set}, and \textsc{Point Line Cover}.
The kernel for \textsc{Edge Dominating Set} has update time $\OO(1)$ and size $\OO(k^2)$, while the kernels for \textsc{$d$-Hitting Set} and \textsc{Point Line Cover} have update times and sizes $\poly(k)$.

An improved dynamic kernel for \textsc{$d$-Hitting Set}, along with a dynamic kernel for \textsc{Set Packing}, was given by Bannach, Heinrich, Reischuk, and Tantau~\cite{Bannach_Heinrich_Reischuk_Tantau_2022}.
An, Cho, Jang, Jung, Lee, Oh, Shin, Shin, and Song~\cite{DBLP:conf/isaac/AnCJJL0SSS24} designed dynamic kernels on unit disk graphs for \textsc{Vertex Cover}, \textsc{Triangle Hitting Set}, \textsc{Feedback Vertex Set}, and \textsc{Cycle Packing}.

\paragraph{Organization of the paper.}
The rest of the paper is organized as follows.
We start by presenting informal sketches of the proofs of \Cref{theo:main} and \Cref{thm:kernelization} in \Cref{sec:overview}.
We present the proofs in detail in \Cref{sec:preliminaries,sec:existence,sec:local_search,sec:splay,sec:main_data_structure,sec:kernelization}.
In particular, we start by presenting definitions and preliminaries in \Cref{sec:preliminaries}.
Then, the proof of the part (\ref{intro:ov:enum11}) discussed above is presented in \Cref{sec:existence}, and the main part of the proof of (\ref{intro:ov:enum12}) in \Cref{sec:local_search}.
\Cref{sec:splay} is dedicated to lifting the dynamic treewidth data structure of~\cite{Korhonen_2025} to maintain the bounded-treewidth parts of the protrusion decomposition.
Then, in \Cref{sec:main_data_structure} we combine the material of \Cref{sec:existence,sec:local_search,sec:splay} to finish the proof of \Cref{theo:main}.
We prove \Cref{thm:kernelization} in \Cref{sec:kernelization}.
We conclude with additional remarks and discussion in \Cref{sec:concl}.

%% file: overview.tex
\section{Overview}
\label{sec:overview}
We first sketch a proof of \Cref{theo:main}, i.e., explain how we maintain an approximately optimal protrusion decomposition of a topological-minor-free graph.
Afterwards, we describe how this data structure can be used for dynamic kernelization, sketching the proof of \Cref{thm:kernelization}.

\subsection{Dynamic protrusion decomposition}
We consider a dynamic $H$-topological-minor-free graph $G$ and a parameter $\cTwMod$.
Our goal is to maintain, under insertions and deletions of edges and isolated vertices, a $(\OO_{H,\cTwMod}(\optTwMod(G)),\OO_{H,\cTwMod}(1))$-protrusion decomposition of $G$.
This is a rooted tree decomposition $(T,\bag)$, where the root-bag $r \in V(T)$ has size and degree $\OO_{H,\cTwMod}(\optTwMod(G))$, and other bags have size $\OO_{H,\cTwMod}(1)$.

We start by describing in \Cref{subsubsec:overview:basic} the definitions and ideas that are based on those of~\cite{Korhonen_2025} (many of which in turn originate from~\cite{Korhonen_2024}, and further from~\cite{Robertson_Seymour_1991}).
The more novel parts of our algorithm are described in \Cref{subsubsec:overview:novel}.

\subsubsection{The framework and basic operations}
\label{subsubsec:overview:basic}

\paragraph{Downwards well-linked superbranch decompositions.}
The key idea of the dynamic treewidth algorithm of~\cite{Korhonen_2025} is to maintain a structure called \emph{downwards well-linked superbranch decomposition}.
We do the same in this paper, although naturally with different constraints to reflect that the superbranch decomposition should correspond to a protrusion decomposition instead of a tree decomposition of bounded width.

A \emph{superbranch decomposition} of a graph $G$ is a pair $\Tc = (T,\Lc)$, where $T$ is a rooted tree, in which every non-leaf node has at least two children, and $\Lc$ is a bijection that maps every leaf of $T$ to an edge of $G$.
This is similar to the classic definition of a branch decomposition~\cite{Robertson_Seymour_1991}, but allowing nodes of degree higher than three.
For a node $t \in V(T)$, we denote by $\Lc[t] \subseteq E(G)$ the set of edges of $G$ that are associated by $\Lc$ with leaves in the subtree of $T$ below $t$.

The \emph{boundary} of an edge set $A \subseteq E(G)$, denoted by $\bd(A) \subseteq V(G)$, is the set of vertices that are incident to edges in both $A$ and $E(G) \setminus A$.
We denote the size of the boundary by $\lambda(A) = |\bd(A)|$.
The function $\lambda \colon 2^{E(G)} \to \mathbb{Z}_{\ge 0}$ is \emph{symmetric and submodular}, meaning that (1) $\lambda(A) = \lambda(E(G) \setminus A)$ and (2) $\lambda(A \cup B) + \lambda(A \cap B) \le \lambda(A)+\lambda(B)$ for all $A,B \subseteq E(G)$~\cite{Robertson_Seymour_1991}.

A set $A \subseteq E(G)$ of edges is \emph{well-linked} if for every bipartition $(A_1,A_2)$ of $A$ it holds that $\lambda(A_1) \geq \lambda(A)$ or $\lambda(A_2) \geq \lambda(A)$ (this is called \emph{robust} in~\cite{Robertson_Seymour_1991}).
The \emph{well-linked number} $\wl(A)$ of a set $A \subseteq E(G)$ is the maximum of $\lambda(A')$ over well-linked subsets $A' \subseteq A$.
A superbranch decomposition is \emph{downwards well-linked} if for every node $t \in V(T)$, the set $\Lc[t]$ is well-linked.

The intuitive reason why well-linkedness is a powerful notion in our context is the following three properties~\cite{Korhonen_2024,Korhonen_2025}:
\begin{enumerate}
\item If $A \subseteq E(G)$ is a well-linked set and $B \subseteq E(G)$, then either $\lambda(B \cup A) \le \lambda(B)$ or $\lambda(B \setminus A) \le \lambda(B)$. In particular, well-linked sets are perfectly ``uncrossable''.
\item Every set $A \subseteq E(G)$ can be partitioned into at most $2^{\lambda(A)}$ well-linked subsets.
\item For all $A \subseteq E(G)$, it holds that $\wl(A) = \Theta(\tw(G[A]))$.
\end{enumerate}
There is also a fourth powerful property, called the ``transitivity'' of well-linkedness in~\cite{Korhonen_2024,Korhonen_2025}, but which requires more complex definitions to state, so we omit it for now.

\paragraph{From superbranch to protrusion decompositions.}
There is a natural way to relate superbranch decompositions to tree decompositions, for which we use the following definition of an \emph{adhesion}.
For an edge $tp \in E(T)$ of a superbranch decomposition between a node $t$ and its parent $p$, the adhesion at $tp$ is the set $\adh(tp) = \bd(\Lc[t])$.
Now, if $(T,\Lc)$ is a superbranch decomposition of a graph $G$, we construct a function $\bag \colon V(T) \to 2^{V(G)}$ so that $(T,\bag)$ is a tree decomposition of $G$ as follows.
For a leaf-node $\ell$ with $\Lc(\ell) = uv$, we set $\bag(\ell) = \{u,v\}$, and for a non-leaf-node $t$, we set $\bag(t) = \bigcup_{s \in N(t)} \adh(st)$, where $N(t)$ denotes the neighbors of~$t$.
(This does not quite work if $G$ contains isolated vertices, but let us ignore that for now.)
We observe that $(T,\bag)$ is an $(\OO_{H,\cTwMod}(\optTwMod(G)),\OO_{H,\cTwMod}(1))$-protrusion decomposition of $G$ if
\begin{enumerate}[label=\alph*., ref=\alph*]
\item every adhesion has size $\OO_{H,\cTwMod}(1)$,\label{overview:cond1:item1}
\item the root-node has degree $\OO_{H,\cTwMod}(\optTwMod(G))$, and\label{overview:cond1:item2}
\item every non-root-node has degree $\OO_{H,\cTwMod}(1)$.\label{overview:cond1:item3}
\end{enumerate}

This follows from the fact that the size of $\bag(t)$ is bounded by the product of the degree of $t$ and the maximum adhesion size.
Instead of maintaining the conditions of \Cref{overview:cond1:item1,overview:cond1:item2,overview:cond1:item3} directly, our goal is to maintain a superbranch decomposition $\Tc = (T,\Lc)$ satisfying that

\begin{enumerate}
\item $\Tc$ is downwards well-linked,\label{overview:cond2:item1}
\item for all non-root $t \in V(T)$, $\wl(\Lc[t]) \le \OO_{H,\cTwMod}(1)$,\label{overview:cond2:item2}
\item the root-node has degree $\OO_{H,\cTwMod}(\optTwMod(G))$,\label{overview:cond2:item3}
\item every non-root-node has degree $\OO_{H,\cTwMod}(1)$, and\label{overview:cond2:item4}
\item $T$ has depth $\OO_{H,\cTwMod}(\log |G|)$.\label{overview:cond2:item5}
\end{enumerate}

The conditions of \Cref{overview:cond2:item1,overview:cond2:item2} imply the condition of \Cref{overview:cond1:item1} by the definition of well-linked number and downwards well-linkedness.
Thus, any protrusion decomposition satisfying \Cref{overview:cond2:item1,overview:cond2:item2,overview:cond2:item3,overview:cond2:item4} also satisfies \Cref{overview:cond1:item1,overview:cond1:item2,overview:cond1:item3}, and thus corresponds to a $(\OO_{H,\cTwMod}(\optTwMod(G)),\OO_{H,\cTwMod}(1))$-protrusion decomposition.
The logarithmic-depth requirement of \Cref{overview:cond2:item5} is for efficient dynamic maintenance of the superbranch decomposition.

\paragraph{Basic maintenance of the superbranch decomposition.}
The idea of our data structure is that we first implement the procedures for inserting and deleting edges in an ``easy'' way that increases the degree of the root by $\OO_{H,\cTwMod}(1)$ per operation but maintains the other invariants.
Then, the hard part of our data structure is a procedure that decreases the degree of the root whenever it is too large compared to $\optTwMod(G)$.
This procedure will be run after every update to keep the degree of the root controlled.

We postpone the hard part to \Cref{subsubsec:overview:novel}, and start here with the ``easy'' procedure for inserting and deleting edges.
At this point, we need to reveal the technical detail that the superbranch decomposition is not actually a superbranch decomposition of $G$, but of the hypergraph $\Hc(G)$ that has vertex set $V(\Hc(G)) = V(G)$, and in addition to the normal edges $\{u,v\}$ for all $uv \in E(G)$, contains singleton edges $\{v\}$ for all $v \in V(G)$.
Our definitions extend naturally to hypergraphs, and later we will also use hypergraphs with edges of size more than $2$.

We use the tree-rotation techniques from~\cite{Korhonen_2025} for maintaining the subtrees below the root.
A key subroutine that we implement with those techniques is a procedure that ``rotates up'' a specified set of leaves of the decomposition.
In particular, given a set $A$ of hyperedges with $|A| = \OO(1)$, the procedure uses tree rotations to transform $(T,\Lc)$ so that 
\begin{enumerate}
\item the leaves corresponding to $A$ become children of the root,
\item the degree of the root increases by $\OO_{H,\cTwMod}(1)$, and
\item all other invariants are maintained.
\end{enumerate}
This subroutine runs in $\OO_{H,\cTwMod}(\log |G|)$ amortized time.

Now, to insert an edge between vertices $u$ and $v$, we use the subroutine with $A = \{\{u\}, \{v\}\}$, and observe that after $\{u\}$ and $\{v\}$ are children of the root, inserting the hyperedge $\{u,v\}$ while maintaining the invariants (but increasing the root-degree by $1$) is trivial, by simply adding a leaf corresponding to it as another child of the root.
Similarly, to delete an edge $uv$, we use the procedure with $A = \{\{u,v\}, \{u\}, \{v\}\}$, after which the deletion becomes similarly straightforward.
In both cases, the degree of the root increases by $\OO_{H,\cTwMod}(1)$ because of the rotating-up subroutine.

The reason why it is essential that also $\{u\}$ and $\{v\}$ are children of the root while inserting/deleting $\{u,v\}$ is that this ensures that the boundaries $\bd(\Lc[t])$ stay unchanged in the decomposition.
Otherwise, the boundaries $\bd(\Lc[t])$ could abruptly change globally throughout the decomposition, potentially ruining downwards well-linkedness.

\subsubsection{Controlling the root degree}
\label{subsubsec:overview:novel}
Let us then turn to the hard part of our data structure, namely, the procedure for controlling the degree of the root.
We will decrease the degree of the root by finding a set of at least $2$ but at most $\OO_{H,\cTwMod}(1)$ subtrees rooted at children of the root, and combining them into one subtree rooted at a new child of the root.
This combination procedure will be done in a straightforward manner, in particular, by simply adding a new child and moving the subtrees to be rooted under that child instead of the root.

Denote the children of the root corresponding to such subtrees by $C = \{c_1, c_2, \ldots, c_h\}$, and denote by $\Lc[C] = \bigcup_{c_i \in C} \Lc[c_i]$ the set of hyperedges corresponding to leaves below $C$.
We need that 
\begin{enumerate}
\item $|C| \ge 2$ so that the degree of the root actually decreases,
\item $|C| \le \OO_{H,\cTwMod}(1)$ so that the degree of the new node is bounded,
\item $\Lc[C]$ is well-linked so that the decomposition stays downwards well-linked, and
\item $\wl(\Lc[C]) \le \OO_{H,\cTwMod}(1)$ so that the well-linked number of the subtrees stays bounded.
\end{enumerate}

Furthermore, the bounds hidden by $\OO_{H,\cTwMod}(\cdot)$ above should not depend on the current parameters of the decomposition, but only on the original parameters $H$ and $\cTwMod$, in order to not make the parameters of the decomposition gradually worse throughout updates.
Decreasing the degree by combining such a set of children $C$ maintains all other invariants of the decomposition except may increase the depth by one because of the new subtree.
However, the new subtree can be balanced in amortized $\OO_{H,\cTwMod}(\log |G|)$ time with the techniques from~\cite{Korhonen_2025}.

It follows that now, the two main challenges are to show that
\begin{enumerate}
\item whenever the degree of the root is too large compared to $\optTwMod(G)$, such a set of children $C$ exists, and
\item in that case, we can also find $C$ efficiently.
\end{enumerate}

For both of these two challenges, the key definition we need is that of the \emph{torso hypergraph}.
Let $t$ be a node of a superbranch decomposition.
The torso of $t$, denoted by $\torso(t)$, is the hypergraph with the vertex set $V(\torso(t)) = \bigcup_{s \in N(t)} \adh(st)$, and having the hyperedge $\adh(st)$ for every $s \in N(t)$.
Note that the same hyperedge can occur multiple times if $\adh(s_1 t) = \adh(s_2 t)$ for $s_1 \neq s_2$, and we indeed treat $\torso(t)$ as a ``multi''-hypergraph, thinking about the hyperedge set of $\torso(t)$ as consisting of labels $e_s$ for all $s \in N(t)$, together with a mapping that maps each label $e_s$ to the set $\adh(st)$.

Let $r$ be the root node and $A \subseteq E(\torso(r))$ a subset of hyperedges of its torso, which naturally corresponds to a set of children $C_A$ of the root.
We denote by $A \triangleright \Tc = \Lc[C_A] \subseteq E(\Hc(G))$ the set of hyperedges corresponding to the leaves under $C_A$.
The ``transitivity of well-linkedness'' from~\cite{Korhonen_2024,Korhonen_2025} tells that $A \triangleright \Tc$ is well-linked in $\Hc(G)$ if and only if $A$ is well-linked in $\torso(r)$.
Now, the task of finding the required set of children $C$ translates into finding a set $A \subseteq E(\torso(r))$ so that
\begin{enumerate}
\item $2 \le |A| \le \OO_{H,\cTwMod}(1)$,\label{overview:cond6:item1}
\item $A$ is well-linked in $\torso(r)$, and\label{overview:cond6:item2}
\item $\wl(A \triangleright \Tc) \le \OO_{H,\cTwMod}(1)$.\label{overview:cond6:item3}
\end{enumerate}

We can in fact simplify these three conditions even further.
The combination of \Cref{overview:cond6:item2,overview:cond6:item3} implies that $\lambda(A)$ should be bounded by $\OO_{H,\cTwMod}(1)$.
By using the fact that any set $A \subseteq E(\torso(r))$ can be partitioned into at most $2^{\lambda(A)}$ well-linked sets, we can relax the condition that $A$ is well-linked into a condition that $A$ is large enough compared to $2^{\lambda(A)}$ -- then a postprocessing routine can find a well-linked subset of $A$.
We end up with the conditions

\begin{enumerate}
\item $2^{\lambda(A)} < |A| \le \OO_{H,\cTwMod}(1)$ and
\item $\wl(A \triangleright \Tc) \le \OO_{H,\cTwMod}(1)$.
\end{enumerate}

\paragraph{Existence of a mergeable subset.}
Let us then get to the proof that such a set $A \subseteq E(\torso(r))$ exists if the root degree is too large.
Note that the root degree equals $|E(\torso(r))|$.

The first ingredient is a proof that $|E(\torso(r))| \le \OO_{H,\cTwMod}(|V(\torso(r))|)$.
This uses the $H$-topological-minor-freeness of $G$ and the downwards well-linkedness, in particular, if $\torso(r)$ would be too dense, we could realize this density as a topological minor via downwards well-linkedness.
Another case is that $\torso(r)$ contains a lot of hyperedges with exactly the same vertex set, but in that case the desired conclusion is easy to achieve by just selecting a subset of them.
Therefore, we can assume that $V(\torso(r))$ is large compared to $\optTwMod(G)$.

Now, we consider an imaginary ``optimal'' protrusion decomposition of $G$.
By~\cite{Kim_Langer_Paul_Reidl_Rossmanith_Sau_Sikdar_2015}, $G$ indeed has an $(\OO_{H,\cTwMod}(\optTwMod(G)), \OO_{H,\cTwMod}(1))$-protrusion decomposition $\Tc^* = (T^*, \bag^*)$, rooted at a node $r^*$.
If $V(\torso(r))$ is larger than $\Omega_{H,\cTwMod}(|\bag^*(r^*)| + \Delta(r^*)) = \Omega_{H,\cTwMod}(\optTwMod(G))$\footnote{Where $\Delta(r^*)$ denotes the degree of $r^*$.}, then there exists a child $c^*$ of $r^*$, so that the subtree of $\Tc^*$ rooted at $c^*$ contains at least $\delta$ vertices from $V(\torso(r))$, and has treewidth $\omega$, where $\delta$ and $\omega$ are parameters in $\OO_{H,\cTwMod}(1)$ with $\omega << \delta$.
By choosing an appropriate subtree under $c^*$, we indeed obtain a set $B \subseteq E(\Hc(G))$ so that

\begin{enumerate}
\item $\delta \le |V(B) \cap V(\torso(r))| \le 3 \cdot \delta$,\footnote{Where $V(B) \subseteq V(G)$ denotes all vertices incident to $B$.}
\item $\lambda(B) \le \omega$, and
\item $\wl(B) \le \omega$,
\end{enumerate}
where $\delta$ is set to be roughly $2^{\Theta(\omega)}$.

The set $B$ would be perfect for us if it would be ``uncrossed'' with the sets $\Lc[c]$ for the children $c$ of the root $r$ of our superbranch decomposition, that is, if either $\Lc[c] \subseteq B$ or $\Lc[c] \cap B = \emptyset$ would hold for all $c$.
This would allow us to construct the desired set $A \subseteq E(\torso(r))$ by taking $e_c \in A$ whenever $\Lc[c] \subseteq B$.

Now the natural idea is to use the downwards well-linkedness to uncross the set $B$ with the sets $\Lc[c]$.
It implies that either $\lambda(B \setminus \Lc[c]) \le \lambda(B)$ or $\lambda(B \cup \Lc[c]) \le \lambda(B)$, so one can always either include or exclude $\Lc[c]$ without increasing $\lambda(B)$.
Furthermore, we show that doing this for all children $c$ does not drastically affect the quantity $|V(B) \cap V(\torso(r))|$, as long as it is large enough compared to $\lambda(B)$.

However, this uncrossing idea has one issue that makes the proof much more complicated: The quantity $\wl(B)$ can increase if we set $B \coloneqq B \cup \Lc[c]$.
On the surface this seems to not be a big issue since $\wl(\Lc[c])$ is anyway bounded by $\OO_{H,\cTwMod}(1)$, but here we run into the issue that the ``new parameters'' of the decomposition should not depend on the ``old parameters'' of the decomposition, or otherwise they would stack up throughout many operations.
In particular, this uncrossing would work fine if we could assume that $\wl(\Lc[c])$ is small enough compared to the other parameters of the decomposition.

We fix this issue by ensuring that this type of uncrossing can indeed only happen if $\wl(\Lc[c])$ is really small.
In particular, we observe that $\wl(\Lc[c])$ can be larger than $\OO(\cTwMod)$ only for at most $\optTwMod(G)$ many subtrees of our decomposition, as the modulator must hit all such subtrees.
By tuning the existence proof of protrusion decompositions from~\cite{Kim_Langer_Paul_Reidl_Rossmanith_Sau_Sikdar_2015}, we can ensure that the set $\bd(\Lc[c])$ is contained in $\bag^*(r^*)$ for all such subtrees, and therefore by choosing $B$ according to this tuned decomposition $(T^*, \bag^*)$, $V(B)$ can intersect such $\bd(\Lc[c])$ only in $\bd(B)$.
With this, we can ensure that the uncrossing of type $B \coloneqq B \cup \Lc[c]$ never happens for such children $c$.

Thus, skipping many tedious technical details here, the uncrossing idea works out in the end to prove the existence of the desired set $A \subseteq E(\torso(r))$.

\paragraph{Finding a mergeable subset.}
The next challenge is to find a set $A \subseteq E(\torso(r))$ with (1) $2^{\lambda(A)} < |A| \le \OO_{H,\cTwMod}(1)$ and (2) $\wl(A \triangleright \Tc) \le \OO_{H,\cTwMod}(1)$ whenever such a set exists, in $\OO_{H,\cTwMod}(\log |G|)$ time.
We start by observing that the proof above gives us some slack for approximation, in particular, it is enough to be able to either return such a set, or to conclude that there is no such set with (1) replaced by (1') $2^{2 \cdot \lambda(A)} < |A| \le \OO_{H,\cTwMod}(1)$ and (2) replaced by a smaller bound on the well-linked number.
In particular, instead of bound the well-linked number of $A \triangleright \Tc$, we bound $\lambda(A)$ and the \emph{internal treewidth} of $A \expand \Tc$, that is $\tw(G[V(A \expand \Tc) \setminus \bd(A)])$.
This is approximately equivalent to bounding the well-linked number.

Let us start by discussing how to even check whether a given set $A$ satisfies these properties.
In our data structure, we maintain $\torso(r)$ explicitly, so (1) is easy enough to check in $\OO_{H,\cTwMod}(1)$ time.
However, checking (2) is not easy, as the internal treewidth involves information not captured by $\torso(r)$.
At this point, we note that thanks to the techniques from~\cite{Korhonen_2025}, we are able to maintain dynamic programming procedures on the subtrees of $\Tc$ below the root.
In particular, since such subtrees correspond to tree decompositions of bounded width, we can use powerful bounded-treewidth machinery to capture information about them.
This is used for the application of our data structure for kernelization, but we also use it here to figure out the internal treewidth of $A \triangleright \Tc$.
In particular, we maintain a modified version of the Bodlaender-Kloks dynamic programming procedure for computing treewidth~\cite{Bodlaender_Kloks_1996} on the subtrees, and when given $A$, combine the information from the subtrees corresponding to $A$ to compute the internal treewidth of $A \triangleright \Tc$.

Now we know how to efficiently check if a given set $A$ is suitable, but still, there are too many candidates to be brute-forced.
Let us sketch an $\OO_{H,\cTwMod}(|E(\torso(r))|)$ time algorithm.
A set $B \subseteq E(\torso(r))$ is \emph{internally connected} if all vertices in $V(B)$ can reach each other by paths through hyperedges in $B$ but not containing vertices from $\bd(B)$ as internal vertices.
Internally connected sets are useful because for parameters $k$ and $s$ and a hyperedge $e$, there are at most $s^k$ internally connected sets $B$ with $|B| \le s$, $\lambda(B) \le k$, and $e \in B$, and they can be listed via a local-search procedure in $s^{\OO(k)}$ time.

Now, a suitable set $A$ can be uniquely partitioned into internally connected sets $A_1, \ldots, A_h$, so that $\bd(A_i) \subseteq \bd(A)$.
By possibly shrinking $|A|$ by a factor of $2^{\lambda(A)}$, which is fine by the slack for approximation, we can assume that $\bd(A_i) = \bd(A)$ for all $A_i$.
We also observe that the internal treewidth of $A \triangleright \Tc$ is equal to the maximum internal treewidth of $A_i \triangleright \Tc$.
Now, this type of $A$ is easy enough to find in $\OO_{H,\cTwMod}(|E(\torso(r))|)$ time by first enumerating all internally connected sets with specified size and boundary size, checking their internal treewidth, and then grouping them by their boundary.
Furthermore, by using the fact that each hyperedge participates in only $\OO_{H,\cTwMod}(1)$ such internally connected sets, we can design a dynamic implementation of this algorithm, that works in $\OO_{H,\cTwMod}(\log |E(\torso(r))|)$ time per update to $\torso(r)$.

\subsection{Dynamic kernelization}
The application of our data structure to kernelization follows the high-level idea of protrusion replacement, which was pioneered by Bodlaender et al.~\cite{Bodlaender_Fomin_Lokshtanov_Penninkx_Saurabh_Thilikos_2016}, and afterwards used in dozens of kernelization algorithms.
However, our implementation of protrusion replacement is non-standard in that we explicitly compute a protrusion decomposition and replace each of its protrusions in one shot.
The typical protrusion replacement implementations work by repeatedly replacing constant-size protrusions, without actually computing a protrusion decomposition.
Nevertheless, this approach of explicitly computing a protrusion decomposition has been used before by Kim, Serna, and Thilikos~\cite{Kim_Serna_Thilikos_2019}, who needed it in the context of kernelization for counting problems.

Let $\mc G$ be a $\cmso$-definable graph class that excludes a topological minor $H$, and $\Pi$ a problem that has FII and is linearly treewidth-bounding on $\mc G$.
For example, we can let $\mc G$ be the planar graphs and $\Pi$ the \textsc{Dominating Set} problem.
Because $\Pi$ is linearly treewidth-bounding in $\mc G$, there exists $\cTwMod$ so that $\optTwMod(G) \le \OO(\opt_{\Pi}(G))$ for all $G \in \mc G$.
Now, it suffices to give a kernelization algorithm parameterized by the parameter $\optTwMod(G)$ instead of $\opt_{\Pi}(G)$.
For this, we apply the data structure of \Cref{theo:main}, initialized with the parameters $H$ and $\cTwMod$.

Let $(T,\Lc)$ be the superbranch decomposition maintained by the data structure, corresponding to a $(\OO_{H,\cTwMod}(\optTwMod(G)), \OO_{H,\cTwMod}(1))$-protrusion decomposition $(T,\bag)$.
For kernelization, it is in fact more instructive to think about $(T,\Lc)$ than about $(T,\bag)$.
For each child $c$ of the root $r$ of $T$, we consider the \emph{boundaried graph} $G_c$, having vertex set $V(G_c) = V(\Lc[c])$, i.e., the union of the hyperedges in $\Lc[c]$, and edge set $E(G_c)$ being the edges of $G$ in $\Lc[c]$, and boundary $\bd(\Lc[c])$.

Now, the idea is to replace the boundaried graph $G_c$ by a \emph{representative} of bounded size.
In particular, we replace $G_c$ by a boundaried graph $R_c$ with the same boundary as $G_c$, but which has bounded size (bounded by a function of the size of the boundary, the problem $\Pi$, and the class $\mc G$).
If we denote by $G'$ the graph obtained after this replacement, we want that $\opt_{\Pi}(G') = \opt_{\Pi}(G) - \Delta_c$, for a non-negative \emph{shifting constant} $\Delta_c$ that can be computed from $G_c$.
We also want that $G' \in \mc G$ if $G \in \mc G$.

The known properties of FII and $\cmso$ imply that such a representative $R_c$ indeed exists (see e.g.~\cite{Bodlaender_Fomin_Lokshtanov_Penninkx_Saurabh_Thilikos_2016}), and can be computed by dynamic programming on the tree decomposition of $G_c$, together with the shifting constant $\Delta_c$.
This dynamic programming can indeed be maintained on the bounded-treewidth subtrees of $(T,\Lc)$ by the dynamic treewidth data structure, so with our data structure we can maintain the representatives $R_c$ and shifting constants $\Delta_c$ for all children $c$ of the root.

Now, the kernel is obtained simply as the union of the boundaried graphs $R_c$ (which may overlap at the boundary vertices).
As these graphs have constant size, and there are at most $\OO_{H,\cTwMod}(\optTwMod(G))$ of them, the size of the kernel is $\OO_{H,\cTwMod}(\optTwMod(G))$.
The final shifting constant $\Delta$ is obtained as the sum of the individual shifting constants $\Delta_c$.

This kernel can be maintained with logarithmic update time and with a constant number of changes to it thanks to the properties guaranteed by the data structure of \Cref{theo:main}.
In particular, it guarantees that one update to the graph causes updates only in a constant number of the subtrees rooted at the children $c$ of the root, so the representatives $R_c$ and shifting constants $\Delta_c$ need to be updated only for them.
Both of them are re-computed in amortized logarithmic update time thanks to the dynamic treewidth data structure we maintain on the subtrees.

%% file: preliminaries.tex
\section{Preliminaries}
\label{sec:preliminaries}
In this section, we introduce the definitions and prove some preliminary results. We start with some basic notations.
For two integers $a$ and $b$, we denote by $[a,b]$ the set of all integers $i$ with $a \leq i \leq b$, and by $[a]$ the set $[1,a]$. For a function $f \colon X \to Y$ and a set $Z \subseteq X$, we denote by $f\restriction_Z \colon Z \to Y$ the restriction of $f$ to $Z$.
For a set $S$, we denote by $\binom{S}{2}$ the set of all subsets of $S$ of size exactly two.
The $\OO_{\bar{p}}(\cdot)$-notation, for a tuple of parameters $\bar{p}$, hides factors that depend on $\bar{p}$ and are computable given it.

\subsection{(Hyper)graphs and trees}

\paragraph{Graphs.} In this paper, all graphs are finite, undirected, and, unless otherwise stated, simple. For a graph $G$, we denote by $V(G)$ the set of vertices and by $E(G) \subseteq \binom{V(G)}{2}$ the set of edges. We also denote an edge $e = \{u,v\}$ by $uv$. For a set $A \subseteq E(G)$ of edges, we denote by $V(A) = \bigcup_{uv \in A} \{u,v\}$ the union of their endpoints, and use $V(\{e\}) \eqqcolon V(e)$.
The \emph{size} of a graph $G$ is $|G| = |V(G)|+|E(G)|$.

A \emph{boundaried graph} $(G,B,\Lambda)$ is a graph $G$ together with a set $B \subseteq V(G)$, called the \emph{boundary} of $G$, consisting of distinguished \emph{boundary vertices}, and an injective labeling $\Lambda \colon B \to \ZI$. Slightly abusing the notation, we sometimes refer to $G$ as the \emph{boundaried graph} and then denote by $\bd(G)$ the set of boundary vertices. The \emph{label set} of a boundaried graph with boundary $B$ is denoted by $\labelSet(G) = \{\Lambda(v) \colon v \in B\}$. A graph $G$ is said to be \emph{$t$-boundaried} if $\labelSet(G) \subseteq [t]$.
We denote by $\mc F$ the set of all boundaried graphs, by
$\mc F_I$ the class of all boundaried graphs with label set $I$ (for $I \subseteq \ZI$), and by $\mc F_{\subseteq I}$ the set $\bigcup_{I' \subseteq I} \mc F_{I'}$.

\paragraph{Minors and minor-free graphs.}

For a graph $G$ and an edge $e = uv \in E(G)$, we denote by $G/e$ the graph obtained from $G$ by \emph{contracting} the edge $e$, that is, by identifying $u$ and $v$ and, if necessary, removing loops and multiple edges. A graph $H$ is a \emph{contraction} of $G$ if it can be obtained from $G$ via zero or more edge contractions. A graph $H$ is a \emph{minor} of $G$ if it is a contraction of a subgraph of $G$. If each contracted edge has at least one endpoint with degree at most two, $H$ is called a \emph{topological minor}.
A graph $G$ is called \emph{$H$-(topological-)minor-free} if $H$ is not a (topological) minor of $G$.
A graph class $\mc G$ is called \emph{$H$-(topological-)minor-free} if every graph $G \in \mc G$ is $H$-(topological-)minor-free, and just \emph{(topological-)minor-free} if such $H$ exists.
In these cases can also say that $G$/$\mc G$ \emph{excludes} $H$ as a (topological) minor.

A particularly interesting subclass of minor-free graphs is the class of apex-minor-free graphs.
A graph $G$ is called an \emph{apex graph} if there exists a vertex $v \in V(G)$ such that $G-v$ is planar. A graph class $\mc G$ is called \emph{apex-minor-free} if there exists an apex graph $H$ such that $\mc G$ is $H$-minor-free.
For example, planar graphs are apex-minor-free, because they exclude $K_5$, which is an apex graph, because $K_4$ is a planar graph.
More generally, graphs of bounded Euler genus are apex-minor-free (see e.g.~\cite{DBLP:journals/algorithmica/Eppstein00}).

\paragraph{Hypergraphs.} Following the definitions of~\cite{Korhonen_2025}, a hypergraph $G$ consists of a set of vertices $V(G)$, a set of hyperedges $E(G)$, and a mapping $V \colon E(G) \to 2^{V(G)}$ that associates each hyperedge with a set of vertices.
For a set $A$ of hyperedges, we denote by $V(A) = \bigcup_{e \in A} V(e)$ the union of their vertex sets.
The \emph{size} of a hypergraph $G$ is $|G| = |V(G)| + \sum_{e \in E(G)} (|V(e)|+1)$.
The \emph{rank} of a hyperedge $e$ is $|V(e)|$, and the \emph{rank} $\rank(G)$ of a hypergraph $G$ is the maximum rank over all its hyperedges. We allow distinct hyperedges $e_1, e_2 \in E(G)$ with $V(e_1) = V(e_2)$, and call the \emph{multiplicity} of an hyperedge $e$ the number of hyperedges $e' \in E(G)$ with $V(e') = V(e)$. This counts $e$ itself, so the multiplicity of a hyperedge is always at least one. The \emph{multiplicity} of a hypergraph is then the maximum multiplicity over all its hyperedges. For a set $E \subseteq E(G)$, we also say the \emph{multiplicity} of $E$ is the maximum multiplicity of all hyperedges in $E$.

The \emph{primal graph} $\Pc(G)$ of a hypergraph $G$ is the graph with vertex set $V(\Pc(G)) = V(G)$ that has an edge between two vertices $u,v \in V(\Pc(G))$ if and only if there exists a hyperedge $e \in E(G)$ with both $u$ and $v$ in $V(e)$. For a graph $G$, the \emph{support hypergraph} $\Hc(G)$ is the hypergraph with vertex set $V(\Hc(G)) = V(G)$ and hyperedge set $E(\Hc(G)) = V(G) \cup E(G)$ such that for every vertex $v \in V(G)$, $V(v) = \{v\}$, and for every edge $uv \in E(G)$, $V(uv) = \{u,v\}$. Note that $\Pc(\Hc(G)) = G$, $\Hc(G)$ has rank at most two and multiplicity one, and $|\Hc(G)| \le \OO(|G|)$.

We call a (hyper)graph \emph{empty} if it contains no vertices and no (hyper)edges.
Given a vertex $v \in V(G)$ of a (hyper)graph $G$, we denote by $N(v)$ the set of neighbors of $v$, that is, all vertices $u\neq v$ for which there exists a (hyper)edge $e\in E(G)$ with $u,v\in V(e)$.
A vertex $v$ is \emph{isolated} if $N(v) = \emptyset$.
Denote by $N[v] = N(v) \cup \{v\}$ the closed neighborhood of $v$, and by $\incidences(v)$ the set of (hyper)edges that are incident to $v$.
For a set $A \subseteq V(G)$ of vertices or $A \subseteq E(G)$ of (hyper)edges of a (hyper)graph $G$, we denote by $G[A]$ the sub(hyper)graph of $G$ induced by $A$. More precisely, when $A \subseteq E(G)$, we define $G[A]$ as the graph with vertex set $V(A)$ and (hyper)edge set $A$.
We use $G \setminus A = G[V(G) \setminus A]$ or $G \setminus A = G[E(G) \setminus A]$, respectively.
Similarly, for a set $A$ of vertices not in $V(G)$ or a set $A$ of (hyper)edges not in $E(G)$ but with endpoints in $V(G)$, we denote by $G \cup A$ the (hyper)graph obtained by adding $A$ to $V(G)$ or $E(G)$.
For two (hyper)graphs $G_1, G_2$, we denote by $G_1 \cup G_2$ the (hyper)graph $G_1 \cup G_2 = (V(G_1) \cup V(G_2), E(G_1) \cup E(G_2))$.

For a set $A \subseteq E(G)$ of (hyper)edges of a (hyper)graph $G$, we denote its complement by $\overline A = E(G) \setminus A$. Then, we denote by $\bd(A) = V(A) \cap V(\overline A)$ the set of \emph{boundary vertices}, by $\lambda(A) = |\bd(A)|$ the number of boundary vertices, and by $\inter(A) = V(A) \setminus \bd(A) = V(A) \setminus V(\overline A)$ the set of \emph{internal vertices} of $A$. As proven by, for example,~\cite{Robertson_Seymour_1991}, the function $\lambda \colon 2^{E(G)} \to \ZO$ is a \emph{symmetric submodular function}, meaning that
\begin{itemize}
    \item $\lambda(A \cup B) + \lambda(A \cap B) \leq \lambda(A) + \lambda(B)$ for all $A,B \in E(G)$ (submodularity) and
    \item $\lambda(A) = \lambda(\overline A)$ for all $A \subseteq E(G)$ (symmetry).
\end{itemize}

Following the definition of~\cite{Korhonen_2024}, a non-empty set $A \subseteq E(G)$ of edges in a hypergraph $G$ is called \emph{internally connected} if there is no \emph{bipartition} $(B_1,B_2)$ of $A$, i.e. $B_1 \cup B_2 = A$ and $B_1 \cap B_2 = \emptyset$, into two non-empty sets $B_1$ and $B_2$ with $\bd(B_i) \subseteq \bd(A)$ for both $i \in [2]$.
A non-empty set $A$ of hyperedges is called \emph{internally disconnected} if it is not internally connected. An \emph{internal component} of a set of hyperedges $A' \subseteq E(G)$ is then defined as an inclusion-maximal internally connected subset $A \subseteq A'$.
Alternatively, one can think of internal connectivity as follows: a set $A$ is internally connected if either it consists of a single hyperedge, or $\Pc(G)[\inter(A)]$ is non-empty, connected, and for every $e \in A$ it holds that $V(e) \setminus \bd(A) \neq \emptyset$.

\paragraph{Trees.}
A \emph{tree} is an acyclic connected graph.
To better distinguish from graphs, we sometimes call the vertices of a tree \emph{nodes}.
A \emph{rooted} tree is a tree $T$, together with a \emph{root} $r \in V(T)$.
For the following definitions, let $T$ be a tree rooted at $r$, and $t \in V(T)$ a node.

The \emph{root path} of $t$ is the unique path between $t$ and $r$.
If $t \neq r$, the \emph{parent} of $t$ is the unique node adjacent to $t$ on the root path of $t$.
A node is a \emph{child} of $t$ if it is adjacent to $t$ and not on the root path of $t$.
We denote by $\chd(t)$ the set of children of $t$, and by $\Delta(t) = |\chd(t)|$ the number of children of $t$.

For a set $X\subseteq V(T)$, we denote by $\Delta(X)=\max_{t'\in T} \Delta(t')$ the maximum number of children of any node in $X$.
We let $\Delta(T)=\Delta(V(T))$, and say that $T$ is \emph{binary} if $\Delta(T) \leq 2$.

The \emph{depth} of $t$ is the number of edges on the root path of $t$.
In particular, the depth of the root is always zero.
The \emph{depth} of $T$, denoted by $\depth(T)$, is the maximum depth over all nodes of $T$.

A \emph{leaf} is a node with no children, and an \emph{internal node} is a node that is not a leaf. 
We denote by $\leaves(T)$ the set of leaves of $T$ and by $\Vint(T) = V(T) \setminus \leaves(T)$ the set of all internal nodes of $T$.
We denote by $\cl(t)=\chd(t)\cap L(T)$ the set of \emph{leaf-children} of $t$, and for a node $\ell \in \cl(t)$ we refer to $\ell$ as a \emph{leaf-child} of $t$.

A node $t$ is an \emph{ancestor} of a node $s$, and $s$ is a \emph{descendant} of $t$, if $t$ is on the root path of $s$.
In particular, every node is both an ancestor and a descendant of itself.
The set of ancestors of a node $t$ is denoted by $\anc(t)$, and the set of descendants be $\desc(t)$.
For a set $X \subseteq V(T)$, we denote by $\anc(X) = \bigcup_{t \in X} \anc(t)$ the set of all ancestors of vertices in $X$.
A \emph{prefix} of a rooted tree is a set $P \subseteq V(T)$ with $P = \anc(P)$, i.e., a connected set of nodes that contains the root.

\subsection{(Hyper)graph decompositions}\label{sec:def_decompositions}

We consider different types of decompsitions of (hyper)graphs.

\paragraph{Tree decompositions.}

A \emph{tree decomposition} of a graph $G$ is a pair $\Tc = (T, \bag)$, where $T$ is a tree and $\bag \colon V(T) \to 2^{V(G)}$ is a function mapping each node $x$ of $T$ to a \emph{bag} of vertices such that:
\begin{itemize}
    \item $V(G) = \bigcup_{x \in V(T)} \bag(x)$.
    \item For every edge $uv \in E(G)$, there exists a node $x \in V(T)$ such that $\bag(x)$ contains both $u$ and $v$.
    \item For every vertex $v \in V(G)$, the set $\{x \in V(T) \mid v \in \bag(x)\}$ induces a connected subtree of $T$.
\end{itemize}
The \emph{width} of a tree decomposition $\Tc$ is the maximum size of a bag minus one. The \emph{treewidth} $\tw(G)$ of a graph $G$ is the minimum width of a tree decomposition of $G$. For a hypergraph $G$, the \emph{treewidth} $\tw(G)$ of $G$ is defined as the treewidth of the primal graph of $G$, i.e., $\tw(G) = \tw(\Pc(G))$. For a set $A$ of (hyper)edges of a (hyper)graph $G$, we define the \emph{treewidth} $\tw(A) = \tw(G[A])$, and the \emph{internal treewidth} $\itw(A) = \tw(G[\inter(A)])$. Similarly, for a boundaried graph $G$, we define the \emph{internal treewidth} to be $\itw(G) = \tw(G \setminus \bd(G))$.

For a graph $G$ and an integer $\cTwMod$, a set $X \subseteq V(G)$ is called a \emph{treewidth-$\cTwMod$-modulator} if $\tw(G \setminus X) \leq \cTwMod$.
We denote the size of a smallest treewidth-$\cTwMod$-modulator of $G$ by $\optTwMod(G)$.

A \emph{rooted tree decomposition} is a tree decomposition $\Tc = (T,\bag)$, where the tree $T$ is a rooted tree. For a rooted tree decomposition $\Tc = (T,\bag)$, given a node $t$ and its parent $p$, we call the set $\adh(tp) = \bag(t) \cap \bag(p)$ the \emph{adhesion} of the edge $tp$. 
An \emph{annotated tree decomposition} of a graph $G$ is a triple $(T,\bag,\edges)$ so that
\begin{itemize}
    \item $T$ is a binary tree,
    \item $\bag \colon V(T) \to 2^{V(G)}$ is a function so that $(T,\bag)$ is a tree decomposition of $G$, and
    \item $\edges \colon V(T) \to 2^{E(G)}$ is a function that maps each node $t \in V(T)$ to the set $\edges(t)$ that contains the edges $uv \in E(G)$, for which $t$ is the unique smallest-depth node with $u,v \in \bag(t)$.
\end{itemize}

Note that $G$ and $(T,\bag)$ define the annotated tree decomposition $(T,\bag,\edges)$ uniquely, and conversely, $(T,\bag,\edges)$ defines $G$.
Given a set $A \subseteq V(T)$, the restriction of $(T,\bag,\edges)$ to $A$, denoted by $(T,\bag,\edges)\restriction_A$, is the tuple $(T[A],\bag\restriction_A,\edges\restriction_A)$, where $\bag\restriction_A$ and $\edges\restriction_A$ are the restrictions of $\bag$ and $\edges$ to $A$, respectively.

Let $G$ be a graph together with an annotated tree decomposition $\Tc = (T, \bag,\edges)$ of width $t-1$. For a node $x \in V(T)$, we denote by $T_x$ the subtree of $T$ rooted at $x$, $V_x = \bigcup_{y \in \desc(x)} \bag(y)$, $E_x = \bigcup_{y \in \desc(x)} \edges(y)$, and $G_x = (V_x,E_x)$.
Then, the triple $(G_x,\bag(x),\Lambda)$ is a $t$-boundaried graph for any injective labeling $\Lambda \colon \bag(x) \to [t]$.
Given a boundaried graph $(G,B,\Lambda)$, a \emph{boundaried tree decomposition} of $(G,B,\Lambda)$ is a rooted tree decomposition, where $B$ is contained in the root bag.

A rooted tree decomposition $\Tc = (T,\bag)$ of a graph $G$ is \emph{nice} if every node $t \in V(T)$ is one of the following types:
\begin{itemize}
    \item \textbf{Leaf:} $t$ is a leaf of $T$,
    \item \textbf{Forget:} $t$ has exactly one child $s$ and there is a vertex $v \in V(G)$ (called the \emph{forgotten vertex}) such that $\bag(t) = \bag(s) \setminus \{v\}$,
    \item \textbf{Introduce:} $t$ has exactly one child $s$ and there is a vertex $v \in V(G)$ (called the \emph{introduced vertex}) such that $\bag(t) = \bag(s) \cup \{v\}$, or
    \item \textbf{Join:} $t$ has exactly two children $s_1,s_2$ and $\bag(t) = \bag(s_1) = \bag(s_2)$.
\end{itemize}

\paragraph{Protrusion decompositions.}
For integers $k$ and $c$, a \emph{$(k,c)$-protrusion decomposition} is a rooted tree decomposition $\Tc = (T,\bag)$ such that
\begin{itemize}
    \item $|\bag(r)| \leq k$,
    \item $\Delta(r) \leq k$, and
    \item $|\bag(t)| \leq c$ for every non-root node $t \in V(T)\setminus\{r\}$.
\end{itemize}
For each root child $t \in \chd(r)$, the tree decomposition $(T_t, \bag\restriction_{V(T_t)})$ is called a \emph{protrusion}.



\paragraph{Well-linkedness.} Let $G$ be a hypergraph. A set $A \subseteq E(G)$ of hyperedges is called \emph{well-linked} if for every bipartition $(A_1,A_2)$ of $A$, it holds that $\lambda(A_1) \geq \lambda(A)$ or $\lambda(A_2) \geq \lambda(A)$. Equivalently, a set $A \subseteq E(G)$ is well-linked if for any two sets $B_1,B_2 \subseteq \bd(A)$ of the same size, possibly overlapping, there are $|B_1| = |B_2|$ vertex-disjoint paths in $\Pc(G[A])$ between $B_1$ and $B_2$.
The \emph{well-linked number} $\wl(G)$ of a hypergraph $G$ (or of a set $E \subseteq E(G)$) is the largest integer $k$ so that there is a well-linked set $A \subseteq E(G)$ (or $A \subseteq E$) with $\lambda(A) = k$.

\begin{lemma}[\cite{Robertson_Seymour_1995}]\label{lem:well_linked_number_treewidth}
    For every graph $G$, $\wl(\Hc(G)) \leq 3 \cdot (\tw(G) + 1)$.
\end{lemma}

\paragraph{Superbranch decompositions.}

A \emph{superbranch decomposition} of a hypergraph $G$ is a pair $\Tc = (T,\Lc)$, where $T$ is a tree, whose every internal node has degree at least three and $\Lc \colon L(T) \to E(G)$ is a bijection from the leaves of $T$ to the edges of $G$. A \emph{rooted superbranch decomposition} is a superbranch decomposition $\Tc = (T,\Lc)$, where $T$ is a rooted tree. For a node $t \in V(T)$ of a rooted superbranch decomposition $\Tc = (T,\Lc)$ of a hypergraph $G$, we denote by $\Lc[t] \subseteq E(G)$ the set of hyperedges that are mapped to leaves in $\leaves(T)$. For each edge $tp$, where $p$ is the parent of $t$, we call the set $\adh(tp) \coloneq \bd(\Lc[t])$ the \emph{adhesion} at $tp$. The maximum size of an adhesion of $\Tc$ is denoted by $\adhsize(\Tc)$. We say that $\Tc$ has \emph{adhesion size} $c$ if $\adhsize(\Tc) \leq c$. A rooted superbranch decomposition of a hypergraph $G$ is \emph{downwards well-linked} if for every node $t \in V(T)$ the set $\Lc[t] \subseteq E(G)$ is well-linked in $G$. As $\adh(tp) = \bd(\Lc[t])$ for each node $t$ with parent $p$, we have $\adhsize(\Tc) \leq \wl(G)$.

For a rooted superbranch decomposition $\Tc = (T,\Lc)$ or a rooted tree decomposition $\Tc = (T,\bag)$, and an internal node $t \in \Vint(T)$, we define the \emph{torso} of $t$ to be the labeled hypergraph $\torso(t)$ with
\begin{itemize}
    \item $E(\torso(t)) = \{e_s \mid st \in E(T)\}$,
    \item $V(e_s) = \adh(st)$ for each $e_s \in E(\torso(t))$, and
    \item $V(\torso(t)) = \bigcup\limits_{e_s \in E(\torso(t))} V(e_s)$.
\end{itemize}

We say that the hyperedge $e_s \in E(\torso(t))$ \emph{corresponds} to the respective node $s$.
Consider a set $A \subseteq E(\torso(t))$ that does not contain the hyperedge $e_p$ corresponding to the edge $tp \in E(T)$, where $p$ is the parent of $t$. Then, $A$ corresponds to a set of children $\mathcal{C}_A = \{c \in \chd(t) \mid e_c \in A\}$ of $t$, which in turn corresponds to a set $\bigcup_{c \in \mathcal{C}_A} \Lc[c]$ of hyperedges of $G$. We denote this set by $A \expand \Tc$. As shown by~\cite{Korhonen_2025}, each well-linked set in the torso of a node of a superbranch decomposition of a hypergraph $G$ corresponds to a well-linked set in $G$.

\begin{restatable}{lemma}{wellLinkedTrans}\label{lem:well_linked_torso}
    Let $G$ be a hypergraph and $\Tc = (T,\Lc)$ a rooted superbranch decomposition of $G$ with root $r$. Let also $t \in V_{\inter}(T)$ be a node, so that $\Lc[c]$ is well-linked for every child $c$ of $t$. Let $C \subseteq E(\torso(t))$ be the set of hyperedges corresponding to the children of $t$. Then, a set $A \subseteq C$ is well-linked in $\torso(t)$ if and only if $A \expand \Tc$ is well-linked in $G$.
\end{restatable}
\begin{proof}
    This is a straightforward generalization of~\cite[Lemma 4.3]{Korhonen_2025}. For the sake of completeness, the full proof can be found in \Cref{sec:missingproofs}.
\end{proof}



We observe that there is a natural connection between superbranch and tree decompositions.
Let $G$ be a hypergraph and let $\Tc = (T, \Lc)$ be a superbranch decomposition of $\Hc(G)$ rooted at $r \in V(T)$. Let $\tilde \Tc = (\tilde T, \bag)$ be a rooted tree decomposition of $G$ with root $r \in V(\tilde T)$ (with a slight abuse of notation, we denote some nodes of $T$ and $\tilde T$ with the same symbol, to simplify things later).
We say that $\tilde \Tc$ \emph{corresponds} to $\Tc$ if
\begin{itemize}
    \item $\bag(r) = V(\torso(r))$,
    \item $\chd_\Tc(r) = \chd_{\tilde \Tc}(r)$, and
    \item for each $t \in \chd(r)$, we have $V(\Lc[t]) = \bigcup_{t' \in \desc_{\tilde \Tc}(t)} \bag(t')$.
\end{itemize}

A protrusion decomposition $\Tc = (T,\bag)$ of a graph $G$ with root $r$ is called \emph{normal} if for every hyperedge $e \in E(\Hc(G))$ there is a node $t \in V(T) \setminus \{r\}$ with $V(e) \subseteq \bag(t)$.
We now show that every protrusion decomposition that corresponds to a superbranch decomposition is normal.

\begin{lemma}\label{lem:corresponding-is-normal}
    Let $G$ be a graph and let $\tilde \Tc = (\tilde T, \bag)$ be a protrusion decomposition of $G$ corresponding to a superbranch decomposition $\Tc = (T, \Lc)$ of $\Hc(G)$.
    Then $\tilde \Tc$ is normal.
\end{lemma}
\begin{proof}
    Let $r$ be the root of $T$ and $\tilde r$ be the root of $\tilde T$.
    Let $e \in E(\Hc(G))$.
    We wish to show that $V(e) \subseteq \bag(\tilde t)$ for some $\tilde t \in V(\tilde T)\setminus\{\tilde r\}$.
    Since $\tilde \Tc$ is a tree decomposition, we are done unless $V(e) \subseteq \bag(\tilde r)$, so assume this.
    Since $r$ is not a leaf, we have $e \in \Lc[t]$ for some $t \in \chd(r)$.
    In particular, $V(e) \subseteq V(\Lc[t])$.
    Since $\tilde \Tc$ corresponds to $\Tc$, we have $V(\Lc[t]) = \bigcup_{\tilde t' \in \desc(\tilde t)} \bag(\tilde t')$ for some $\tilde t \in \chd(\tilde r)$.
    So $V(e) \subseteq \bag(\tilde t')$ for some $\tilde t' \in \desc(\tilde t)$.
    For each $v \in V(e)$, the subgraph of $\tilde T$ induced by the vertices with bags containing $v$ is connected, and contains both $\tilde r$ and a node in $\desc(\tilde t)$, so in particular, it contains $\tilde t$.
    So $V(e) \subseteq \bag(\tilde t)$.
\end{proof}

\paragraph{Representation of objects.} 
Graphs are represented by the adjacency list format, with each edge $uv$ appearing in both the list of $u$ and the list of $v$, and having a pointer to its other appearance.
Note that we can insert a vertex in $\OO(1)$ time, insert an edge in $\OO(1)$ time (given that we know it does not already exist), delete an edge in $\OO(1)$ time given a pointer to it, and delete an isolated vertex in $\OO(1)$ time.
A hypergraph $G$ is represented as a bipartite graph with bipartition $(V(G), E(G))$, and there is an edge between $v\in V(G)$ and $e\in E(G)$ if and only if $v \in V(e)$. A tree is represented as a graph, and a rooted tree also contains a global pointer to the root, and each non-root node $t$ stores a pointer to the edge $tp$ where $p$ is the parent of $t$.

A superbranch decomposition $\Tc=(T,\Lc)$ of a hypergraph $G$ is represented as follows: 
\begin{itemize}
    \item a representation of $T$,
    \item for each $\ell\in L(T)$ a pointer from $\ell$ to $\Lc(\ell)$, 
    \item for each $e\in E(G)$ a pointer from $e$ to $\Lc^{-1}(e)$,
    \item for each node $t\in V(T)$, the number of $|L[t]|$ of leaf descendants of it,
    \item for each internal node $t\in \Vint$, a representation of $\torso(t)$, where additionally each hyperedge $e_s\in E(\torso(t))$ corresponding to an edge $st\in E(T)$ stores a pointer to $st$, and $st$ stores a pointer to $e_s$.
\end{itemize}

An annotated tree decomposition $\Tc=(T, \bag, \edges)$ rooted at $r$ is represented as follows:
\begin{itemize}
    \item a representation of $T$,
    \item for each $t\in V(T)\setminus \{r\}$, the set $\bag(t)$ and $\edges(t)$ are each stored as a linked list to which $t$ contains pointers to, and
    \item for the root $r$, the set $\edges(r)$ and $\bag(t)$ are each stored as a balanced binary search tree containing all edges $e\in \edges(r)$ and vertices $v\in \bag(r)$, respectively, which $r$ contains pointers to.
\end{itemize}

We now describe an algorithm that, given a set of hyperedges $C\subseteq E(G)$, computes the vertex set $V(C)$ and the boundary $\bd(C)$.
\begin{lemma}\label{lem:compute-boundary}
    Let $G$ be a hypergraph of rank $r$ whose representation is stored. 
    There is an algorithm that given a set $C\subseteq E(G)$, in time $\OO(|C| \cdot r)$ returns the sets $V(C)$ and $\bd(C)$. 
\end{lemma}
\begin{proof}
    We start by computing a linked-list representation of $V(C)$.
    For each hyperedge $e\in C$, we iterate over all vertices $v\in V(e)$ and add $v$ to $V(C)$, if it has not already been added.
    The running time is at most $\OO(|C| \cdot r)$.
    
    To compute $\bd(C)$, we do the following.
    In the bipartite representation of $G$, we mark all vertices corresponding to edges $e\in C$. Then, for every vertex $v\in V(C)$, we check if $v$ has an unmarked neighbor and if so, include $v$ in $\bd(C)$. 
    Every edge incident to a marked vertex is touched at most once, and therefore the running time is at most $\OO(|C| \cdot r)$.   
\end{proof}

\paragraph{Basic (hyper)graph operations.} For representing manipulations of (hyper)graphs, we introduce a set of operations, called \emph{basic (hyper)graph operations}, defined as follows.
\begin{itemize}
    \item $\addVertex$: Given a new vertex $v\notin V(G)$, add $v$ to $V(G)$ and return a pointer to it, in $\OO(1)$ time.
    \item $\deleteVertex$: Given a pointer to an isolated vertex $v\in V(G)$, remove $v$ from $V(G)$ in $\OO(1)$ time.
    \item $\addEdge / \addHyperedge$: For a new (hyper)edge $e\notin E(G)$ with $V(e)\subseteq V(G)$, given pointers to each vertex $v\in V(e)$, add $e$ to $E(G)$ and return a pointer to it, in $\OO(|V(e)|)$ time. 
    \item $\deleteEdge / \deleteHyperedge$: Given a pointer to a (hyper)edge $e\in E(G)$, remove $e$ from $E(G)$, in $\OO(|V(e)|)$ time.
\end{itemize}

If $G$ and $G'$ are the (hyper)graphs before and after one basic (hyper)graph operation $o$, we say that $o$ \emph{transforms} $G$ into $G'$.
A \emph{sequence of operations} is a finite sequence $\mc C = c_1 c_2 \dots c_k$, where each $c_i$ specifies a basic (hyper)graph operation together with the information required to perform it: For adding a vertex, the vertex $v$ is stored, for deleting a vertex, the pointer to vertex $v$ is stored, for adding a (hyper)edge, the edge $e$ is stored, and for deleting a (hyper)edge, the list of pointers to each vertex $v\in V(e)$ is stored. Applying $\mc C$ to a (hyper)graph $G$ produces a sequence of (hyper)graphs $G=G^{(0)},G^{(1)},\dots,G^{(k)}$. The \emph{length} of $\mc C$ is $|\mc C|=k$.
The \emph{size} of a sequence $\|{\mc C}\|$ is the total information stored: A vertex operation has size $1$, and a (hyper)edge operation has size $|V(e)|+1$. This also means that, given a sequence $\mc C$ of operations, it can be performed in $\OO(\|{\mc C}\|)$ time. 
Note that $|\mc C|\leq \|\mc C\|$.


%
%
%
%
%

\paragraph{Tree decomposition automata.}
We will now informally define tree decomposition automata for annotated tree decompositions.
For more formal definitions, see~\Cref{sec:automata}.
The idea of a tree decomposition automaton is to assign a state to each node of an annotated tree decomposition, in a way where the state of a node depends only on the immediate properties of the node and its children, as well as the states of its children.
In particular, one should think of a tree decomposition automaton as describing a bottom-up dynamic programming algorithm on tree decompositions.

More concretely, a \emph{tree decomposition automaton} of \emph{width} $\ell$ is a tuple $(Q, F, \iota, \delta)$, where $Q$ is a set of \emph{states} (containing the \emph{null state} $\bot$), $F \subseteq Q$ is a set of \emph{accepting states}, $\iota$ is an \emph{initial mapping} assigning a state $\iota(G)$ to each boundaried graph $G$ with $|V(G)| \leq \ell+1$, and $\delta$ is a \emph{transition mapping} computing the state of a node $x$ given $\bag(x)$, $\bag(y)$ for every $y \in \chd(x)$, $\edges(x)$, and the states of the children of $x$.
The state of a leaf $l \in L(T)$ is given by $\iota(G_l)$.
Let $\Tc = (T,\bag,\edges)$ be an annotated tree decomposition of width at most $\ell$. 
The \emph{run} $\run_\autom^\Tc \colon V(T) \to Q$ of $\autom$ on $\Tc$ is the map from nodes to their states.
The initial mapping $\iota$ and the transition mapping $\delta$ are to be computable, whereas the run $\run_\autom^\Tc$ merely recalls states already computed using $\iota$ and $\delta$.
Note that we can compute the full run by using the initial mapping $\iota$ on each leaf of $\Tc$ and then using the transition mapping $\delta$ in a bottom-up fashion.
We say that $\autom$ has \emph{evaluation time} $\tau$ if the functions $\iota$ and $\delta$ can be evaluated in time $\tau$, and it can be decided whether a state $q \in Q$ is in $F$ in time $\tau$.

\subsection{Parameterized graph problems and kernelization}\label{sec:def_fpt}

In this section, we introduce all the necessary definitions around parameterized graph problems and kernelization, mostly following~\cite{kernelization-book}.
A \emph{parameterized graph problem} is a problem $\Pi$ with instances $(G,k)$, where $G$ is a graph and $k \in \mathbb{Z}$ an integer such that one of the following holds:
\begin{itemize}
    \item For all $G, k$ with $k < 0$: $(G,k) \in \Pi$.
    \item For all $G, k$ with $k < 0$: $(G,k) \notin \Pi$.
\end{itemize}
Naturally, for maximization problems, we will always consider the first case, and for minimization problems, the second case.
We restrict ourselves to optimization problems that are parameterized by the solution size.
For this, we denote by $\opt_{\Pi}(G)$ the size of an optimal solution for $\Pi$ on $G$, i.e., 
\begin{align*}
    \opt_\Pi(G) = \begin{cases}
    \min\{k \mid (G,k) \in \Pi\} & \text{if such a } k \text{ exists}\\
    +\infty & \text{otherwise}
    \end{cases}
\end{align*}
if $\Pi$ is a minimization problem, and
\begin{align*}
    \opt_\Pi(G) = \begin{cases}
    \max\{k \mid (G,k) \in \Pi\} & \text{if such a } k \text{ exists}\\
    -\infty & \text{otherwise}
    \end{cases}
\end{align*}
if $\Pi$ is a maximization problem.

A parameterized problem is \emph{linearly treewidth-bounding} on a class of graphs $\mc G$ if there are constants $\cTwMod$ and $c$ so that for all $G \in \mc G$ it holds that $\tw(G) \le \OO(\optTwMod(G))$.


\paragraph{FPT and kernelization.}

An parameterized graph problem $\Pi$ is called \emph{fixed-parameter tractable} if there exists an algorithm $\mc A$ (called \emph{fixed-parameter algorithm}) that, given a graph $G$ and an integer $k$, correctly decides if $(G,k) \in \Pi$ in time $f(k) \cdot n^{\OO(1)}$, where $f$ is a computable function. The complexity class of all fixed-parameter tractable problems is called \emph{FPT}. The most fundamental way of finding fixed-parameter algorithms is via kernelization algorithms.

A \emph{kernelization algorithm} for a parameterized graph problem $\Pi$ is an algorithm $\mc A$ that, given a graph $G$, outputs in polynomial time a graph $K$ together with a non-positive integer $\Delta$ such that for every integer $k$, $(G,k) \in \Pi$ if and only if $(K,k+\Delta) \in \Pi$. Moreover, there exists a computable function $g \colon \mathbb{N} \to \mathbb{N}$ such that $|K| \leq g(k)$. The output $(K,\Delta)$ of the algorithm $\mc A$ is called the \emph{kernel}. If $g(k) \le k^{\OO(1)}$ or $g(k) \le \OO(k)$, $K$ (or $(K,\Delta)$) is called a \emph{polynomial} or \emph{linear kernel}, respectively.

\paragraph{(Counting) monadic second order logic.} The syntax of \emph{monadic second order logic} ($\mso$) of graphs consists of variables for vertices, edges, sets of vertices, and sets of edges, Boolean connectives $\land, \lor , \neg, \rightarrow, \leftrightarrow$, and quantifiers $\exists, \forall$ that can be applied to all four types of variables. Furthermore, the following atomic formulas are included:
\begin{itemize}
    \item $x \doteq y$, where $x$ and $y$ are variables of the same type;
    \item $x \in X$, where $x$ is a vertex (edge) variable and $X$ a vertex (edge) set variable;
    \item $\mathsf{adj}(x,y)$, where $x$ and $y$ are vertex variables, and the interpretation is that $x$ and $y$ are adjacent;
    \item $\mathsf{inc}(x,y)$, where $x$ is a vertex variable and $y$ an edge variable, and the interpretation is that $x$ and $y$ are incident.
\end{itemize}
Additionally, $\mso$ can be extended to \emph{counting monadic second order logic} ($\cmso$) by adding the atomic formulas $\mathsf{card}_{q,r}$ for $0 \leq q \leq r$, $r \geq 2$, which are interpreted as $\mathsf{card}_{q,r}(S) = \true$ if and only if $|S| \equiv q$ $\pmod r$. A graph class $\mc G$ is called \emph{$\cmso$-definable} if there is a $\cmso$-formula $\phi$ such that $G \in \mc G$ if and only if $G \models \phi$. For a detailed introduction to $\cmso$, we refer to~\cite{Courcelle_1990,Courcelle_1992,Courcelle_1997}.

\paragraph{Finite integer index.}

In our dynamic kernelization algorithm, we want to replace boundaried subgraphs by ``equivalent'' graphs. For this, the notion of finite integer index plays a central role.
Let $X = (G_X, B_X, \Lambda_X)$, $Y = (G_Y, B_Y, \Lambda_Y)$ be two boundaried graphs. We define $X \oplusu Y$ as the (not boundaried) graph obtained by taking the disjoint union of $G_X$ and $G_Y$ and then identifying boundary vertices of $G_X$ and $G_Y$ that share the same label. If $X$ and $Y$ have the same label set, we further define $X \oplusb Y = (X \oplusu Y,B,\Lambda)$ to be the boundaried graph, where the boundary $B$ contains exactly the vertices that have been identified in the gluing process, and the labeling $\Lambda$ is the corresponding labeling from $X$ (or equivalently, $Y$). We observe that $\oplusu$ is commutative and $\oplusb$ is commutative and associative. If $G = X \oplusu Y$ and $Y'$ is another boundaried graph, we say that $G' = X \oplusu Y'$ is the graph obtained from $G$ by \emph{replacing} $Y$ with $Y'$.

For a parameterized graph problem $\Pi$ and two boundaried graphs $G_1,G_2 \in \mc F$, we say $G_1 \equiv_\Pi G_2$ if and only if $\labelSet(G_1) = \labelSet(G_2)$ and there exists a \emph{transposition constant} $\Delta \coloneq \Delta_{\Pi}(G_1,G_2) \in \mathbb{Z}$ such that
\begin{equation*}
    \forall (F,k) \in \mc F \times \mathbb{Z}: (F \oplusu G_1,k) \in \Pi \Leftrightarrow (F \oplusu G_2, k + \Delta) \in \Pi.
\end{equation*}
We remark that $\equiv_\Pi$ is indeed an equivalence relation (see~\cite[Exercise 16.2]{kernelization-book}). 

The goal of protrusion replacement is to replace a large protrusion with an equivalent small graph. To do this efficiently, we need to bound the number of equivalence classes of $\equiv_\Pi$. Unfortunately, this is not possible, as it is easy to see that two boundaried graphs with different label sets belong to different equivalence classes. However, for some problems $\Pi$, the number of equivalence classes might be bounded if we restrict $\equiv_\Pi$ to boundaried graphs with the same label set $I \subseteq \ZI$.
A parameterized problem $\Pi$ has \emph{finite integer index} (\emph{FII}) if for every finite $I \subseteq \ZI$, the number of equivalence classes of $\equiv_{\Pi}$ that are subsets of $\mc F_I$ is finite.

%% file: existence.tex
\section{Existence of mergeable root children}\label{sec:existence}

When the degree of the root of the superbranch decomposition grows too large, we need to reduce it without creating another high-degree node and without violating the downwards well-linkedness. That is, we need to find a set of root-children, identified by their corresponding hyperedges in the torso, that we can merge into a downwards well-linked subtree with small treewidth and adhesion size.
The goal of this section is to prove \Cref{mergeable-children}, which asserts the existence of such a set of root-children.
Then, in the subsequent \Cref{sec:local_search}, we will describe how this set can be found efficiently.

\begin{restatable}{lemma}{MergeableChildren}\label{mergeable-children}
    Let $H$ be a graph and $\cTwMod$, $\cSbdProt$ integers. There are integers $\cNewSbdProt \coloneq \cNewSbdProt(H,\cTwMod)$ and $\cSbdRootDeg \coloneq \cSbdRootDeg(H,\cTwMod,\cSbdProt)$ such that for every $k$ the following holds:
    Let $G$ be an $H$-topological-minor-free graph that has a treewidth-$\cTwMod$-modulator of size $k$.
    Let $\Tc = (T, \Lc)$ be a downwards well-linked superbranch decomposition of $\Hc(G)$ rooted at $r$ and with adhesion size $\cSbdProt$. 
    If $\Delta(r) \geq \cSbdRootDeg \cdot k$, then there exists a subset of hyperedges $B \subseteq E(\torso(r))$ such that
    \begin{enumerate}
        \item\label{ite:boundary} $\lambda(B) \leq \cNewSbdProt$,
        \item\label{ite:treewidht} $\itw(B \expand \Tc) \leq \cNewSbdProt$,
        \item\label{ite:size} $2^{\cNewSbdProt+2} \leq |B| \leq \OO_{H,\cTwMod}(\cSbdProt)$, and
        \item\label{ite:internal_components} for every internal component $B'$ of $B$, we have $\bd(B') = \bd(B)$.
    \end{enumerate}
    Furthermore, the functions $\cNewSbdProt(H,\cTwMod)$ and $\cSbdRootDeg(H,\cTwMod,\cSbdProt)$ are computable.
\end{restatable}

The rest of this section consists of the proof of \Cref{mergeable-children}, divided into two main parts.
These parts correspond to the two reasons for the degree $\Delta(r)$ to be high: Either $\torso(r)$ has many vertices, or there are hyperedges with high multiplicity in $\torso(r)$.
We start with the first case in \Cref{sec:many_torso_vertices}, and continue with the second in \Cref{sec:unique_or_duplicates}.

\subsection{Many torso vertices}
\label{sec:many_torso_vertices}
This section deals with the case of \Cref{mergeable-children} where there are many vertices in $\torso(r)$.
Specifically, this section is dedicated to the proof of the following lemma.

\begin{restatable}{lemma}{ManyTorsoVertices}\label{lem:many_torso_vertices}
    Let $H$ be a graph and $\cTwMod$, $\cSbdProt$ integers.
    There are integers $\cNewSbdProt \coloneq \cNewSbdProt(H,\cTwMod)$ and $\cTorsoVertices \coloneq \cTorsoVertices(H,\cTwMod,\cSbdProt)$ such that for every $k$ the following holds:
    Let $G$ be an $H$-topological-minor-free graph with a treewidth-$\cTwMod$-modulator $X$ of size $k$.
    Let $\Tc = (T, \Lc)$ be a downwards well-linked superbranch decomposition of $\Hc(G)$ with adhesion size $\cSbdProt$ and root $r$.
    If $|V(\torso(r))| \geq \cTorsoVertices \cdot k$, then there exists a subset of hyperedges $B \subseteq E(\torso(r))$ such that
    \begin{enumerate}
        \item $\lambda(B) \leq \cNewSbdProt$,
        \item $\itw(B \expand \Tc) \leq \cNewSbdProt$,
        \item $2^{2\cNewSbdProt+2} \le |B|$, and
        \item $|V(B)| \le \OO_{H,\cTwMod}(\cSbdProt)$.
    \end{enumerate}
    Furthermore, the functions $\cNewSbdProt(H,\cTwMod)$ and $\cTorsoVertices(H,\cTwMod,\cSbdProt)$ are computable.
\end{restatable}

We split the proof of \Cref{lem:many_torso_vertices} into two parts. First, in \Cref{sec:before_uncrossing} we show that we can find a subset $C \subseteq E(\Hc(G))$ of hyperedges of the support hypergraph of $G$, instead of a subset of $E(\torso(r))$, but with similar properties as the desired set $B$ from \Cref{mergeable-children} (see \Cref{lem:before_uncrossing}). Then, in \Cref{sec:uncrossing} we show that we can ``uncross'' this set $C$, i.e., for every child $c$ of the root, we either include or exclude the entire set $\Lc[c]$ from $C$, roughly preserving our conditions on the boundary size and internal treewidth of $C$. Finally, the set $B$ is obtained by taking exactly those torso hyperedges that correspond to root children $c$, where $\Lc[c]$ is included in $C$.

\subsubsection{Before uncrossing}\label{sec:before_uncrossing}

The goal of this section is to prove the following \Cref{lem:before_uncrossing} saying that there exists a set $C \subseteq \Hc(G)$ with small boundary and internal treewidth, which we will uncross later in \Cref{sec:uncrossing} to obtain the set $B \subseteq E(\torso(r))$ for \Cref{lem:many_torso_vertices}.
The last requirement (\Cref{cond:high_tw_subtree}) is crucial to ensure that the treewidth does not grow too much during the uncrossing.

\begin{restatable}{lemma}{BeforeUncrossing}\label{lem:before_uncrossing}
    Let $H$ be a graph and $\cTwMod$, $\cSbdProt$ integers.
    There are integers $\cPdProt \coloneq \cPdProt(H,\cTwMod)$ and $\cTorsoVertices \coloneq \cTorsoVertices(H,\cTwMod, \cSbdProt)$ such that for every $\cTorsoIntersect \geq 3\cPdProt+1$ and every $k$ the following holds:
    Let $G$ be an $H$-topological-minor-free graph that has a treewidth-$\cTwMod$-modulator of size $k$. Let also $\Tc = (T,\Lc)$ be a downwards well-linked superbranch decomposition of $\Hc(G)$ rooted at $r$ and with adhesion size $\cSbdProt$.
    If $|V(\torso(r))| \geq \cTorsoVertices \cdot (\cTorsoIntersect + 1) \cdot k$, 
    then there exists a set of hyperedges $C \subseteq  E(\Hc(G))$, such that
    \begin{enumerate}
        \item\label{cond:size} $\cTorsoIntersect \leq |\inter(C) \cap V(\torso(r))| \leq 3 \cdot \cTorsoIntersect$,
        \item\label{cond:boundary} $\lambda(C) \leq \cPdProt + 1$,
        \item\label{cond:internal_treewidth} $\itw(C) \leq \cPdProt$, and
        \item\label{cond:high_tw_subtree} for every root-child $s \in \chd(r)$ with $\itw(\Lc[s]) > \cTwMod$, we have $\inter(C) \cap \bd(\Lc[s]) = \emptyset$.
    \end{enumerate}
    Furthermore, $\cNewSbdProt$ is computable when given $H$ and $\cTwMod$, and $\cSbdRootDeg$ is computable when given $H$, $\cTwMod$, and $\cSbdProt$.
\end{restatable}

We will find the set $C$ within a subtree of an optimal protrusion decomposition, the existence of which is given by the following two lemmas. First, Fomin, Lokshtanov, Saurabh, and Thilikos~\cite{Fomin_Lokshtanov_Saurabh_Thilikos_2020} showed that every $H$-minor-free graph with a treewidth-$\cTwMod$-modulator of size $k$ has a $(\OO_{H,\cTwMod}(k), \OO_{H,\cTwMod}(1))$-protrusion decomposition. This result has been extended to $H$-topological-minor-free graphs by Kim, Langer, Paul, Reidl, Rossmanith, Sau, and Sikdar~\cite{Kim_Langer_Paul_Reidl_Rossmanith_Sau_Sikdar_2015}, even though it was only explicitly stated and proven later by Kim, Serna, and Thilikos~\cite{Kim_Serna_Thilikos_2019} (the proof appears only in the full version of~\cite{Kim_Serna_Thilikos_2019}).


\begin{lemma}[\cite{Fomin_Lokshtanov_Saurabh_Thilikos_2020,Kim_Langer_Paul_Reidl_Rossmanith_Sau_Sikdar_2015,Kim_Serna_Thilikos_2019}]\label{lem:modulator_implies_decomposition}
    Let $H$ be a graph and $\cTwMod$ an integer.
    For every $k$, every $H$-topological-minor-free graph $G$ with a treewidth-$\cTwMod$-modulator $X$ of size $k$ has an $(\OO_{H,\cTwMod}(k),\OO_{H,\cTwMod}(1))$-protrusion decomposition, whose root bag contains $X$.
\end{lemma}



We remark that in the case of topological-minor-free graphs it is nowhere stated explicitly that the treewidth-$\cTwMod$-modulator $X$ is contained in the root bag of the protrusion decomposition. However, this follows directly from the algorithm~\cite[Algorithm 1]{Kim_Langer_Paul_Reidl_Rossmanith_Sau_Sikdar_2015} (see also Lemma 2 in the full version of~\cite{Kim_Serna_Thilikos_2019}).

We only want to merge root-children whose subtrees induce a small internal treewidth. To that end, we have one additional requirement on our ``optimal'' protrusion decomposition, in which we find the set $C$ for \Cref{lem:before_uncrossing}, namely, we want the root bag to contain the boundaries of all high-treewidth root-children of our current superbranch decomposition. The existence of such a protrusion decomposition is given by the following lemma.






\begin{lemma}\label{exists-prot-decomp}
    Let $H$ be a graph and $\cTwMod$, $\cSbdProt$ integers.
    For every $k$ the following holds: Let $G$ be an $H$-topological-minor-free graph with treewidth-$\cTwMod$-modulator $X$ of size $k$. 
    Let also $\Tc = (T, \Lc)$ be a downwards well-linked superbranch decomposition of $\Hc(G)$ with root $r$ and adhesion size~$\cSbdProt$.
    Let $Y = \{c \in \chd(r) \mid \itw(\mc \Lc[c]) > \cTwMod\}$ and $X' = X \cup \bigcup\limits_{c \in Y} \bd(\Lc[c]).$
    Then, there exists an $(\OO_{H,\cTwMod,\cSbdProt}(k), \OO_{H,\cTwMod}(1))$-protrusion decomposition of $G$ whose root bag contains $X'$.
\end{lemma}
\begin{proof} 
    First, we remark that $X' \supseteq X$ is again a treewidth-$\cTwMod$-modulator of $G$. In the following, we first bound the size of $X'$, and then apply \Cref{lem:modulator_implies_decomposition} to obtain a protrusion decomposition of $G$.

    Let $s \in Y$ be a root-child with $\itw(\Lc[s]) > \cTwMod$. Since $X$ is a treewidth-$\cTwMod$-modulator, $G \setminus X$ has treewidth at most $\cTwMod$. Thus, the subgraph of $G$ induced by $\inter(\Lc[c])$ is not completely contained in $G \setminus X$, so $X \cap \inter(\Lc[c]) \neq \emptyset$.
    Moreover, for two distinct root-children $s_1,s_2 \in Y$, we have that $\inter(\Lc[s_1]) \cap \inter(\Lc[s_2]) = \emptyset$, and in particular, the vertices from $\inter(\Lc[s_1])$ and $\inter(\Lc[s_2])$ that are contained in $X$ are disjoint. Since $X$ has size $k$, it follows that there are at most $k$ such root-children $s$ with $\itw(\Lc[s]) > \cTwMod$, and thus $|Y| \leq k$ and $|X'| \leq k + k \cdot \cSbdProt$.
    
    Applying \Cref{lem:modulator_implies_decomposition}, $G$ has an $(\OO_{H,\cTwMod}(|X'|), \OO_{H,\cTwMod}(1))$-protrusion decomposition, that is, an $(\OO_{H,\cTwMod,\cSbdProt}(k), \OO_{H,\cTwMod}(1))$-protrusion decomposition, whose root bag contains $X'$.
\end{proof}

Now we have everything that we need to prove \Cref{lem:before_uncrossing}, which we re-state here before the proof.

\BeforeUncrossing*
\begin{proof}

    Let $H$ be a graph, $\cTwMod$, $\cSbdProt$, $k$ integers, and let $G$ be an $H$-topological-minor-free graph that has a treewidth-$\cTwMod$-modulator $X$ of size $k$. Let $\Tc = (T,\Lc)$ be a downwards well-linked superbranch decomposition of $\Hc(G)$ with adhesion size $\cSbdProt$ that is rooted at $r$. Let $Y$ and $X'$ be defined as in \Cref{exists-prot-decomp}, i.e., $Y = \{s \in \chd(r) \mid \itw(\Lc[s]) > \cTwMod\}$ and $X' = X \cup \bigcup\limits_{s \in Y} \bd(\Lc[s])$.
    Then, by \Cref{exists-prot-decomp}, there exist constants $c_1 \le \OO_{H,\cTwMod,\cSbdProt}(1)$ and $c_2 \le \OO_{H,\cTwMod}(1)$ such that $G$ has a $(c_1 \cdot k, c_2)$-protrusion decomposition $\Tc^* = (T^*,\bag^*)$ rooted at $r^*$ with $X' \subseteq \bag^*(r^*)$.
    We set $\cTorsoVertices = \cTorsoVertices(H,\cTwMod,\cSbdProt) \coloneq 2c_1$ and $\cPdProt = \cPdProt(H,\cTwMod) \coloneq c_2$.

    Given an integer $\cTorsoIntersect \geq 3\cPdProt + 1$, we show that if $|V(\torso(r))| \geq \cTorsoVertices \cdot (\cTorsoIntersect + 1) \cdot k$, then there exists a set $C \subseteq E(\Hc(G))$ that satisfies \Cref{cond:size,cond:boundary,cond:internal_treewidth,cond:high_tw_subtree}. First, we observe that since the bag of the root $r^*$ of the ``optimal'' protrusion decomposition $\Tc^*$ has size at most $c_1 \cdot k$, there are at most $c_1 \cdot k$ vertices from $V(\torso(r))$ in $\bag^*(r^*)$. By the pigeonhole principle, and since $r^*$ has at most $c_1 \cdot k$ children, there is a root-child $s \in \chd(r^*)$ such that $\inter(E(\Hc(G_s)))$ contains at least
    \[\frac{|V(\torso(r))| - c_1 \cdot k}{c_1 \cdot k}\geq \frac{\cTorsoVertices \cdot (\cTorsoIntersect + 1) \cdot k - c_1 \cdot k}{c_1 \cdot k} = \frac{2 \cdot c_1 \cdot (\cTorsoIntersect + 1) \cdot k - c_1 \cdot k}{c_1 \cdot k} = 2\cTorsoIntersect + 1 \geq \cTorsoIntersect + \cPdProt + 1\]
    vertices from $V(\torso(r))$.

    We can assume that every node $t \in V(T^*) \setminus \{r^*\}$ has at most two children. This can be achieved by replacing every node $t \in V(T^*) \setminus \{r^*\}$ with more than two children, say $s_1,s_2,\dots,s_q$, by a path $t_1,t_2,\dots,t_q$, where $\bag(t_i) = \bag(t)$, and then replacing the edge $s_i t$ by $s_i t_i$ for $i = 1,\dots,q$.
    Let $t \in V(T^*)$ such that $|\inter(E(\Hc(G_t))) \cap V(\torso(r))| \geq \cTorsoIntersect + \cPdProt + 1$ but for all (one or two) children $s \in \chd(t)$ we have $|\inter(E(\Hc(G_{s}))) \cap V(\torso(r))| \leq \cTorsoIntersect + \cPdProt$. We set $C \coloneq E(\Hc(G_t)) \setminus \{e \mid V(e) = \{v\} \subseteq \bag(t)\} \subseteq E(\Hc(G))$. Since removing the singleton edges $e$ with $V(e) = \{v\} \in \bag(t)$ from $E(\Hc(G_t))$ cannot increase the number of internal vertices and decreases it by at most $|\bag(t)| \leq \cPdProt + 1$, we have $\cTorsoIntersect \leq |\inter(E(\Hc(G_t))) \cap V(\torso(r))| - \cPdProt - 1 \leq |\inter(C) \cap V(\torso(r))| \leq |\inter(E(\Hc(G_t))) \cap V(\torso(r))| \leq 2(\cTorsoIntersect + \cPdProt) + \cPdProt + 1 \leq 3\cTorsoIntersect$, so \Cref{cond:size} is satisfied.
     
    Note that any boundary vertex $v \in \bd(C)$ is also in $\bag(t)$, so $\lambda(C) \leq \cPdProt + 1$, satisfying~\Cref{cond:boundary}.
    Furthermore, we have $\itw(C) \leq \tw(\Hc(G_t)) \leq \cPdProt$, satisfying~\Cref{cond:internal_treewidth}. Now, consider a vertex $x \in X'$ with $x \in V(C)$. Then, there is a node $t'$ in the subtree of $T^*$ rooted at $t$ such that $x \in \bag(t')$. Since $x \in \bag(r^*)$, it follows that $x \in \bag(t)$. Thus, the singleton edge $e_x \in E(\Hc(G))$ with $V(e_x) = \{x\}$ is not in $C$. Since $x \in V(C)$, it follows that $x \in \bd(C)$. Hence, $V(C) \cap X' \subseteq \bd(C) \cap X'$, so $\inter(C) \cap X' = \emptyset$.
    Since for every root-child $s$ with $\itw(\Lc[s]) \geq \cTwMod$, we have $\bd(\Lc[s]) \subseteq X'$, it holds that $\inter(C) \cap \bd(\Lc[s]) \subseteq \inter(C) \cap X' = \emptyset$, finally satisfying \Cref{cond:high_tw_subtree}.
\end{proof}

\subsubsection{Uncrossing}\label{sec:uncrossing}
The second part of the proof of \Cref{lem:many_torso_vertices} is to uncross the set $C$ from \Cref{lem:before_uncrossing}, i.e., for every root-child $s$, either include or exclude $\Lc[s]$ completely from $C$. During the uncrossing, we need to make sure that the treewidth and the boundary stay bounded, and the set does not become too small. The following lemma tells us that always one of the two options, either including or excluding a well-linked set $A \subseteq E(G)$ from a set $B \subseteq E(G)$, does not increase the size of the boundary.

\begin{lemma}\label{lem:uncrossing}
    Let $G$ be a hypergraph, $A \subseteq E(G)$ a well-linked set and $B \subseteq E(G)$. Then, either $\lambda(B \cup A) \leq \lambda(B)$ or $\lambda(B \setminus A) \leq \lambda(B)$.
\end{lemma}
\begin{proof}
    Assume that $\lambda(B \cup A) > \lambda(B)$. We show that $\lambda(B \setminus A) \leq \lambda(B)$. First, by the submodularity of $\lambda$ we have $\lambda(B \cap A) < \lambda(A)$. Since $A$ is well-linked and $(B \cap A, \compl{B} \cap A)$ is a bipartition of $A$, it follows that $\lambda(\compl{B} \cap A) \geq \lambda(A)$. Then, again due to submodularity, $\lambda(\compl{B} \cup A) \leq \lambda(\compl{B})$, where $\lambda(\compl{B} \cup A) = \lambda(B \setminus A)$ and $\lambda(\compl{B}) = \lambda(B)$ both due to the symmetry of $\lambda$. Thus, $\lambda(B \setminus A) \leq \lambda(B)$.
\end{proof}

When excluding something from $C$, the internal treewidth $\itw(C)$ clearly does not increase.
When adding something to $C$, $\itw(C)$ could increase.
Nevertheless, we observe the internal treewidth increases by at most the smaller boundary size of the two sets.

\begin{lemma}\label{obs:itwmaxmin}
    Let $G$ be a hypergraph, $A,B\subseteq E(G)$. 
    Then $\itw(A\cup B)\leq \max(\itw(A),\itw(B)) + \min(\lambda(A),\lambda(B))$.
\end{lemma}
\begin{proof}
    It follows from the symmetry and submodularity of $\lambda$ that $\lambda(A \setminus B) + \lambda(B \setminus A) \le \lambda(A)+\lambda(B)$.
    Therefore, either $\lambda(A \setminus B) \le \lambda(A)$ or $\lambda(B \setminus A) \le \lambda(B)$.
    Thus, by either replacing $A$ by $A \setminus B$ or $B$ by $B \setminus A$, we can assume without loss of generality that $A$ and $B$ are disjoint.

    When $A$ and $B$ are disjoint, the inequality \[\itw(A \cup B) \le \max(\itw(A), \itw(B)) + |V(A) \cap V(B)| \le \max(\itw(A),\itw(B)) + \min(\lambda(A),\lambda(B))\] clearly holds.
%
\end{proof}

Lastly, the following lemma tells us that the set $C$ cannot cross internal components of a set $\Lc[s]$ arbitrarily.

\begin{lemma}\label{lem:subsetofB}
    Let $G$ be a hypergraph, $A\subseteq E(G)$ an internally connected set, and $B\subseteq E(G)$. If $\bd(B) \cap \inter(A) = \emptyset$ and $A \cap B \neq \emptyset$, then $A\subseteq B$.
\end{lemma}
\begin{proof}
    
    The case that $|A|=1$ is trivial, so assume $|A|>1$.
    Assume for the sake of contradiction that there exists an edge $e$ such that $e\in A$ and $e\notin B$. First, we show that then $\inter(A)\cap\inter(B) = \emptyset$.

    Suppose there exists a vertex $u \in \inter(A) \cap \inter(B)$. Since $A$ is internally connected and $|A|>1$, there exists a vertex $v \in V(e) \setminus \bd(A)$. Then, $v \in \inter(A)$ and we remark that $v \notin \inter(B)$, since $e \notin B$. As $A$ is internally connected, there is a path between $u$ and $v$ in $\Pc(G)[\inter(A)]$. Since $u\in \inter(B)$ and $v\notin \inter(B)$ there is a vertex $z$ on the path such that $z\in \bd(B)$. This is a contradiction to $\bd(B) \cap \inter(A) = \emptyset$.

    Thus, we have $\inter(A) \cap \inter(B) = \emptyset$.
    Let $e'\in A\cap B$ and $v\in V(e')$ with $v'\in\inter(A)$. Note that $v'\notin \bd(B)$, as otherwise $\bd(B) \cap \inter(A)\neq \emptyset$. But then $v'\in \inter(B)$, a contradiction.
\end{proof}

Now we are ready to finish the proof of \Cref{lem:many_torso_vertices}, which we re-state here.

\ManyTorsoVertices*
\begin{proof}
Let $\cPdProt' = \cPdProt'(H,\cTwMod)$ be the constant $\cPdProt(H,\cTwMod)$ of \Cref{lem:before_uncrossing}.
Here, we set $\cNewSbdProt = \cNewSbdProt(H,\cTwMod) \coloneq \max(\cTwMod + \cPdProt'+1, 2 \cdot \cPdProt' + (\cPdProt' + 1)^2 + 1)$.
Furthermore, we set $\cTorsoIntersect \coloneq \cSbdProt \cdot (\cPdProt' + 1) + \cSbdProt \cdot 2^{2\cNewSbdProt+2} \geq 3\cPdProt' + 1$.
By \Cref{lem:before_uncrossing}, there exists $\cTorsoVertices = \cTorsoVertices(H,\cTwMod,\cSbdProt,\cTorsoIntersect)$, so that for every $k$, the following holds:
Let $G$ be an $H$-topological-minor-free graph with a treewidth-$\cTwMod$-modulator of size $k$, and $\Tc = (T,\Lc)$ be a downwards well-linked superbranch decomposition of $\Hc(G)$ with $\adhsize(\Tc) = \cSbdProt$ and root $r$, where $|V(\torso(r))| \geq \cTorsoVertices \cdot k$. Then, there exists a set $C \subseteq E(\Hc(G))$ satisfying \Cref{cond:size,cond:boundary,cond:internal_treewidth,cond:high_tw_subtree}, i.e.,
\begin{enumerate}
    \item $\cTorsoIntersect \leq |\inter(C) \cap V(\torso(r))| \leq 3 \cdot \cTorsoIntersect$,
    \item $\lambda(C) \leq \cPdProt' + 1$,
    \item $\itw(C) \leq \cPdProt'$, and
    \item for every root-child $s \in \chd(r)$ with $\itw(\Lc[s]) > \cTwMod$, we have $\inter(C) \cap \bd(\Lc[s]) = \emptyset$.
\end{enumerate}

For each child $s$ of $r$ that does not satisfy $\Lc[s] \subseteq C$ or $\Lc[s] \cap C = \emptyset$, we will modify $C$ by either including or excluding $\Lc[s]$ completely from $C$.
Throughout this uncrossing process, we change the internal treewidth $\itw(C)$ and $|\inter(C) \cap V(\torso(r))|$ slightly so that we end up with
\begin{enumerate}
    \item\label{cond2:size} $\cTorsoIntersect - \cSbdProt \cdot (\cPdProt' + 1) = \cSbdProt \cdot 2^{2\cNewSbdProt+2} \leq |\inter(C) \cap V(\torso(r))| \leq 3 \cdot \cTorsoIntersect + \cSbdProt \cdot (\cPdProt' + 1) + \cPdProt' + 1 \leq \OO_{H,\cTwMod}(\cSbdProt)$, 
    \item $\lambda(C)\leq \cPdProt' + 1 \leq \cNewSbdProt$,
    \item $\itw(C)\leq \max(\cTwMod + \cPdProt'+1, 2 \cdot \cPdProt' + (\cPdProt' + 1)^2 + 1) \leq \cNewSbdProt$, and
    \item for each root-child $s \in \chd(r)$ either $\Lc[s] \subseteq C$ or $\Lc[s] \cap C = \emptyset$.
\end{enumerate}
Then, since each child  $s$ of $r$ satisfies $|V(\Lc[s]) \cap V(\torso(r))|=\lambda(\Lc[s])\leq \adhsize(\Tc) = \cSbdProt$, it follows from \Cref{cond2:size} that $C$ must contain the set $\Lc[s]$ for at least $2^{2\cNewSbdProt+2}$ children $s$ of $r$. Then, our desired set $B$ will be the set of hyperedges corresponding to exactly those children, i.e., we choose $B$ so that $B \expand \Tc = C$. Note that $|V(B)| = |V(C) \cap V(\torso(r))| \leq |\inter(C) \cap V(\torso(r))| + \lambda(C) \leq \OO_{H,\cTwMod}(\cSbdProt)$.

We then describe how to do the uncrossing while achieving the requirements.
We say that a child $s$ of $r$ is \emph{crossing} if neither $\Lc[s] \subseteq C$ nor $\Lc[s] \cap C = \emptyset$ holds.
By \emph{uncrossing} $s$, we mean setting either $C \coloneq C \cup \Lc[s]$ or $C \coloneq C \setminus \Lc[s]$.
Because the sets $\Lc[s]$ are pairwise disjoint, uncrossing $s$ does not affect whether other children $s'$ are crossing or not.
Therefore, we can analyze the process of uncrossing all crossing children $s$ by uncrossing them one-by-one, in an arbitrary order.
This uncrossing is done in three steps.

\paragraph{Step 1.} We begin with the crossing children $s$ with $\itw(\Lc[s]) > \cTwMod$.
For each such child, we exclude $\Lc[s]$ from $C$, that is, we set $C:=C\setminus \Lc[s]$.
Clearly $\inter(C)$ does not grow, and we still have $\itw(C)\leq \cPdProt'$. Recall that $V(\Lc[s]) \cap V(\torso(r))=\bd(\Lc[s])$ and  $\bd(\Lc[s]) \cap \inter(C) = \emptyset$ by \Cref{cond:high_tw_subtree}. Therefore, $\Lc[s]$ contains no hyperedge incident to a vertex in $\inter(C) \cap V(\torso(r))$, and thus, $|\inter(C) \cap V(\torso(r))|$ is preserved. Lastly, we will show that $\bd(C \setminus \Lc[s]) \subseteq \bd(C)$, implying that $\lambda(C \setminus \Lc[s]) \leq \lambda(C)$. For this, consider a vertex $v \in \bd(C \setminus \Lc[s])$ and suppose $v \notin \bd(C)$. Then $v \in \inter(C) \cap \bd(\Lc[s])$, which contradicts \Cref{cond:high_tw_subtree}. Thus, $\lambda(C)$ is preserved.


\begin{figure}[!t]
    \centering
    \includegraphics[page=1]{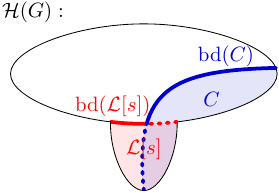}
    \caption{Step 2 of the uncrossing: We want to uncross $\Lc[s]$, where $s$ is a root-child with $\inter(\Lc[s]) \cap \bd(C) \neq \emptyset$. The boundary of $C \setminus \Lc[s]$ is the boundary of $C$ (blue) without the part that lies in $\Lc[s]$ (blue, dashed), that is, $\bd(C) \cap \inter(\Lc[s])$, combined with the boundary of $\Lc[s]$ that lies in $C$ (red, dashed), that is $\inter(C) \cap \bd(\Lc[s])$.}
    \label{fig:step2}
\end{figure}

\paragraph{Step 2.} We continue with the children $s$ of $r$ with $\inter(\Lc[s]) \cap \bd(C) \neq \emptyset$ (see \Cref{fig:step2}). Since for two different root children $s_1, s_2 \in \chd(r)$ the sets $\inter(\Lc[s_1])$ and $\inter(\Lc[s_2])$ are disjoint, we have at most $\lambda(C)$ such children. By \Cref{lem:uncrossing} we can either include or exclude $\Lc[s]$ from $C$ without increasing the boundary.
We show that, when applying \Cref{lem:uncrossing} to $\Lc[s]$ for one fixed child, $\itw(C)$ changes by at most $\lambda(C)$, and $|\inter(C) \cap V(\torso(r))|$ changes by at most $\cSbdProt$.
Since there are at most $\lambda(C) \leq \cPdProt' + 1$ such children, we have
$\cTorsoIntersect - \cSbdProt \cdot (\cPdProt' + 1) \leq |\inter(C) \cap V(\torso(r))| \leq 3 \cdot \cTorsoIntersect + \cSbdProt \cdot (\cPdProt' + 1)$ and $\itw(C) \leq \cPdProt' + (\cPdProt' + 1)^2$ after this step.

First, assume that we are in the first case of \Cref{lem:uncrossing}, i.e., $\lambda(C \cup \Lc[s]) \leq \lambda(C)$, and thus, we include $\Lc[s]$, i.e., we set $C \coloneq C \cup \Lc [t]$. By \Cref{obs:itwmaxmin}, we have $\itw(C \cup \Lc[t])\leq \max\{\itw(\Lc[t]), \itw(C)\} + \lambda(C)$. Furthermore, $|\inter(C)\cap V(\torso(r))|$ is clearly not reduced, but it could increase by at most $|V(\Lc[s]) \cap V(\torso(r))| \leq \cSbdProt$.

Now, assume that we exclude $\Lc[s]$, i.e., we set $C:=C\setminus \Lc [s]$. In this case, clearly we have $\itw(C \setminus \Lc[s]) \leq \itw(C)$, but $|\inter(C) \cap V(\torso(r))|$ might decrease.
The decrease is again bounded by $|V(\Lc[s]) \cap V(\torso(r))| \leq \cSbdProt$.

\begin{figure}[!t]
    \centering
    \begin{subfigure}[t]{0.32\textwidth}
        \includegraphics[page=2]{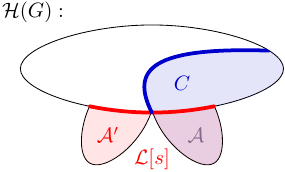}
        \caption{This case is not possible: $(\mc A, \mc A')$ is a bipartition of $\Lc[s]$ that contradicts the well-linkedness of $\Lc[s]$ ($\lambda(\mc A), \lambda(\mc A') < \lambda(\Lc[s])$).}
        \label{subfig:step3_impossible}
    \end{subfigure}
    \hfill
    \begin{subfigure}[t]{0.32\textwidth}
        \includegraphics[page=4]{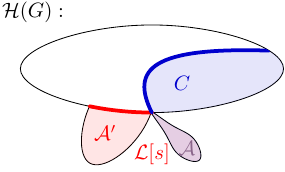}
        \caption{If $\bd(\mc A) \subseteq \bd(\mc A') = \bd(\Lc[s])$, we exclude $\Lc[s]$ from $C$. In this case, $\lambda(C)$ and $|\inter(C) \cap V(\torso(r))|$ remain unchanged and $\itw(C)$ can only decrease.}
        \label{subfig:step3_exclude}
    \end{subfigure}
    \hfill
    \begin{subfigure}[t]{0.32\textwidth}
        \includegraphics[page=3]{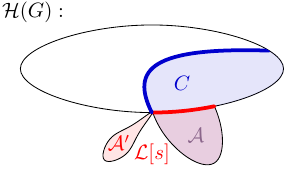}
        \caption{If $\bd(\mc A') \subseteq \bd(\mc A) = \bd(\Lc[s])$, we include $\Lc[s]$ in $C$. In this case, vertices from $\bd(C)$ might become internal, which increases $|\inter(C) \cap V(\torso(r))|$ and $\itw(C)$.}
        \label{subfig:step3_include}
    \end{subfigure}
    \caption{Step 3 of the uncrossing: We want to uncross $\Lc[s]$, where $s$ is a root-child with $\inter(\Lc[s]) \cap \bd(C) = \emptyset$.
    $\mc A$ is the union of the internal components of $\Lc[s]$ that intersect $C$ (which implies $\mc A \subseteq C$ by \Cref{lem:subsetofB}), $\mc A'$ is the union of the remaining components. \Cref{subfig:step3_impossible} visualizes, why we always have $\bd(\mc A) \subseteq \bd(\mc A') = \bd(\Lc[s])$ or $\bd(\mc A') \subseteq \bd(\mc A) = \bd(\Lc[s])$. In the first case, where $\bd(\mc A) \subseteq \bd(\mc A')$ as visualized in \Cref{subfig:step3_exclude}, we exclude $\Lc[s]$ from $C$, while in the latter case, where $\bd(\mc A') \subseteq \bd(\mc A)$ as visualized in \Cref{subfig:step3_include}, we include $\Lc[s]$ in $C$.}
    \label{fig:step3}
\end{figure}

\paragraph{Step 3.} 
For this last step, we are left with the crossing children $s$ with $\inter(\Lc[s]) \cap \bd(C) = \emptyset$ and $\itw(\Lc[s])\leq \cTwMod$ (see \Cref{fig:step3}).
For a fixed crossing child $s$, let $A_1,\dots,A_\ell$ denote the internal components of $\Lc[s]$.
Let $\mc A = \bigcup_{i \colon A_i\cap C\neq\emptyset} A_i$ be the union of the internal components intersecting $C$ and $\mc A' = \Lc[s] \setminus \mc A$ the union of the remaining internal components.

By \Cref{lem:subsetofB} we have $A_i \subseteq C$ for every $A_i$ in $\mc A$, and thus $\mc A \subseteq C$.
This implies that $\mc A'$ is non-empty, as otherwise $\Lc[s] \subseteq C$.
By the well-linkedness of $\Lc[s]$, we either have $\bd(\mc A')\subseteq \bd(\mc A) = \bd(\Lc[s])$ or $\bd(\mc A) \subseteq \mc \bd(\mc A') = \bd(\Lc[s])$.
Otherwise, $(\mc A, \mc A')$ would be a bipartition of $\Lc[s]$ with both $\bd(\mc A), \bd(\mc A') \subsetneq \bd(\Lc[s])$, and thus $\lambda(\mc A), \lambda(\mc A') < \lambda(\Lc[s])$ (see \Cref{subfig:step3_impossible}).

If $\bd(\mc A) \subseteq \bd(\mc A')$ (see \Cref{subfig:step3_exclude}), we set $C \coloneq C \setminus \Lc[s] = C \setminus \mc A$.
In this case, every vertex $v \in \bd(\mc A)$ is incident to a hyperedge in $\mc A \subseteq C$ and to a hyperedge in $\mc A'$, where $\mc A' \cap C = \emptyset$, so $v \in \bd(C)$, and hence we have $\bd(\mc A)\subseteq \bd(C)$. 
Therefore, excluding $\mc A$ cannot change the boundary of $C$. 
Furthermore, $|\inter(C) \cap V(\torso(r))|$ and $\itw(C)$ are clearly preserved.

If $\bd(\mc A')\subseteq \bd(\mc A)$ (see \Cref{subfig:step3_include}), we set $C \coloneq C \cup \Lc[s] = C \cup \mc A'$.
With similar arguments as before, we have $\bd(\mc A')\subseteq \bd(C)$, so including $\Lc[s]$ cannot add new vertices to the boundary.
However, some boundary vertices can become internal, which can increase both $|\inter(C) \cap V(\torso(r))|$ and $\itw(C)$.
Specifically, let $D \subseteq \bd(C)$ be the boundary vertices that become internal when we add $\mc A'$ to $C$, and note that $D \subseteq \bd(\mc A') \subseteq \bd(\Lc[s])$.
We observe that $|\inter(C) \cap V(\torso(r))|$ increases by exactly $|D|$.
Furthermore, we observe that 
\[\itw(C \cup \Lc[s]) \le \max(\itw(C), \itw(\Lc[s])) + |D| \le \max(\itw(C), \cTwMod) + |D|.\]

To analyze the total increase of $\itw(C)$ and $|\inter(C) \cap V(\torso(r))|$ throughout consecutive uncrossings in step 3, we recall that $\lambda(C)$ never increases, and $|D|$ is bounded by the decrease of $\lambda(C)$.
Therefore, the sum of $|D|$ over all uncrossings is at most the initial value of $\lambda(C)$, i.e., at most $\cPdProt' + 1$.
It follows that, after all uncrossings in step 3, we have $\lambda(C) \leq \cPdProt' + 1$, $\itw(C) \leq \max(\cTwMod + \cPdProt'+1, 2 \cdot \cPdProt' + (\cPdProt' + 1)^2 + 1)$, and $\cTorsoIntersect - \cSbdProt \cdot (\cPdProt' + 1) \leq |\inter(C) \cap V(\torso(r))| \leq 3 \cdot \cTorsoIntersect + \cSbdProt \cdot (\cPdProt' + 1) + \cPdProt' + 1$, as promised.
\end{proof}

\subsection{Proof of \Cref{mergeable-children}}\label{sec:unique_or_duplicates}
In this section, we deal with the case where there are not many vertices in $\torso(r)$, but $\Delta(r)$ is still high, which will then complete the proof of \Cref{mergeable-children}. As we show in the following \Cref{torso-vertex-bound-if-unique-edges-topological}, few vertices but many edges in $\torso(r)$ implies that there is an edge with high multiplicity. Later, we argue that we can pick such a high-multiplicity edge in a way so that the set $B$ for \Cref{mergeable-children} with the required properties can be formed by taking only copies of this edge.

We first recall that $H$-topological-minor-free graphs have bounded average degree.

\begin{lemma}[\cite{Mader_1967}]
\label{lem:topminorfreesparse}
If $G$ is a $H$-topological-minor-free graph, then $|E(G)| \le \OO_H(|V(G)|)$.
\end{lemma}

It is an easy consequence of \Cref{lem:topminorfreesparse} that $H$-topological-minor-free graphs have a linear number of cliques.

\begin{lemma}
\label{lem:topminorfreecliques}
If $G$ is a $H$-topological-minor-free graph, then $G$ contains at most $\OO_H(|V(G)|)$ cliques.
\end{lemma}
\begin{proof}
By \Cref{lem:topminorfreesparse}, let $c_H \le \OO_H(1)$ be a constant so that $H$-topological-minor-free graphs have at most $c_H \cdot |V(G)|$ edges.
We prove that $H$-topological-minor-free graphs have at most $2^{2 \cdot c_H} \cdot |V(G)|$ cliques by induction on $|V(G)|$.
This obviously holds when $|V(G)| = 0$.
When $|V(G)| \ge 1$, there exists a vertex $v \in V(G)$ of degree $\le 2 \cdot c_H$.
There are at most $2^{2 \cdot c_H}$ cliques containing $v$.
By induction, $G \setminus \{v\}$ contains at most $2^{2 \cdot c_H} \cdot (|V(G)|-1)$ cliques, so in total there are at most $2^{2 \cdot c_H} \cdot |V(G)|$ cliques.
\end{proof}

We need the following lemma to construct topological minors based on a torso of a node of a superbranch decomposition.

\begin{lemma}\label{lem:outside_paths}
    Let $G$ be a graph, $\Tc = (T,\Lc)$ be a downwards well-linked superbranch decomposition of $\Hc(G)$, and $t \in V(T)$ a non-root node. For every two vertices $u,v \in \bd(\Lc[t])$, there exists a path between $u$ and $v$ whose all edges are in $\Lc[t]$ and intermediate vertices in $\inter(\Lc[t])$.
\end{lemma}
\begin{proof}
    Let $B_1 = \bd(\Lc[t]) \setminus \{u\}$ and $B_2 = \bd(\Lc[t]) \setminus \{v\}$. Due to the well-linkedness of $\Lc[t]$, there are $|B_1| = |B_2|$ vertex-disjoint paths between $B_1$ and $B_2$ using only edges in $\Lc[t]$. The vertices in $B_1 \cap B_2 = \bd(\Lc[t]) \setminus \{u,v\}$ are forced to connect to themselves by one-vertex paths, so $u$ must connect to $v$ avoiding vertices in $B_1 \cap B_2$.
\end{proof}

We call such a path from \Cref{lem:outside_paths} a \emph{$t$-outside path}. We then apply \Cref{lem:topminorfreesparse,lem:topminorfreecliques,lem:outside_paths} to bound the number of hyperedges in multiplicity-one subsets of $E(\torso(r))$ by the number of their vertices.

\begin{lemma}\label{torso-vertex-bound-if-unique-edges-topological}
    Let $H$ be a graph, $G$ an $H$-topological-minor-free graph, and $\Tc = (T,\Lc)$ a downwards well-linked superbranch decomposition of $\Hc(G)$ with root $r$.
    Let also $E \subseteq E(\torso(r))$ be a set of hyperedges with multiplicity one.
    It holds that $|E| \leq \OO_H(|V(E)|)$.
\end{lemma}
\begin{proof}


    We say that a subset $R \subseteq E$ is \emph{routable} if there exists an injective function $g \colon R \to \binom{V(E)}{2}$ that maps each hyperedge $e \in R$ to a unique pair of distinct vertices $g(e) = \{u_e, v_e\} \in \binom{V(e)}{2}$.

    \begin{claim}
    \label{torso-vertex-bound-if-unique-edges-topological:claim1}
    Any routable set $R$ has size at most $|R| \le \OO_H(|V(E)|)$.
    \end{claim}
    \begin{claimproof}
    Denote $R = \{e_1, \ldots, e_p\}$, and let $S = \{s_1, \ldots, s_p\} \subseteq \chd(r)$ be the corresponding set of children of $r$ in $T$.
    Let $g \colon R \to \binom{V(E)}{2}$ be the injective function, and denote $g(e_i) = \{u_i, v_i\}$.
    By \Cref{lem:outside_paths}, for each $i \in [p]$ there exists an $s_i$-outside path $P_i$ in $G$ between $u_i$ and $v_i$.
    For different $i \neq j$, the paths $P_i$ and $P_j$ can only intersect in their endpoints.
    We obtain a graph $G'$ by taking the vertex set $V(E)$, and adding an edge between $u_i$ and $v_i$ for each $i \in [p]$.
    We observe that $G'$ is a topological minor of $G$, and therefore $H$-topological-minor-free.
    By \Cref{lem:topminorfreesparse}, $G'$ has at most $\OO_H(|V(E)|)$ edges, and therefore $|R| \le \OO_H(|V(E)|)$.
    \end{claimproof}

    It remains to show that if $|E|$ would be too large compared to $|V(E)|$, then we could construct a large routable set.

    \begin{claim}
    \label{torso-vertex-bound-if-unique-edges-topological:claim2}
    There exists a routable set $R$ of size $|R| \ge |E| - \OO_H(|V(E)|)$
    \end{claim}
    \begin{claimproof}
    Let $R \subseteq E$ be a maximum-size routable set and $g \colon R \to \binom{V(E)}{2}$ the corresponding function.
    As in \Cref{torso-vertex-bound-if-unique-edges-topological:claim1}, we define a graph $G'$ with the vertex set $V(E)$ and the edge set consisting of the pairs of vertices $\{g(e) \mid e \in R\}$.
    As before, we observe that $G'$ is a topological minor of $G$, and therefore $G'$ is $H$-topological-minor-free.
    By \Cref{lem:topminorfreecliques}, $G'$ has at most $\OO_H(|V(E)|)$ cliques.
    
    For each hyperedge $e \in E \setminus R$ and for each pair $u,v \in V(e)$ of distinct vertices, there must exist $e' \in R$ with $g(e') = \{u,v\}$, as otherwise we could increase $|R|$ by adding $e$ to it.
    Therefore, for each $e \in E \setminus R$ the set $V(e)$ is a clique in $G'$.
    Because $E$ has multiplicity one, the cliques are distinct, and therefore $|E \setminus R| \le \OO_H(|V(E)|)$.
    \end{claimproof}

    By combining the conclusions of \Cref{torso-vertex-bound-if-unique-edges-topological:claim1,torso-vertex-bound-if-unique-edges-topological:claim2}, we deduce $|E| \le \OO_H(|V(E)|)$.
\end{proof}

\Cref{torso-vertex-bound-if-unique-edges-topological} essentially tells us that few vertices and many edges in $\torso(r)$ imply high multiplicity. Next, we show that if there is an edge with high multiplicity, then there is also a low-rank edge with high multiplicity. Then, we can later find the set $B$ within the copies of this edge.

\begin{lemma}\label{lem:large_boundary_duplicates}
    Let $H$ be a graph on $h$ vertices and let $G$ be an $H$-topological-minor-free graph together with a downwards well-linked superbranch decomposition $\Tc = (T,\Lc)$ of $\Hc(G)$ rooted at $r$. For every hyperedge $e \in E(\torso(r))$ with rank $|V(e)| \geq h$, the multiplicity of $e$ is less than $\binom{h}{2}$.
\end{lemma}
\begin{proof}
    Let $e \in E(\torso(r))$ be a hyperedge with $|V(e)| \geq h$.
    Suppose that there are $\ell = \binom{h}{2}$ children $s_1,\dots,s_{\ell} \in \chd(r)$ with $\bd(\Lc[s_1]) = \dots = \bd(\Lc[s_{\ell}]) = V(e)$.
    Let $S \coloneq \{s_1,\dots,s_{\ell}\}$.
    We show that $G$ contains the complete graph $K_h$ as a topological minor, which is a contradiction to the $H$-topological-minor-freeness.

    To that end, we fix a subset $V_h \subseteq V(e)$ of size exactly $h$ together with an arbitrary bijective mapping $g \colon S \to \binom{V_h}{2}$.
    For $1 \leq i \leq \ell$, let $g(s_i) = \{u_i, v_i\} \subseteq V_h \subseteq \bd(\Lc[s_i])$.
    We remark that by \Cref{lem:outside_paths}, there is an $s_i$-outside path between $u_i$ and $v_i$ for every $1 \leq i \leq \ell$.
    Then, we contract these $s_i$-outside paths until there is only one edge left of each and obtain the complete graph $K_h$.
    Since the $s_i$-outside paths are vertex-disjoint, except for their endpoints, this $K_h$ is a topological minor of $G$, which is a contradiction.
\end{proof}

We are now ready to put together an almost-final version of \Cref{mergeable-children} as the following \Cref{lem:mergeable_children_without_internal_components}.
It lacks only one condition compared to \Cref{mergeable-children}, which is then immediately added in the following final proof of \Cref{mergeable-children}.

In the case when there are many vertices in $\torso(r)$, the proof of \Cref{lem:mergeable_children_without_internal_components} follows by combining \Cref{lem:many_torso_vertices} with \Cref{torso-vertex-bound-if-unique-edges-topological,lem:large_boundary_duplicates}.
Otherwise, we can directly use the results of \Cref{torso-vertex-bound-if-unique-edges-topological,lem:large_boundary_duplicates} to obtain the desired result in the form of a low-rank high-multiplicity hyperedge of the torso.



\begin{lemma}\label{lem:mergeable_children_without_internal_components}
    Let $H$ be a graph and $\cTwMod$, $\cSbdProt$ integers. There are integers $\cNewSbdProt \coloneq \cNewSbdProt(H,\cTwMod)$ and $\cSbdRootDeg \coloneq \cSbdRootDeg(H,\cTwMod,\cSbdProt)$ such that for every $k$ the following holds:
    Let $G$ be an $H$-topological-minor-free graph that has a treewidth-$\cTwMod$-modulator of size $k$.
    Let $\Tc = (T, \Lc)$ be a downwards well-linked superbranch decomposition of $\Hc(G)$ rooted at $r$ and with adhesion size $\cSbdProt$. 
    If $\Delta(r) \geq \cSbdRootDeg \cdot k$, then there exists a subset of hyperedges $B \subseteq E(\torso(r))$ such that
    \begin{enumerate}
        \item $\lambda(B) \leq \cNewSbdProt$,
        \item $\itw(B \expand \Tc) \leq \cNewSbdProt$, and
        \item $2^{2\cNewSbdProt+2} \leq |B| \leq \OO_{H,\cTwMod}(\cSbdProt)$.
    \end{enumerate}
    Furthermore, $\cNewSbdProt$ is computable when given $H$ and $\cTwMod$, and $\cSbdRootDeg$ is computable when given $H$, $\cTwMod$, and $\cSbdProt$.
\end{lemma}
\begin{proof}
    Let $H$ be a graph and $\cTwMod,\cSbdProt,k$ integers.
    Let $\cNewSbdProt = \cNewSbdProt(H,\cTwMod)$ be the maximum of the integer $\cNewSbdProt(H,\cTwMod)$ of \Cref{lem:many_torso_vertices} and $\cTwMod + |V(H)|$.
    Let also $\cTorsoVertices$ be the integer $\cTorsoVertices(H,\cTwMod,\cSbdProt)$ from \Cref{lem:many_torso_vertices}.
    Furthermore, let $c_H \le \OO_H(1)$ be the constant from \Cref{torso-vertex-bound-if-unique-edges-topological} so that for every $E \subseteq E(\torso(r))$ with multiplicity one it holds that $|E| \leq c_H \cdot |V(E)|$.

    We set $\cSbdRootDeg \coloneq c_H \cdot \cTorsoVertices \cdot \left( \binom{|V(H)|}{2} + 2^{2\cNewSbdProt+2} \right) + 1$. Let $G$ be an $H$-topological-minor-free graph that has a treewidth-$\cTwMod$-modulator of size $k$. Let $\Tc = (T,\Lc)$ be a downwards well-linked superbranch decomposition of $\Hc(G)$ with adhesion size $\cSbdProt$ and root $r$, where $\Delta(r) = |E(\torso(r))| \geq \cSbdRootDeg \cdot k = c_H \cdot \cTorsoVertices \cdot \binom{|V(H)|}{2} \cdot k + c_H \cdot \cTorsoVertices \cdot 2^{2\cNewSbdProt + 2} \cdot k + k$. We distinguish the two cases that $|V(\torso(r))| \geq \cTorsoVertices \cdot k$ or $|V(\torso(r))| < \cTorsoVertices \cdot k$.

    \paragraph{Many torso vertices.}
    In the first case, by \Cref{lem:many_torso_vertices}, there exists a set $B' \subseteq E(\torso(r))$ of hyperedges such that $\lambda(B') \leq \cNewSbdProt$, $\itw(B' \expand \Tc) \leq \cNewSbdProt$, $|B'| \geq 2^{2\cNewSbdProt + 2}$, and $|V(B')| \leq \OO_{H,\cTwMod}(\cSbdProt)$.
    If $|B'| \leq |V(B')| \cdot c_H \cdot \max\{\binom{|V(H)|}{2}, 2^{2 \cNewSbdProt + 2}\} \le \OO_{H,\cTwMod}(\cSbdProt)$, we can take $B'$ as $B$ and are done.
    Otherwise, by the pigeonhole principle, there is a hyperedge $e \in B'$ with multiplicity at least $\max\{\binom{|V(H)|}{2},2^{2\cNewSbdProt+2}\}$ in $B'$.
    By \Cref{lem:large_boundary_duplicates}, $|V(e)| < |V(H)|$.
    We choose $B \subseteq \{e' \in B' \mid V(e') = V(e)\}$ arbitrarily such that $|B| = 2^{2\cNewSbdProt+2}$.
    Then, since $B' \subseteq B$, we have $\itw(B' \expand \Tc) \leq \itw(B \expand \Tc) \leq \cNewSbdProt$, and $\lambda(B') \leq |V(B')| = |V(e)| < |V(H)| \le \cNewSbdProt$.

    \paragraph{Few torso vertices.}
    From now on, assume that $|V(\torso(r))| < \cTorsoVertices \cdot k$.
    Since for every $E \in E(\torso(r))$ with multiplicity one, we have $|E| \leq c_H \cdot |V(E)| \leq c_H \cdot |V(\torso(r))| < c_H \cdot \cTorsoVertices \cdot k$, it follows that there are fewer than $c_H \cdot \cTorsoVertices \cdot k$ hyperedges with pairwise different vertex sets in $E(\torso(r))$. Thus, by \Cref{lem:large_boundary_duplicates}, there are fewer than $\binom{|V(H)|}{2} \cdot c_H \cdot \cTorsoVertices \cdot k$ edges of rank at least $|V(H)|$ and hence more than $c_H \cdot \cTorsoVertices \cdot 2^{2\cNewSbdProt + 2} \cdot k + k$ edges of rank smaller than $|V(H)|$ in $E(\torso(r))$. Next, we argue that of all these low-rank edges, there are more than $c_H \cdot \cTorsoVertices \cdot 2^{2\cNewSbdProt + 2} \cdot k$ low-rank edges $e \in E(\torso(r))$ with $\itw(\{e\} \expand \Tc) \leq \cTwMod$.

    For this, let $X$ be a treewidth-$\cTwMod$-modulator of $G$ of size $k$, i.e., with $\tw(G-X) \leq \cTwMod$.
    Therefore, for every hyperedge $e \in E(\torso(r))$, in order to have $\itw(\{e\} \expand \Tc) > \cTwMod$, we need $\inter(\{e\} \expand \Tc) \cap X \neq \emptyset$.
    Because for two distinct hyperedges $e_1,e_2 \in E(\torso(r))$, we have $\inter(\{e_1\} \expand \Tc) \cap \inter(\{e_2\} \expand \Tc) = \emptyset$, there are at most $|X| = k$ hyperedges $e \in E(\torso(r))$ with $\inter(\{e\} \expand \Tc) \cap X \neq \emptyset$ and thus with $\itw(\{e\} \expand \Tc) > \cTwMod$.

    It follows that there are more than $c_H \cdot \cTorsoVertices \cdot 2^{2\cNewSbdProt + 2} \cdot k$ hyperedges $e \in E(\torso(r))$ with rank $|V(e)| < |V(H)|$ and $\itw(\{e\} \expand \Tc) \leq \cTwMod$. Since there are fewer than $c_H \cdot \cTorsoVertices \cdot k$ hyperedges with pairwise different vertex sets in $E(\torso(r))$, by the pigeonhole principle, there is a set $B \subseteq E(\torso(r))$ of size $|B| = 2^{2\cNewSbdProt+2}$ such that for all hyperedges $e \in B$, we have $V(e) = V(B)$ with $|V(B)| < |V(H)|$ and $\itw(\{e\} \expand \Tc)$. It follows that $\bd(B) \subseteq V(B)$, so $\lambda(B) < |V(H)| \leq \cNewSbdProt$ and $\itw(B \expand \Tc) \leq \max\limits_{e \in B}\{\itw(\{e\} \expand \Tc)\} + \lambda(B) < \cTwMod + |V(H)| \leq \cNewSbdProt$.
\end{proof}

The only remaining condition to prove \Cref{mergeable-children} is \Cref{ite:internal_components}, i.e., that for every internal component $B'$ of the desired set $B$, we have $\bd(B') = \bd(B)$.
We ensure this condition by taking a subset of the set $B$ from \Cref{lem:mergeable_children_without_internal_components} consisting of internal components with the same boundary.
Let us now re-state \Cref{mergeable-children} and finish its proof.

\MergeableChildren*
\begin{proof}
    Let $H$ be a graph and $\cTwMod$, $\cSbdProt$, $k$ integers.
    Let $\cSbdRootDeg, \cNewSbdProt$ be the integers from \Cref{lem:mergeable_children_without_internal_components} such that the following holds:
    Let $G$ be an $H$-topological-minor-free graph that has a treewidth-$\cTwMod$-modulator of size $k$.
    Let $\Tc = (T,\Lc)$ be a downwards well-linked superbranch decomposition of $\Hc(G)$ rooted at $r$ and with adhesion size $\cSbdProt$.
    By \Cref{lem:mergeable_children_without_internal_components}, there exists a set $C \subseteq E(\torso(r))$ with $\lambda(C), \itw(C \expand \Tc) \leq \cNewSbdProt$ and $2^{2\cNewSbdProt + 2} \leq |C| \leq \OO_{H,\cTwMod}(\cSbdProt)$.
    We consider the internal components of $C$.
    First, we observe that if $C'$ is an internal component of $C$, then $\bd(C') \subseteq \bd(C)$.
    
    Now, we group the internal components of $C$ by their boundary. Since each such boundary is a subset of $\bd(C)$, there are at most $2^{\lambda(C)} \leq 2^{\cNewSbdProt}$ such groups. Thus, by the pigeonhole principle, there exists at least one group such that the union over all internal components in this group contains at least $|C| / 2^{\cNewSbdProt} \geq 2^{\cNewSbdProt+2}$ hyperedges. Let $B$ be the union of all internal components in this group. Clearly, $B$ satisfies \Cref{ite:size,ite:internal_components}. It is also easy to see that $\itw(B \expand \Tc) \leq \itw(C \expand \Tc) \leq \cNewSbdProt$, since $B \subseteq C$. Furthermore, since each internal component of $B$ has the same boundary, and this boundary is a subset of $\bd(C)$, $\lambda(B) \leq \lambda(C) \leq \cNewSbdProt$.
\end{proof}

%% file: local-search.tex
\section{Dynamic local search}\label{sec:local_search}
In the previous section, we showed that if the degree of the root of our superbranch decomposition becomes too high, then there exists a set of root-children, identified by their corresponding hyperedges in the torso of the root, that we can merge. In detail, we showed that there exists a set $B \subseteq E(\torso(r))$ of bounded size such that $B \expand \Tc$ has small boundary and internal treewidth. Additionally, we can assume that for every internal component $B'$ of $B$ we have $\bd(B') = \bd(B)$, which makes this set $B$ traceable via a local search -- a property that we will exploit in this section.

Here, we construct a data structure that keeps track of all small internally connected sets in the torso of the root, grouped by their respective boundary, so that we can easily find the desired set $B$ within such a group. An important ingredient to this data structure is the following local search algorithm that efficiently finds all small internally connected sets that intersect a given set $I$ of hyperedges. In our data structure, we will then employ this local search algorithm to locally recompute internally connected sets after every update to the torso of the root.

The local search algorithm is based on a similar algorithm given by Korhonen~\cite[Lemma~7.3]{Korhonen_2024}, the main difference being that they did not require an upper bound on the size of the sets $A$ but instead on the multiplicity of the hypergraph $G$. For clarity, we again describe the algorithm here.

\begin{lemma}[Based on {\cite[Lemma 7.3]{Korhonen_2024}}]\label{lem:static_local_search}
    Let $G$ be a hypergraph of rank $r$, whose representation is already stored. There is an algorithm that, given three integers $p,s,k \geq 0$, and two sets $I \subseteq E(G)$ and $X \subseteq V(I)$ of size $1 \le |I| \leq p$ and $|X| \leq s$, in time $p r \cdot s^{\OO(k)}$ lists all sets $A \subseteq E(G)$ so that
    \begin{itemize}
        \item $I \subseteq A$,
        \item $X \subseteq \bd(A)$,
        \item $|A| \leq p$,
        \item $|V(A)| \leq s$,
        \item $\lambda(A) \leq k$, and
        \item every internal component of $A$ intersects $I$.
    \end{itemize}
    The number of such sets is at most $s^{\OO(k)}$.
\end{lemma}
\begin{proof}
    We describe a recursive algorithm.
    In each recursive step, we first compute in time $\OO(|I| \cdot r) = \OO(p \cdot r)$ the set $V(I)$ and $\bd(I)$ using \Cref{lem:compute-boundary}. If $|I| > p$, $|V(I)| > s$, or $|X| > k$, we can immediately return the empty list. Now, we can assume that $|V(I)| \leq s$ and $|I| \leq p$. 
    If $X$ intersects $\inter(I)$, we can again return the empty list. If $X = \bd(I)$, the only set $A$ that satisfies the requirements is $I$ itself, so we can output $I$.
    
    Otherwise, we can assume that $X \subsetneq \bd(I)$. Then, we choose an arbitrary $v \in \bd(I) \setminus X$. For any set that satisfies the requirements, we either have $v \in \bd(A)$ or $v \in \inter(A)$. We now make two recursive calls, one enumerating the sets $A$ of the first type and the other enumerating the sets $A$ of the latter type. For the first type, it suffices to simply add $v$ to $X$. For the second case, we need to add all hyperedges in $\incidences(v)$ to $I$. If $|\incidences(v)| > p$, the recursive call with $I \cup \incidences(v)$ would immediately return, so we can skip the recursion in this case. Otherwise, we recurse on $I \cup \incidences(v)$.

    We observe that each recursive call takes time $\OO(p \cdot r)$.
    Moreover, we show that the recursion tree has size at most $s^{\OO(k)}$. For this, we observe that in a recursive call of the first type, $|X|$ increases, and in a recursive call of the second type, $|\inter(I)|$ increases. As $|X|$ is bounded by $k$ and $|\inter(I)|$ by $s$, we find that the size of the recursion tree is at most $\binom{s+k}{k} \leq s^{\OO(k)}$.
    It follows that the total running time of our algorithm is $p r \cdot s^{\OO(k)}$ and that there are at most $s^{\OO(k)}$ sets in the list.
\end{proof}

Now, we are ready to construct our local search data structure, which we will later use to keep track of the small internally connected sets with small internal treewidth and boundary size in the torso of the root of our superbranch decomposition. To abstract away from the internal treewidth requirement, we introduce the concept of an oracle $\oracle$ that, given a set $S$ of hyperedges, depends on $S$ itself and on the boundary $\bd(S)$ and returns true $\true$ or false $\false$ -- later corresponding to whether $S \expand \Tc$ has small internal treewidth or not. As we are only interested in small sets of size at most $s_2$, we restrict the oracle to be \emph{$s_2$-bounded}, i.e., the oracle only needs to decide correctly whether the internal treewidth of $S \expand \Tc$ is small if $|S| \leq s_2$.

Then, given a hypergraph $G$, this data structure supports insertions and deletions of vertices and hyperedges to $G$, and the $\query$ operation allows us to find a set $C$ of bounded size and bounded boundary such that every internal component of $C$ has the same boundary as $C$ and satisfies the oracle $\oracle$. Applied to $\torso(r)$ with an internal treewidth oracle (realized by using the internal treewidth automaton from \Cref{lem:internal_treewidth_automaton}) this data structure returns (if possible) a set $C \subseteq E(\torso(r))$ of hyperedges, corresponding to mergeable subtrees of our superbranch decomposition. Whenever the degree $\Delta(r)$ is too high, the existence of such a set $C$ is guaranteed by \Cref{mergeable-children}, so in this case the data structure will always return such a set $C$ and we can decrease $\Delta(r)$ by merging the corresponding subtrees.

Recall that during this paper we are working with \emph{labeled} hyperedges, i.e., each insertion inserts a hyperedge with a new label. In particular, we do not require the hypergraph $G$ to have multiplicity one, i.e., there might be hyperedges with the same vertex set but different labels in $G$. Then, most importantly, for a set $S \subseteq E(G)$ of labeled hyperedges, the result of the oracle $\oracle(S)$ really depends on the labels of the hyperedges, and not only their vertex sets.

\begin{lemma}\label{lem:dynamic_local_search}
    Let $G$ be a dynamic hypergraph of rank $r$ and let $s_1$, $s_2$, and $k$ be integers with $s_1 \le s_2$. Let $\oracle$ be an $s_2$-bounded oracle that, given a set $S \subseteq E(G)$ of labeled hyperedges of size $|S| \leq s_2$, depends on $S$ and $\bd(S)$ and returns $\true$ or $\false$ in time $\OO_{r,s_2}(\log |E(G)|)$.
    There is a data structure that maintains $G$ and supports the following operations:
    \begin{itemize}
        \item $\init(G)$: Given an empty hypergraph $G$, initialize the data structure. Runs in time $\OO(1)$.
        \item $\addVertex(v)$: Given a vertex $v \notin V(G)$, add $v$ to $G$. Runs in time $\OO(1)$.
        \item $\deleteVertex(v)$: Given an isolated vertex $v \in V(G)$, remove $v$ from $G$. Runs in time $\OO(1)$.
        \item $\addHyperedge(e)$: Given a hyperedge $e \notin E(G)$ of rank at most $r$ with a new label, add $e$ to $G$. Runs in time $\OO_{r,s_2,k}(\log |E(G)|)$.
        \item $\deleteHyperedge(e)$: Given a hyperedge $e \in E(G)$, remove $e$ from $G$. Runs in time $\OO_{r,s_2,k}(\log |E(G)|)$.
        \item $\query$: Returns a set $C \subseteq E(G)$ such that:
        \begin{itemize}
            \item $s_1 / 2 \leq |C| \leq s_2$,
            \item $\lambda(C) \leq k$,
            \item for every internal component $C'$ of $C$, we have $\bd(C') = \bd(C)$, and
            \item for every internal component $C'$ of $C$, we have $\oracle(C') = \true$,
        \end{itemize}
        or reports that no such set of size $s_1 \leq |C| \leq s_2$ exists.
        Runs in time $\OO_{r,s_2,k}(1)$.
    \end{itemize}
\end{lemma}
\begin{proof}
    We define a \emph{chip} of a hypergraph $G$ to be an internally connected set $Z \subseteq E(G)$ with $|Z| \leq s_2$, $\lambda(Z) \leq k$, and $\oracle(Z) = \true$.
    Note that due to the bound on the rank, it also holds that $|V(Z)| \leq r \cdot s_2$.    
    For a set $\mc Z$ of chips, we denote by $\vol(\mc Z) \coloneq \sum\limits_{Z \in \mc Z} |Z|$ the total size over all chips in $\mc Z$. For a set $B \subseteq V(G)$ of size $|B| \leq k$, we denote by $\chips(B)$ the set of all chips with boundary $B$ and by $\vol(B) \coloneq \vol(\chips(B))$. We maintain the following three data structures.
    \begin{itemize}
        \item For every $B \subseteq V(G)$ of size $|B| \leq k$ with $\vol(B) > 0$, a balanced binary search tree $\bbst_B$ containing all chips $Z \in \chips(B)$.
        \item A max-heap containing all sets $B \subseteq V(G)$ of size $|B| \leq k$ with $\vol(B) > 0$ ranked by $\vol(B)$ and pointing to the corresponding search tree $\bbst_B$.
        \item A balanced binary search tree $\bbst^*$ that contains each set $B$ of size $|B| \leq k$ with $\vol(B) > 0$ together with a pointer pointing to the element of the max-heap that corresponds to $B$.
    \end{itemize}

    Before describing how we maintain these three data structures, we first show that the local search algorithm from \Cref{lem:static_local_search} can be used to find all chips of $G$ containing a specific set of hyperedges. This will later be used to locally recompute chips after every update.
    \begin{claim}\label{claim:static_local_search_chips}
        There is an algorithm that, given a set $I \subseteq E(G)$ with $1 \leq |I| \leq s_2$ and a set $X \subseteq I$, enumerates all chips $Z$ with $I \subseteq Z$ and $X \subseteq \bd(Z)$ in time $\OO_{r,s_2,k}(\log |E(G)|)$.
        Moreover, the number of such chips is $\OO_{s_2,k}(1)$.
    \end{claim}
    \begin{claimproof}
        Given two sets $I \subseteq E(G)$ and $X \subseteq I$, in order to find all chips $Z$ with $I \subseteq Z$ and $X \subseteq \bd(Z)$, we call the algorithm from \Cref{lem:static_local_search} with $I$, $X$, and the required bounds from the definition of a chip, and then filter out the sets that are no chips. That is, for every set $A$ enumerated by the algorithm, we check whether it is internally connected and satisfies the oracle.
        For a set $A$, checking whether it is internally connected can be done in $(r+|A|)^{\OO(1)} = \OO_{r,s_2}(1)$ time, and checking whether $\oracle(A) = \true$ can be done in time $\OO_{r,s_2}(\log |E(G)|)$.
        Since the algorithm from \Cref{lem:static_local_search} takes time $s_2 \cdot r \cdot (r \cdot s_2)^{\OO(k)} = (r \cdot s_2)^{\OO(k)}$ and returns at most $s_2^{\OO(k)}$ sets $A$, the total running time is $\OO_{r,s_2,k}(\log |E(G)|)$.
    \end{claimproof}

    To bound the running times for updates to our three data structures, we now use \Cref{claim:static_local_search_chips} to show that there are at most $\OO_{r,s_2,k}(|E(G)|)$ chips in $G$.

     \begin{claim}
        There are at most $\OO_{s_2,k}(|E(G)|)$ chips in $G$.
    \end{claim}
    \begin{claimproof}
        If we call the algorithm from \Cref{claim:static_local_search_chips} with $I = \{e\}$ for some hyperedge $e \in E(G)$ and $X = \emptyset$, we obtain all chips $Z \subseteq E(G)$ with $e \in A$.
        So, repeating this call to the algorithm for every hyperedge $e \in E(G)$ gives us all chips. Since every call returns at most $\OO_{s_2,k}(1)$ sets, there are at most $\OO_{s_2,k}(|E(G)|$ chips in $G$.
    \end{claimproof}

    It follows that each balanced binary search tree $\bbst_B$ contains at most $\OO_{s_2,k}(|E(G)|)$ chips, each of size at most $s_2$, so inserting and deleting takes time at most $\OO_{s_2,k}(\log |E(G)|)$.
    It also holds that there are at most $\OO_{s_2,k}(|E(G)|)$ sets $B \subseteq V(G)$ with $\vol(B) > 0$, so with the help of $\bbst^*$, we can find the max-heap entry corresponding to a given set $B \subseteq V(G)$ in $\OO_{s_2,k}(\log |E(G)|)$ time, and thus inserting or deleting $B$, or updating the rank $\vol(B)$ in the max-heap can also be done in $\OO_{s_2,k}(\log |E(G)|)$ time.

    In the following, we first describe how we maintain these three data structures, i.e., how we implement the operations $\init$, $\addVertex$, $\deleteVertex$, $\addHyperedge$, and $\deleteHyperedge$, before we describe how queries are answered. For the $\addHyperedge$ and $\deleteHyperedge$ operations, let $G$ denote the hypergraph before and $G'$ the hypergraph after the insertion/deletion of the hyperedge. For a set $Z \in E(G) \cap E(G')$, we denote by $\bd_G(Z)$ and $\bd_{G'}(Z)$, $\lambda_G(Z)$ and $\lambda_{G'}(Z)$, and $\inter_G(Z)$ and $\inter_{G'}(Z)$ the boundary, boundary size and set of internal vertices of $Z$ in $G$ or $G'$, respectively. To avoid duplicates, we assume that during the $\addHyperedge$ and $\deleteHyperedge$ operations, every time a chip $Z$ is inserted into its binary search tree $\bbst_{\bd_{G'}(Z)}$, we first check whether it is already there, and if so, skip its insertion.
    
    Without explicitly stating it later in the descriptions of the operations, we assume that whenever we add the first chip to its balanced binary search tree or remove the last one, the search tree is created or removed, and the max-heap and the balanced binary search tree $\bbst^*$ are updated accordingly. That is, whenever we want to add a new chip $Z$ to a balanced binary search tree $\bbst_B$ but $\bbst_{B}$ does not yet exist, we create $\bbst_{B}$ as a new balanced binary search tree with only one entry $Z$. We also compute $\vol(B) = |Z|$, and insert the set $B$ together with a pointer to $\bbst_{B}$ and with rank $\vol(B)$ into the max-heap. Lastly, we add $B$ together with a pointer to the corresponding max-heap entry to $\bbst^*$. As argued before, these operations can be done in $\OO_{r,s_2,k}(\log |E(G)|)$ time. Similarly, when we remove a chip $Z$ from a balanced binary search tree $\bbst_B$, and afterward have $\vol(B) = 0$, we can delete the entire search tree $\bbst_B$ together with the corresponding entries in the max-heap and the balanced binary search tree $\bbst^*$. For this, we look up the pointer to the corresponding max-heap entry in $\bbst^*$. Then, the max-heap entry points to the search tree $\bbst_B$, which we delete. Next, we delete the entry from the max-heap and restore the max-heap property. Lastly, we remove the entry from the $\bbst^*$. This can all be done in $\OO_{r,s_2,k}(\log |E(G)|)$ time.

    
    \paragraph{$\init$:} Given a hypergraph $G$ without vertices and hyperedges, the data structure can clearly be initialized in $\OO(1)$ time.

    \paragraph{$\addVertex(v)$ / $\deleteVertex(v)$:} We remark that inserting or deleting an isolated vertex does neither change anything about the existing chips nor creates any new chips. The only thing that needs to be done is updating the representation of $G$ itself, which clearly can be done in $\OO(1)$ time.

    \paragraph{$\addHyperedge(e)$:} Before we update our three data structures, let us first consider the representation of $G$. Recall that a hypergraph $G$ is stored as a bipartite graph, where the two parts of the bipartition are $V(G)$ and $E(G)$ and there is an edge between $v \in V(G)$ and $e' \in E(G)$ if $v \in V(e')$. Since the new hyperedge $e$ has rank at most $r$, we need to add one vertex and at most $r$ edges to the representation of $G$, which can be done in $\OO(r)$ time.
    
    Now, we describe how the three data structures are updated. The only thing that can happen to an existing chip $Z$ when adding a hyperedge to $G$ is that some internal vertices can become boundary vertices, which increases the boundary size and might destroy the internal connectivity. We say a chip $Z$ of $G$ is \emph{affected} if $\bd_G(Z) \neq \bd_{G'}(Z)$ and remark that an affected chip is not necessarily a chip of $G'$. However, every \emph{unaffected} chip, i.e., every chip of $G$ that is not affected, is also a chip of $G'$. For this, note that $\bd_G(Z) = \bd_{G'}(Z)$ also implies that $\inter_G(Z) = \inter_{G'}(Z)$. Hence, every unaffected chip $Z$ is internally connected in $G'$, $\lambda_{G'}(Z) = \lambda_G(Z) \leq k$, and $\oracle(Z) = \true$ in $G'$, so $Z$ is a chip of $G'$. Therefore, we only need to take care of the affected chips in the following.
    
    For every affected chip $Z$, we need to remove $Z$ from $\bbst_{\bd_G(Z)}$, and if $Z$ is a chip of $G'$, insert it into $\bbst_{\bd_{G'}(Z)}$. To find out, which chips are affected, we first argue that if $\inter_G(Z) \cap V(e) = \emptyset$ for a chip $Z$ of $G$, then $\bd_G(Z) = \bd_{G'}(Z)$, so $Z$ is unaffected.

    \begin{claim}\label{claim:add_edge_boundary}
        Let $Z \subseteq E(G)$ be a set with $\inter_G(Z) \cap V(e) = \emptyset$. Then, $\bd_G(Z) = \bd_{G'}(Z)$.
    \end{claim}
    \begin{claimproof}
        First, we observe that $\bd_G(Z) \subseteq \bd_{G'}(Z)$ as adding a hyperedge cannot decrease the boundary. To show that $\bd_G(Z) \supseteq \bd_{G'}(Z)$, let $v \in \bd_{G'}(Z)$. Then, $v \in V(Z)$, so either $v \in \bd_G(Z)$ or $v \in \inter_G(Z)$. In the first case, we are done. In the second case, there is no hyperedge $e' \in E(G) \setminus Z$ with $v \in V(e')$ and we know that $v \notin V(e)$ since $\inter_G(Z) \cap V(e) = \emptyset$. Thus, there is also no hyperedge $e' \in E(G') \setminus Z$, which contradicts $v \in \bd_{G'}(Z)$. Hence, $\bd_G(Z) = \bd_{G'}(Z)$.
    \end{claimproof}
    
    It follows that every affected chip $Z$ has a vertex $v \in \inter_G(Z) \cap V(e)$. Now, we can determine the affected chips as follows. Note that for every vertex $v \in V(e)$ with $|\incidences_G(v)| > s_2$ and every chip $Z$, we have $\incidences_G(v) \nsubseteq Z$, so $v \notin \inter_G(Z)$. Thus, we only need to consider vertices $v \in V(e)$ with $|\incidences_G(v)| \leq s_2$. For a fixed such vertex $v$, we can determine all chips $Z$ with $\incidences_G(v) \subseteq Z$ and thus with $v \in \inter_G(Z) \cap V(e)$, by calling the algorithm from \Cref{claim:static_local_search_chips} with $I = \incidences_G(v)$ and $X = \emptyset$ on the hypergraph $G$, which takes time $\OO_{r,s_2,k}(\log |E(G)|)$ and returns at most $\OO_{s_2,k}(1)$ chips. Then, for every such chip $Z$ we compute the boundary $\bd_G(Z)$ in time $\OO(r \cdot s_2)$, remove $Z$ from the corresponding balanced binary search tree $\bbst_{\bd_G(Z)}$, and update the rank $\vol(\bd_G(Z))$ in the max-heap in time $\OO_{r,s_2,k}(\log |E(G)|)$.

    Now, the balanced binary search trees contain only those chips of $G'$ that are also chips of $G$ with $\inter_G(Z) \cap V(e) = \emptyset$, and each one of them is in the correct balanced binary search tree with respect to their boundary in $G'$. Next, we reinsert the chips $Z$ of $G'$ with $\inter_G(Z) \cap V(e) \neq \emptyset$, which were deleted in the previous step. To do this, we again apply the algorithm from \Cref{claim:static_local_search_chips}, now on $G'$, with $I = \incidences_{G'}(v) \setminus \{e\}$ and $X = \{v\}$ for every vertex $v \in V(e)$ with $|\incidences(v) \setminus \{e\}| \leq s_2$. For each enumerated chip $Z$ of $G'$, we compute $\bd_{G'}(Z)$, insert $Z$ into $\bbst_{\bd_{G'}}(Z)$, and update the rank $\vol(\bd_{G'}(Z))$ in time $\OO_{r,s_2,k}(\log |E(G)|)$.

    Afterwards, every chip $Z$ of $G'$ with $Z \subseteq E(G)$ is in the correct balanced binary search tree $\bbst_{G'}(Z)$ since either $Z$ was an unaffected chip of $G$ (more precisely, $\inter_G(Z) \cap V(e) = \emptyset$) or we (re)inserted $Z$ in the previous step. Thus, the only chips $Z$ of $G'$ that are missing are the ones with $Z \subsetneq E(G)$, i.e., with $e \in Z$. We can find all these chips using the algorithm from \Cref{claim:static_local_search_chips} with $I = \{e\}$ and $X = \emptyset$ in time $\OO_{r,s_2,k}(\log |E(G)|)$ and there are at most $\OO_{s_2,k}(1)$ of them. For each one of them, we compute $\bd_{G'}(Z)$, insert $Z$ into $\bbst_{\bd_{G'}(Z)}$, and update the rank $\vol(\bd_{G'}(Z))$ in time $\OO_{r,s_2,k}(\log |E(G)|)$.

    \paragraph{$\deleteHyperedge(e)$:} Again, we start with updating the representation of $G$, which is a bipartite graph $H$ with parts $V(G)$ and $E(G)$. For the hyperedge $e \in E(G)$, we need to delete every edge that is incident to $e$ in $H$. This can be done in time $\OO(r)$ by iterating over the adjacencies of $e$ and using the pointers to find the second appearance of the respective edge. Lastly, deleting the vertex $e \in E(G)$ from $H$ can be done in $\OO(1)$ time.

    We continue with updating our data structures for the chips of $G$.
    When deleting $e$, every chip $Z$ with $e \in Z$ needs to be removed. By calling the algorithm from \Cref{claim:static_local_search_chips} with $I = \{e\}$ and $X = \emptyset$, we find all chips of $G$ that contain $e$ in time $\OO_{r,s_2,k}(\log |E(G)|)$, and there are at most $\OO_{s_2,k}(1)$ of them. We remove every such $Z$ from $\bbst_{\bd_G(Z)}$ and update $\vol(\bd_G(Z))$ in time $\OO_{r,s_2,k}(\log |E(G)|)$.

    For sets $Z \subseteq E(G)$ with $e \notin Z$, the boundary might decrease when deleting $e$. For one, this means that some of the existing chips need to be regrouped. Secondly, there might be chips $Z$ in $G'$ that are not chips of $G$ because $\bd_G(Z) > k$ or $Z$ is not internally connected in $G$. In the following, we show that this boundary decrease is the only thing that can happen to sets $Z$ with $e \notin Z$, before we take care of these sets.

    \begin{claim}\label{claim:delete_edge}
        Let $Z \subseteq E(G)$ be a set with $e \notin Z$. Then, $\lambda_G(Z) - r \leq  \lambda_{G'}(Z) \leq \lambda_G(Z)$ and $\inter_{G'}(Z) \supseteq \inter_G(Z)$. Furthermore, if there is no vertex $v \in \bd_G(Z) \cap V(e)$ with $\incidences_G(v) \setminus \{e\} \subseteq Z$, then $\bd_G(Z) = \bd_{G'}(Z)$ and $\inter_{G'}(Z) \supseteq \inter_G(Z)$.
    \end{claim}
    \begin{claimproof}
        First, note that deleting $e$ cannot increase the boundary of a set $Z$, i.e., $\bd_{G'}(Z) \subseteq \bd_G(Z)$, so $\lambda_{G'}(Z) \leq \lambda_{G}(Z)$ and $\inter_{G'}(Z) \supseteq \inter_G(Z)$. Moreover, we observe that for every vertex $v \in \bd_{G}(Z) \setminus \bd_{G'}(Z)$, we have $v \in V(e)$ and $\incidences_G(v) \setminus \{e\} \subseteq Z$. Since $|V(e)| \leq r$, it follows that $\lambda_{G'}(Z) \geq \lambda_G(Z) - r$. Moreover, if there is no such vertex $v \in \bd_G(Z) \cap V(e)$ with $\incidences_G(v) \setminus \{e\} \subseteq Z$, then $\bd_G(Z) = \bd_{G'}(Z)$.
    \end{claimproof}

    By \Cref{claim:delete_edge}, for a set $Z \subseteq E(G)$ with $e \notin Z$, the only thing that can happen when deleting $e$ is that the boundary might decrease and that $Z$ might become internally connected, and this happens only if there is a vertex $v \in V(e) \cap \bd_G(Z)$ with $\incidences_G(v) \setminus \{e\} \subseteq Z$. We find all chips $Z$ of $G'$ with a vertex $v \in V(e) \cap \bd_G(Z)$ and $\incidences_G(v) \setminus \{e\} \subseteq Z$ as follows. If  $| \incidences_G(v) \setminus \{e\} | > s_2$, there is no chip $Z$ with $\incidences_G(v) \setminus \{e\} \subseteq Z$.
    For all other vertices $v \in V(e)$, we call the algorithm from \Cref{claim:static_local_search_chips} on $G$ with $I = \incidences_G(v) \setminus \{e\}$, $X = \{v\}$, and a relaxed bound $\lambda(Z) \leq k+r$ on the boundary size, which takes time $\OO_{r,s_2,k}(\log |E(G)|)$ and returns at most $\OO_{s_2,k}(1)$ chips.
    We remark that the algorithm from \Cref{claim:static_local_search_chips} can clearly work with this relaxed bound on $\lambda$. We check whether $\bd_{G'}(Z) \leq k$ in time $\OO_{s_2,r,k}(1)$ and if so, add $Z$ to $\bbst_{\bd_{G'}(Z)}$, and update the rank $\vol(\bd_{G'}(Z))$ in the max-heap in time $\OO_{r,s_2,k}(\log |E(G)|)$. Before inserting $Z$ into $\bbst_{\bd_{G'}(Z)}$, we also check whether $Z$ is in $\bbst_{\bd_G(Z)}$, and if so, remove it from there and update the rank $\vol(\bd_G(Z))$ in the max-heap.

    \paragraph{$\query$:}     
    First, we find the set $B \subseteq V(G)$ of size $|B| \leq k$ with the maximum value of $\vol(B)$ in the max-heap in time $\OO(1)$, which also points to the corresponding balanced binary search tree $\bbst_B$. If $\vol(B) < s_1$, we can report that there is no suitable set $C$ of size $s_1 \leq |C| \leq s_2$. Otherwise, we iterate over (at most) the first $s_2$ sets in the balanced binary search tree $\bbst_B$. If we find a chip $Z \in \chips(B)$ with $|Z| \geq s_1 / 2$, we set $C \coloneq Z$ and return. Clearly, this set $C$ satisfies all the requirements as $Z$ is a chip, so $Z$ is internally connected, $s_1 / 2 \leq |Z| \leq s_2$, $\lambda(Z) \leq k$, and $\oracle(Z) = \true$.

    If we do not find such a large set $Z$, then the first $s_2$ sets (or less, if $\bbst_B$ contains less than $s_2$ sets) in $\bbst_B$ all have size smaller than $s_1 / 2$. In this case, we can find an arbitrary subset $\mc Z$ of these first (at most) $s_2$ chips such that $s_1 / 2 \leq \vol(\mc Z) < s_1 \leq s_2$ and we set $C = \bigcup_{Z \in \mc Z} Z$. 
    We now show that all chip pieces of $\chips(B)$ are disjoint.
    \begin{claim}\label{lem:internally_connected_disjoint}
        Let $G$ be a hypergraph and $A,B \subseteq E(G)$ be two internally connected sets with $\bd(A) = \bd(B)$. If $A \cap B \neq \emptyset$, then $A = B$.
    \end{claim}
    \begin{claimproof}
        Let $e \in A \cap B$ and suppose $A \neq B$.
        First, we argue that there is a vertex $v \in V(e) \cap \inter(A) \cap \inter(B)$. For this, observe that $V(e) \cap \inter(A) = V(e) \setminus \bd(A) = V(e) \setminus \bd(B) = V(e) \cap \inter(B)$, so every vertex in $V(e) \cap \inter(A)$ is also in $V(e) \cap \inter(B)$. Moreover, there is a vertex $v \in V(e) \cap \inter(A)$, since otherwise $V(e) \subseteq \bd(A) = \bd(B)$, which implies by the definition of internal connectivity $|A| = |B| = 1$, so $A = B = \{e\}$ contradicting $A \neq B$. So, let $v \in V(e) \cap \inter(A) \cap \inter(B)$.
    
        Without loss of generality, we assume that $A \setminus B \neq \emptyset$ and consider a hyperedge $e' \in A \setminus B$. Note that $V(e') \subseteq V(B)$ is clearly not possible since $V(e') \cap \inter(B) = \emptyset$ and $V(e') \nsubseteq \bd(B)$ by the definition of internal connectivity. So, let $v' \in V(e') \setminus V(B)$ and note that then $v' \notin \bd(B) = \bd(A)$, so $v' \in \inter(A)$. Since $v,v' \in \inter(A)$ and $A$ is internally connected, there is a path between $v$ and $v'$ in $\Pc(G)[\inter(A)]$. This is a contradiction since $\inter(A) \cap \bd(B) = \emptyset$, so this is a path from a vertex $v' \notin V(B)$ to a vertex $v \in \inter(B)$ that does not use vertices from $\bd(B)$.
    \end{claimproof}
    Note that since $\vol(B) \geq s_1$ there is a chip $Z \in \chips(B) \setminus \mc Z$, and since by \Cref{lem:internally_connected_disjoint} the chips in $\chips(B)$ are disjoint, we have $C \cap Z = \emptyset$. Thus, $\bd(C) = \bd(Z) = B$, so $\lambda(C) = |B| \leq k$, and every internal component $C'$ of $C$ is a chip $C' \in \chips(B)$, so $\bd(C') = B = \bd(C)$ and $\oracle(C') = \true$.

    Since we only need to look at (at most) the first $s_2$ sets in $\bbst_B$, determine their size, and possibly take the union of these sets, the $\query$ operation can be done in time $\OO_{s_2,r,k}(1)$.
\end{proof}

%% file: balance-protrusions.tex
\section{Balancing protrusions}\label{sec:splay}
Recall that our goal is to build a dynamic protrusion decomposition.
This can be seen as having two main challenges, (1) maintaining a root with small bag and low-degree, and (2) maintaining protrusions with small treewidth.
In the previous sections, we built a data structure, that will enable us to handle the first challenge, and in this section we will build a data structure that will enable us to handle the second challenge.
In~\Cref{sec:main_data_structure}, we will combine the two parts to maintain the full protrusion decomposition.

The techniques of this section are based on the dynamic treewidth data structure of~\cite{Korhonen_2025}.
We essentially want to use a copy of this data structure for each of our protrusions.
Unfortunately, we cannot do so in a completely black-box manner, because we will sometimes need to split and merge protrusions.
Nevertheless, our data structure will largely reuse subroutines from~\cite{Korhonen_2025}.
What we obtain is captured in~\Cref{lem:dynamic_prot}, which provides a dynamic protrusion decomposition data structure, but leaves the maintenance of the degree of the root to the user of the data structure.

\begin{lemma}\label{lem:dynamic_prot}
    There is a data structure that, initialized with an integer $c\geq 3$, maintains a graph $G$, the support hypergraph $\Hc(G)$ of $G$, a downwards well-linked superbranch decomposition $\Tc$ of $\Hc(G)$ with adhesion size at most $c$, as well as an annotated normal $(\Delta(r) \cdot c, \OO_c(1))$-protrusion decomposition $\tilde{\Tc}$ corresponding to $\Tc$ under the following operations:
    \begin{itemize}
        \item $\init(c)$: Given an integer $c \geq 3$, initialize the data structure with $c$ and an empty graph $G$.
        \item $\addVertex(v)$: Given a new vertex $v \notin V(G)$, add $v$ to $G$. 
        \item $\deleteVertex(v)$: Given an isolated vertex $v \in V(G)$, remove $v$ from $G$.
        \item $\addEdge(e)$: Given a new edge $e \in \binom{V(G)}{2} \setminus E(G)$, add $e$ to $G$.
        \item $\deleteEdge(e)$: Given an edge $e \in E(G)$, remove $e$ from $G$.
        \item $\merge(A)$: Given a proper non-empty subset $A \subsetneq \chd(r)$ of root-children with $2\leq|A|\leq 2^{2c}+1$ and $|\overline{A}| \ge 2$, corresponding to a set $B_A \subseteq E(\torso(r))$ such that $B_A \expand \Tc$ is well-linked in $\Hc(G)$ and $\wl(B_A \expand \Tc) \leq c$, add a new root-child node $r'$ and change the parent of each node in $A$ to be $r'$ instead of $r$, leaving $T$ otherwise unchanged.
    \end{itemize}
    The $\init$ operation runs in $\OO(1)$ time, the $\merge$ operation runs in $\OO_c(1)$ amortized time, and the remaining operations each run in $\OO_c( \log |G| )$ amortized time.
    For each operation, the changes to $\torso(r)$ can be described as a sequence $\mc C$ of hypergraph operations of size $\OO_c(1)$, which is returned.
    Changes to $\edges(r)$ can be described as a list of insertions and deletions of size at most $\OO_c(1)$, which is returned.
    
    Moreover, if upon initialization the data structure is provided a tree decomposition automaton $\autom$ with evaluation time $\tau$, then a run of $\autom$ on each protrusion is maintained, incurring an additional $\tau$ factor on the running times.
    If the automaton state at the root-child node corresponding to an hyperedge $e \in E(\torso(r))$ is modified by an operation, the corresponding sequence $\mc C$ will contain the operation $\deleteHyperedge(e)$.
\end{lemma}

The remainder of this section consists of proving this lemma.
In \Cref{sec:basic-rotations}, we introduce a standardized set of subroutines through which we will modify our superbranch decomposition, based on~\cite{Korhonen_2025}.
In \Cref{sec:balancing-sbd}, we generalize tree-balancing subroutines from~\cite{Korhonen_2025} to be suitable in our setting, and then in \Cref{sec:blancing-proof}, we provide the full data structure of~\Cref{lem:dynamic_prot}.

\subsection{Basic rotations}\label{sec:basic-rotations}

In this subsection, we describe four subroutines for local transformation of superbranch decompositions, which we call \emph{basic rotations}.
Each basic rotation modifies only a local part of the tree and updates the whole representation, including torsos accordingly. 
Our basic rotations correspond to those of~\cite{Korhonen_2025}, although we must be a bit more careful here, to avoid spending time linear in the root size.
In the following, let $G$ be a hypergraph and $\Tc = (T,\Lc)$ a superbranch decomposition of $G$.

\paragraph{Contracting.}

Let $pt \in E(T)$ be an edge of $T$ so that $p$ is the parent of $t$ and $t\in \Vint(T)$. \emph{Contracting} the edge $pt$ in $\Tc$, simply means contracting $pt$ in $T$ while preserving the leaf mapping $\Lc$. 
In the following lemma, we give an efficient algorithm for contracting an edge of a superbranch decomposition.

\begin{lemma}\label{lem:contract}
    A representation of $\Tc$ can be turned into a representation of $\Tc$ with $pt$ contracted in time $\OO(\adhsize(\Tc) \cdot |\torso(t)|)$. 
    Torso changes can be described as a sequence of hypergraph operations of size at most the running time, and these sequences are returned.
\end{lemma}
\begin{proof}
    Let $\Tc'=(T', \Lc')$ denote the new superbranch decomposition.
    We start by transforming $T$ into $T'$, so that $p$ corresponds to the new contracted node. 
    This is done by deleting $t$ and its incident edges from $T$, and then inserting an edge $ps$ for each $s\in \chd(t)$. This takes $\OO(|\torso(t)|)$ time. 
    To transform $\torso_\Tc(p)$ to $\torso_{\Tc'}(p)$, we start by deleting $e_t$ corresponding to $\adh_\Tc(pt)$  in time $\OO(|V(e_t)|)=O(|\torso(t)|)$.
    Then, for each hyperedge  $e\in E(\torso(t))\setminus\{e_t\}$, we insert $e$ into the bipartite representation of $\torso(p)$:
    \begin{itemize}
        \item add a vertex corresponding to $e$ in $\OO(1)$ time, and
        \item for each $v\in V(e)$ check whether $v\in V(\torso_\Tc(p))$. If so, add an edge between the existing vertex for $v$ and the new vertex for $e$. Otherwise, create a new vertex for $v$ and add the edge.
    \end{itemize}    
    Note that $v\in V(\torso_\Tc(p))$ if and only if $v\in \adh_\Tc(pt)$, so we can check this and retrieve a pointer to the vertex in $\OO(\adhsize(\Tc))$ time.
    As there are $\OO(|\torso(t)|)$ edges to insert, the total running time is $\OO(\adhsize(\Tc)  \cdot |\torso(t)|)$. 

    During this transformation, we can also construct a sequence $\mc C$ of hypergraph operations that describes the change of $\torso(p)$. The sequence consists of
    \begin{itemize}
        \item $\deleteEdge(e_t)$,
        \item for each $e\in E(\torso(t)\setminus\{e_t\})$: $\addVertex(v)$ for each new $v\in V(e)$ and $\addEdge(e)$
    \end{itemize}
    The size of $\mc C$ is upper bounded by $\adhsize(\Tc) + |\torso(t)| = \OO(\adhsize(\Tc) \cdot |\torso(t)|)$
\end{proof}

\paragraph{Splitting.}  Let $t\in \Vint(T)$ be an internal node and $C\subseteq E(\torso(t))$ a set of hyperedges with $|C|\geq2$ and $|\overline{C}|\geq 2$.
We assume that $C$ does not contain a hyperedge corresponding to the parent of $t$ (if it exists).
\textit{Splitting} the subset $C$ from $t$ in $\Tc$, means replacing $t$ with two new adjacent nodes, $t'$ and $t_{C}$, so that $t'$ is adjacent to every node $s$ with $e_s\in \overline{C}$ and $t_C$ is adjacent to every node $s$ with $e_s\in C$.

\begin{lemma}\label{lem:split}
    A representation of $\Tc$ can be turned into a representation with $C$ splitting from $t$, in time $\OO(\adhsize(\Tc) \cdot |C|)$.
    Torso changes can be described as a sequence of hypergraph operations of size at most the running time, and these sequences are returned.
\end{lemma}
\begin{proof}
    Let $\Tc'=(T', \Lc')$ denote the new superbranch decomposition.
    We first transform $T$ into $T'$. 
    Rename $t$ into $t'$, add a new node $t_C$, and insert the edge $t' t_C$. 
    For every child $s$ of $t$ with $e_s \in C$, remove the edge $t s$ and instead connect $t_C$ to $s$. 
    This step takes $\OO(|C|)$ time.

    Next, we update the adhesions. For each $s$ with $e_s\in C$, we set $\adh_{\Tc'}(t_Cc)=\adh_{\Tc}(tc)$. 
    For the edge $t' t_C$ we set $\adh_{\Tc'}(t' t_c) = \bd(C)$, where $\bd(C)$ can be computed in time $O(|C|\cdot \adhsize(\Tc))$ by \Cref{lem:compute-boundary}.

    To obtain $\torso(t')$, remove all hyperedges of $C$ from $\torso(t)$ and remove every vertex $v$ that becomes isolated after removing $C$,
    and finally add the new hyperedge $e_{t_C}$ with $V(e_{t_C})=\bd(C)$.
    To obtain $\torso(t_C)$, copy the part of $\torso(t)$ corresponding to edges in $C$ and add $e_{t_C}$. 
    Both steps take time $\OO(|C| \cdot \adhsize(\Tc))$.
    
    Performing these transformations, we can simultaneously construct a sequence $\mc C$ describing the change of $\torso(t)$ to $\torso(t')$. The sequence consists of
    \begin{itemize}
        \item $\deleteEdge(e)$ for each $e\in C$, 
        \item $\addEdge(e_{t_C})$, and
        \item $\deleteVertex(v)$ for each vertex $v$ that becomes isolated after deleting $C$. 
    \end{itemize}    
    The size of $\mc C$ is upper bounded by $\OO(\adhsize(\Tc) \cdot |C|)$. Similarly, the changes of $\torso(t_C)$ can be described by a sequence $\mc C'$ consisting of
    \begin{itemize}
        \item $\addVertex(v)$ for each $v\in V(C)$,
        \item $\addEdge(e)$ for each $e\in C$, 
        \item $\addEdge(e_{t_C})$.
    \end{itemize}    
    The size of $\mc C'$ is upper bounded by $\OO(\adhsize(\Tc) \cdot |C|)$.
\end{proof}

\paragraph{Inserting a leaf.}
We use this leaf insert operation to insert a new hyperedge $e$ to $G$.
This operation may also introduce new vertices, where all vertices in $V(e)\setminus V(G)$ are added to $G$ together with $e$.
To insert the hyperedge into the superbranch decomposition, we proceed as follows.

Let $t\in \Vint(T)$ be an internal node of $T$, and let $\cl(t) = \chd(t) \cap L(T)$ denote the set of leaf-children of $t$. Now, for a hyperedge $e\notin E(G)$ with $V(e) \cap V(G) \subseteq V(\Lc(\cl(t)))$, \textit{inserting} $e$ as a child of $t$ in $\Tc$ means adding a leaf-node $\ell_e$ as a child of $t$ in $T$, and setting $\Lc(\ell_e)=e$. 

To obtain an algorithm independent of $|\torso(t)|$, we take as an additional input a set $X \subseteq \cl(t)$ with $V(e) \cap V(G) \subseteq V(\Lc(X))$. Since only edges incident to $X$ are affected by the insertion, $\torso(t)$ can be updated by considering $X$ rather than the entire torso.

\begin{lemma}\label{lem:insert-leaf}
    Given a new hyperedge $e$ and $X \subseteq \cl(t)$ with $V(e)\cap V(G)\subseteq V(\Lc(X))$, a representation of $\Tc$ can be turned into a representation of $\Tc$ with $e$ inserted as a child of $t$ in time $\OO((|X| \cdot \rank(G)+1) \cdot |V(e)| + |\anc(t)|)$.
    Torso changes can be described as a sequence of hypergraph operations of size at most the running time, and these sequences are returned.
\end{lemma}
\begin{proof}
    Let $\Tc'$ denote the new tree decomposition and let $e'$ denote the hyperedge with $V(e')=V(e) \cap V(G)$. 
    Since $V(e') \subseteq V(\Lc(X))$, any edge $uv$ of $T'$ not incident to a child in $X$ satisfies $\adh_{\Tc'}(vt)=\adh_{\Tc}(vt)$.
    Hence, only adhesions between $t$ and the children in $X$ should be updated. 
    For $x \in X$ and $v \in V(\Lc(x))$, note that $v \notin \adh_{\Tc'}(xt)$ implies $v \notin V(\torso(t))$.
    So, to update $\adh_{\Tc'}(xt)$ and the corresponding hyperedge $e_x \in E(\torso(t))$, whenever $v \in V(e')$ and $v \notin \adh_{\Tc}(xt)$ we insert $v$ into $\adh_{\Tc'}(xt)$, add $v$ to $V(\torso(t))$, and add $v$ to $e_x$ in the bipartite representation of $\torso(t)$. This takes $\OO(|X| \cdot \rank(G) \cdot |V(e)|)$ time.

    Next, we add a new leaf $w$ with edge $tw$, set $\adh_{\Tc'}(tw) = V(e')$, and insert the hyperedge $e'$ into $\torso(t)$. 
    Since $V(e') \subseteq V(\Lc(X))$, all vertices of $V(e')$ can be identified in the bipartite representation of $\torso(t)$ by scanning the hyperedges corresponding to $X$, so the insertion runs in time $\OO(|X| \cdot \rank(G) \cdot |V(e)|)$. Finally, we need to increase the stored number of descendant leaves for the ancestors of $t$ in time $\OO(|\anc(t)|)$.

    During this transformation, we can also construct a sequence $\mc C$ of hypergraph operations that describes the change of $\torso(t)$. The sequence consists of
    \begin{itemize}
        \item $\addVertex(v)$ for each new $v\in V(e')\setminus V(\torso_\Tc(t))$,
        \item $\deleteEdge(e_x)$ for each $e_x\in E(\torso_{\Tc}(t))$ with $x\in X$,
        \item $\addEdge(e_x)$ for each $e_x\in E(\torso_{\Tc'}(t))$ with $x\in X$, and
        \item $\addEdge(e')$.
    \end{itemize}
    The size of $\mc C$ is upper bounded by $|V(e)|+ 2 \cdot |X|\cdot (\rank(G)+1) + (|V(e)|+1)$.
\end{proof}

\paragraph{Deleting a leaf.} 

Let $t\in \Vint(T)$ be an internal node of $T$ with at least $3$ children.
Now, for a leaf $\ell\in \cl(t)$ with $V(\Lc(\ell)) \cap V(G\setminus\Lc(\ell))\subseteq V(\Lc(\cl(t)\setminus\{\ell\}))$, \textit{deleting} $\ell$ from $\Tc$ means deleting $\ell$ and the edge $\ell t$ from $T$.   

\begin{lemma}\label{lem:delete-leaf}
    A representation of $\Tc$ can be turned into a representation of $\Tc$ with $\ell$ deleted in time $\OO(\adhsize(\Tc)^2 + |\anc(t)|)$.
    Torso changes can be described as a sequence of hypergraph operations of size at most the running time, and these sequences are returned.
\end{lemma}
\begin{proof}
    Let $\Tc'$ denote the new tree decomposition. 
    Since $V(\Lc(\ell)) \cap V(G\setminus\Lc(\ell))\subseteq V(\Lc(\cl(t)\setminus\{\ell\}))$, any edge $uv$ of $T'$ not incident to a child of $t$ satisfies $\adh_{\Tc'}(vt)=\adh_{\Tc}(vt)$.
    Hence, only adhesions between $t$ and the children should be updated.

    For a child $s$, a vertex $v\in \adh_{\Tc}(ts)$ may need to be removed only if $v\in V(e_\ell)$.
    This is the case if $v$ is incident only to $e_s$ and $e_\ell$ in $\torso_{\Tc}(t)$.
    This also means that each $v\in V(e_\ell)$ can affect at most one adhesion, which can be updated in $\OO(\adhsize(\Tc))$ time.
    As $|V(e_\ell)|\leq \adhsize(\Tc)$ this takes $\OO(\adhsize(\Tc)^2)$ time.
    Finally, we remove $\ell$ from $T$, and then we need to decrease the stored number of descendant leaves for the ancestors of $t$ in time $\OO(|\anc(t)|)$.

    During this transformation, we can also construct a sequence $\mc C$ of hypergraph operations that describes the change of $\torso(t)$.  The sequence consists of
    \begin{itemize}
        \item $\deleteEdge(e_\ell)$,
        \item for each hyperedge that loses a vertex of $V(e_\ell)$, update it by performing $\deleteEdge$ followed by $\addEdge$,
        \item $\deleteVertex(v)$ for each $x\in V(e_\ell)$ that becomes isolated after deleting $e_\ell$.
    \end{itemize}
    The size of $\mc C$ is upper bounded by $\OO(\adhsize(\Tc)^2)$.
\end{proof}

\paragraph{Representing basic rotations.}
We will denote a basic rotation by $s$, where $s$ encodes the type of rotation (contraction, split, leaf insertion, or leaf deletion) together with the information required to perform it: for a contraction of an edge $st$, the pair $(s,t)$ is stored; for a split of a node $t$ with edge set $C$, the pair $(t,C)$ is stored; for inserting a leaf at $t$ with respect to a set $X$, the pair $(t,X)$ is stored; and for deleting a leaf $\ell$, the node $\ell$ is stored.

If $\Tc$ and $\Tc'$ are the superbranch decompositions before and after applying a basic rotation $s$, we say that $s$ \emph{transforms} $\Tc$ into $\Tc'$.
The size of a rotation, denoted $\|s\|$, is defined so that applying $s$ takes time $\OO(\|s\|)$. Specifically, a contraction has size $\adhsize(\Tc) \cdot |\torso(t)|$, a split has size $\adhsize(\Tc) \cdot |C|$, a leaf insertion has size $(|X| \cdot \adhsize(\Tc)+1) \cdot |V(e)|+|\anc(t)|$, and a leaf deletion has size $\adhsize(\Tc)^2+|\anc(t)|$.

A \emph{sequence of basic rotations} is a finite sequence $\Sc=s_1s_2\cdots s_k$, where each $s_i$ is a basic rotation as described above. Applying $\Sc$ to a superbranch decomposition $\Tc$ \emph{transforms} $\Tc$ into a new decomposition $\Tc'$ by applying $s_1$ to $s_k$ one sequentially. The \emph{size} of the sequence is $\|\Sc\|=\sum_i(\|s_i\|+1)$, so that applying $\Sc$ takes time $\OO(\|\Sc\|)$.
We define $V_\Tc(\Sc)$ as the set of nodes of $T$ involved in the rotations: all internal nodes involved in a contraction or a split, and all nodes deleted or inserted together with their parents. Analogously, $V_{\Tc'}(\Sc)$ denotes the corresponding set in the transformed decomposition.
The \emph{trace} of $\Sc$ with respect to $\Tc$ is $\trace_\Tc(\Sc)=\anc(V_\Tc(\Sc))$, the set of ancestors of nodes involved in the sequence. We define $\trace_{\Tc'}(\Sc)$ analogously. Finally, we let $\|\Sc\|_\Tc=\|\Sc\|+|\trace_\Tc(\Sc)|$. Observe that $|\trace_{\Tc'}(\Sc)|\le \|\Sc\|_\Tc$.

\subsection{Balancing subroutines}\label{sec:balancing-sbd}

In this section, we introduce balancing subroutines that allow us to maintain the invariants of the superbranch decomposition under dynamic updates.
To this end, we define the notions of $c$-semigood and $c$-good.
The property of being $c$-good captures exactly what we wish to maintain: bounded degree and logarithmic depth.
However, we sometimes revert to maintaining the weaker condition of $c$-semigood and afterwards restoring the $c$-good property by applying a dedicated subroutine.
The definitions and subroutines of this subsection are taken from~\cite{Korhonen_2025} with slight adaptations to our setting.

\paragraph{$c$-good and $c$-semigood.}
Let $\Tc=(T,\Lc)$ be a rooted superbranch decomposition of a hypergraph $G$. For a node $x\in V(T)$, we say that $\Tc$ is \emph{$c$-semigood at $x$}, if
\begin{itemize}
    \item for every $v\in V(T_x)$, the set $\Lc[v]$ is well-linked in $G$ (i.e., the subtree is downwards well-linked),
    \item $\Delta(T_x)\leq 2^{2c}+1$,
    \item $\wl(\Lc[x])\leq c$.
\end{itemize}
For a node $x\in V(T)$, we say that $x$ is \emph{$d$-unbalanced} for an integer $d\geq 1$ if there exists a $s\in \desc(t)$, so that $\depth(s) \geq \depth(t)+d$ and $|\Lc[s]|\geq (2/3) \cdot |\Lc[x]|$.
A node $x\in V(T)$ is \emph{$d$-balanced} if it is not $d$-unbalanced, i.e., no such descendant exists.
For a node $x\in V(T)$, we say that $\Tc$ is \emph{$c$-good at $x$}, if 
\begin{itemize}
    \item  every $v\in V(T_x)$ is $2^{2c+1}$-balanced,
    \item $\Tc$ is $c$-semigood at $x$.
\end{itemize}
We say that $\Tc$ is \emph{$c$-good} (respectively, \emph{$c$-semigood}) if, for every root child $t$, $\Tc$ is $c$-good (respectively, $c$-semigood) at $t$. 

Being $c$-good at a node implies that the corresponding part of the decomposition has logarithmic height, by the following lemma.

\begin{lemma}[{\cite[Lemma 6.2]{Korhonen_2025}}]\label{lem:c-good-balanced}
    Let $G$ be a hypergraph and $\Tc = (T, \Lc)$ a rooted superbranch decomposition of $G$. For a node $x\in V(T)$, if $\Tc$ is $c$-good at $x$ then $\depth(T_x)\leq 2^{\OO(c)}\log |T_x|$ 
\end{lemma}

\paragraph{Potential functions.}
The analysis of the balancing subroutines is based on potential functions.
As in \cite{Korhonen_2025}, we define the potential of a single internal node $t \in V_{\inter}(T)$ as \[\Phi_\Tc(t) = (\Delta(t)-1)\log_2|\Lc[t]|.\]
Note that all internal nodes then have potential at least $1$. We then define the potential of $\Tc$ as
\[\Phi(\Tc) = \sum_{t \in V_\inter(T)}\Phi_\Tc(t).\]

\paragraph{Balancing subroutines.}
We now give the balancing subroutines, described in the two lemmas below.
The first shows how to restore the $c$-good property at a node.

\begin{lemma}[Generalization of {\cite[Lemma 7.1]{Korhonen_2025}}]\label{lem:c-semigood-to-good}
    Let $G$ be a hypergraph, $\Tc=(T, \Lc)$ a superbranch decomposition of $G$ rooted at $r$, and $c\geq 1$ an integer.
    Suppose also that $t$ is a root child for which $\Tc$ is $c$-semigood at $t$.
    There is an algorithm that, given $c$, $t$, and a prefix $R\subseteq V(T_t)$ of $T_t$ containing all $2^{2c+1}$-unbalanced nodes of $T_t$, transforms $\Tc$ into a superbranch decomposition $\Tc' = (T', \Lc')$ of $G$ rooted at $r'$ via a sequence $\Sc$ of basic rotations, which is returned, such that,
    \begin{itemize}
        \item $\Phi(\Tc')\leq \Phi(\Tc)$,
        \item $V_\Tc(\Sc)\subseteq \desc(t)$,
        \item $\Tc'$ is $c$-good at $t'$,
    \end{itemize}
    where $t'$ is the child of $r'$ with $V_{\Tc'}(\Sc)\subseteq \desc(t')$ if $|\Sc| > 0$, and $t' = t$ otherwise.
    The running time of the algorithm is $2^{\OO(c)} (|R| + \Phi(\Tc ) - \Phi(\Tc'))$, which is also an upper bound for $\|S\|_{\Tc}$.
\end{lemma}
\begin{proof}
    The lemma from~\cite{Korhonen_2025} only works on the root of tree decompositions, not on subtrees of tree decompositions.
    In~\cite{Korhonen_2025} tree decompositions have a degree-one ``dummy root'', making the unique node at depth 1 (what we think of as the root) an internal node.
    This means that their situation actually is very similar to ours, where our root child $t$ acts as the root a tree decomposition of the protrusion $\Lc[t]$.
    The main difference is that our ``root'' $t$ can have a boundary/adhesion to nodes outside $\Lc[t]$.
    The proof carries through, essentially without modification, so we will not repeat it here.
    Instead, we mention the following changes to be made to the proof of~\cite{Korhonen_2025} (which will not make sense without reading~\cite{Korhonen_2025}).
    Instead of assuming that $\Tc$ is $c$-good and that $\wl(G) \leq c$, we should assume this for the subtree we wish to balance, and $R$ should be a prefix of this subtree, containing all the unbalanced nodes of the subtree.
    (This is what we do in our statement.)
    Their proof uses their Lemma 5.6 and 5.4, where $\wl(G)$ should be replaced by $\wl(\Lc[t])$.
    Finally, note that our definition of semigood includes bounded well-linked number.
\end{proof}

The next balancing subroutine allows us to move a hyperedge close to the root.
This will be useful when inserting and deleting hyperedges in the data structure.

\begin{lemma}[Generalization of {\cite[Lemma 8.1]{Korhonen_2025}}]\label{lem:rotate-to-root}
    Let $G$ be a hypergraph, and $\Tc=(T, \Lc)$ a superbranch decomposition of $G$ rooted at $r$ and $c\geq 3$ an integer.
    Suppose also that $t$ is a root child such that $\Tc$ is $c$-good at $t$.
    There is an algorithm that, given $c$, $t$, and a hyperedge $e \in \Lc[t]$ with $|V(e)|\leq 2$, transforms $\Tc$ into a superbranch decomposition $\Tc' = (T', \Lc')$ of $G$ rooted at $r'$ via a sequence $\Sc$ of basic rotations, which is returned, such that,
    \begin{itemize}
        \item $\Phi(\Tc') \leq \Phi(\Tc) + 2^{\OO(c)}\log |G|$,
        \item $V_{\Tc}(\mc S)\subseteq \desc(t)$,
        \item $\Lc^{\prime-1}(e) \in \chd(t')$,
        \item $\Tc'$ is $c$-semigood at $t'$, 
    \end{itemize}
    where $t'$ is the child of $r'$ with $\Lc^{\prime-1}(e) \in \desc(t')$. 
    The running time of the algorithm is $2^{\OO(c)}\log |G|$, which is also an upper bound for $\| S\|_{\Tc}$.
\end{lemma}
\begin{proof}
    This lemma follows from~\cite{Korhonen_2025} in a simliarly straightforward manner as~\Cref{lem:c-semigood-to-good}.
    We again mention the trivial changes to be made to their proof.
    Instead of assuming that $\Tc$ is $c$-good and that $\wl(G) \leq c$, we should assume this for the subtree rooted at the node $t$ we wish to rotate to, and depth should be measured against this node instead of the unique node at depth 1.
    (This is what we do in our statement.)
    Their Lemma 8.1 uses Lemma 8.2, 5.6, and 5.4, where $\wl(G)$ should be replaced by $\wl(\Lc[t])$.
\end{proof}

\subsection{Proofs for the protrusion-balancing data structure}\label{sec:blancing-proof}
In this subsection we prove the correctness and time bounds of the balancing data structure of~\Cref{lem:dynamic_prot}.
First, we give a subroutine for rotating leaves upwards until they become root children.

\begin{lemma}\label{claim:isolate}
    Let $\Tc=(T, \Lc)$ be $c$-good superbranch decomposition rooted at $r$ of $\Hc(G)$, for an integer $c \geq 3$.
    There is an algorithm that, given $c$ and a set of hyperedges $X \subseteq E(\Hc(G))$ with $|X|\leq 3$, transforms $\Tc$ into a $c$-good superbranch decomposition $\Tc' = (T', \Lc')$ rooted at $r'$ via a sequence $\Sc$ of basic rotations, so that 
    \begin{itemize}
        \item $\Lc^{\prime-1}(X)\subseteq\chd(r')$ and
        \item $\Phi(\Tc')\leq \Phi(\Tc)+ \OO_c(\log |G|)$.
    \end{itemize}    
    The running time is $\OO_c(\log | G| + \max\{\Phi(\Tc)-\Phi(\Tc_3),0\})$, which is also an upper bound for $\|\Sc\|_\Tc$.
    Furthermore, $\torso_{\Tc'}(r')$ can be obtained from $\torso_{\Tc}(r)$ via a sequence $\mc C$ of basic hypergraph operations of size $\OO_c(1)$. 
    If the subtree of $\Tc$ rooted at the node corresponding to an edge $e \in E(\torso_\Tc(r))$ is modified by the algorithm, then $\mc C$ is guaranteed to contain the operation $\deleteHyperedge(e)$.
    The algorithm returns $\Sc$ and $\mc C$.
\end{lemma}
\begin{proof}
    We first describe an algorithm that satisfies all the claimed properties but takes a single hyperedge $e$ as input instead of a set $X$.
    The full algorithm for a set $X$ is then obtained by applying it once for each $e\in X$.
    To ensure that all hyperedges in $X$ end up as children of the root, the algorithm will also guarantee that for every $x \in \chd(r) \setminus \{t\}$, where $t$ is the root child with $e \in \Lc[t]$, the node $x$ remains a child of the new root $r'$ and the subtree rooted at $x$ is left unchanged.
    This way, once a hyperedge is moved to the root, it is not moved away again in later calls.
    
    The algorithm is divided into three steps, each producing a new superbranch decomposition $\Tc_i$, with $\Tc_3=(T_3,\Lc_3)$ being the final superbranch decomposition that is $c$-good and satisfies $\Lc^{-1}_3(e)\in \chd(r_3)$ where $r_3$ is the root.

    \paragraph{Step 1.}  
    In this step, we will rotate $e$ close to the root. 
    First identify the root child $t\in \chd(r)$ with $e\in \Lc[t]$ in time $\OO(\depth(\Tc))=\OO_c(\log |G|)$.
    By applying \Cref{lem:rotate-to-root} to $\Tc$ with $c$, $t$ and $e$, we obtain in time $\OO_c(\log |G|)$ a superbranch decomposition $\Tc_1=(T_1,\Lc_1)$ rooted at $r_1$ via a sequence $\Sc_1$ of basic rotations, so that, 
    \begin{itemize}
        \item $\Phi(\Tc_1) \leq \Phi(\Tc) + \OO_c(\log |G|)$, 
        \item $V_{\Tc}(\mc S_1)\subseteq \desc(t)$,
        \item $\|\Sc_1\|_\Tc \leq \OO_c(\log | G|)$,
        \item $\Tc_1$ is $c$-semigood at $t_1$, 
        \item $\Lc_1^{-1}(e) \in \chd(t_1)$,
    \end{itemize}
    where $t_1$ is the child of $r_1$ with $V_{\Tc'}(\Sc)\subseteq \desc(t_1)$. 
    Because $V_{\Tc}(\mc S_1)\subseteq \desc(t)$, the only difference between $\Tc$ and $\Tc_1$ is the subtree at $t_1$, in the sense that replacing the subtree $(T_1)_{t_1}$ with $T_{t}$ gives a superbranch decomposition identical to $\Tc$.
    This means that for each $s \in \chd(r_1) \setminus \{t_1\}$, the subtree rooted at $s$ is not touched by the algorithm, so $\Tc_1$ is still $c$-good at $s$. 
    This also means that $\torso_{\Tc}(r)=\torso_{\Tc_1}(r_1)$. 

    \paragraph{Step 2.}
    Since $\Lc_1^{-1}(e)\in \chd(t_1)$, we can make $\Lc_1^{-1}(e)$ a root child by contracting the edge $rt_1$ using the basic rotation of \Cref{lem:contract} in $\OO(\adhsize(\Tc)\cdot  |\torso(t_1)|) = \OO_c(1)$ time. 
    Let $\Tc_2=(T_2,\Lc_2)$ be the resulting superbranch decomposition rooted at $r_2$ (the contracted node). 
    By the lemma, we further have that $\torso_{\Tc_2}(r_2)$ can be obtained from $\torso_{\Tc_1}(r_1)$ by a sequence $\mc C$ of size $\OO_c(1)$.
    As $\Delta(r_2)$ increases by at most $\OO_c(1)$, the potential increases by at most $\OO_c(\log |G|)$, i.e. $\Phi(\Tc_2)\leq \Phi(\Tc_1) + \OO_c(\log |G|)$.
    Let $\Sc_{1,2}$ be the concatenation of $\Sc$ and this contraction basic rotation, satisfying $\|\Sc_{1,2}\|_{\Tc} \leq \|S_1\|_{\Tc} + \OO_c(1)\leq \OO_c(\log |G|)$.

    \paragraph{Step 3.}
    Finally we will transform $\Tc_2$ (which is $c$-semigood) into a $c$-good superbranch decomposition $\Tc_3$.
    Let $R\subseteq \chd(r_2) \setminus\{\Lc_2^{-1}(e)\}$ be the set of children of $r_2$ corresponding to the children of $t_1$ in $\Tc_1$, excluding the child associated with $e$. 
    Since $\Tc_1$ is $c$-semigood at $t_1$, it follows that for each $s\in R$, the superbranch decomposition $\Tc_2$ is $c$-semigood at $s$.   
    Note that for every other child of $r_2$, the subtree is still $c$-good at this node.
    Thus, to obtain a $c$-good superbranch decomposition, it suffices to consider the children in $R$. 
    For a node $s\in R$, define $P_s=\trace_{\Tc_2}(\Sc_{1,2})\cap V((T_2)_s)$.
    This is a prefix of $(T_2)_s$ that contains all of its $2^{2c+1}$-unbalanced nodes, and $|P_s| \leq \|\Sc_{1,2}\|_{\Tc} \leq \OO_c(\log | G|)$. 
    Applying \Cref{lem:c-semigood-to-good} with $s$ and $P_s$ for each $s\in R$, transforms $\Tc_2$ into a $c$-good superbranch decomposition $\Tc_3$ rooted at $r_3$ via a sequence $\Sc_3$ of basic rotations, so that,
    \begin{itemize}
        \item $\Phi(\Tc_3)\leq \Phi(\Tc_2)$ and
        \item $V_{\Tc_2}(\Sc_3) \subseteq \bigcup_{s\in R} \desc(s)$.
    \end{itemize}
    The latter implies that $\Lc^{-1}_3(e)$ is still a child of the root. 
    The running time is $\OO_c(\log |G| + \Phi(\Tc_2) - \Phi(\Tc_3))$, which also bounds $\|\Sc_3\|_{\Tc_2}$.     
    Since $r_2\notin V_{\Tc_2}(\Sc_2)$, we have $\torso_{\Tc_2}(r_2)=\torso_{\Tc_3}(r_3)$. 
    Therefore we can obtain $\torso_{\Tc_3}(r_3)$ from $\torso_{\Tc}(r)$ via the sequence $\mc C$ of size $\OO_c(1)$. 
    Combining inequalities from the previous steps, $\Phi(\Tc_3) \leq \Phi(\Tc) + \OO_c(\log |G|)$.
    The total running time is $\OO_c(\log |G|)+(\Phi(\Tc_2) - \Phi(\Tc_3))$, and since $\Phi(\Tc_2) \leq \Phi(\Tc)+\OO_c(\log |G|)$ we get the claimed running time.
    Let $\Sc$ be the entire sequence of basic rotations, i.e., the concatenation of $\Sc_{1,2}$ and $\Sc_3$.
    Then the size $\|\mc S\|_{\Tc}\leq \|\mc S_{1,2}\|_\Tc + \|\mc S_3\|_{\Tc_2}$ is upper bounded by the running time.
    Observe that $V_\Tc(\Sc)\subseteq \desc(t) \cup \{r\}$. 
    Since the only basic rotation involving $r$ is the contraction, every $x \in \chd(r) \setminus \{t\}$ remains a child of the new root $r'$ and the subtree rooted at $x$ is left unchanged.
    This also implies that for every hyperedge $e$ in $E(\torso_\Tc(r))$, if $\mc C$ does not delete $e$, then the subtree of $\Tc$ rooted at the node corresponding to $e$ does not change during the transformation.
    Finally, we return $\mc C$ and $\mc S$.
\end{proof}

We now prove the first part of the protrusion balancing data structure, maintaining a superbranch decomposition.

\begin{proof}[Proof of \Cref{lem:dynamic_prot}]
    We first show how to maintain the superbranch decomposition under all operations, before moving to the protrusion decomposition.

\paragraph{Superbranch decomposition.}
    We will maintain a $c$-good superbranch decomposition $\Tc$ of $\Hc(G)$ and analyze the amortized running time using the potential function $\Phi$ defined in~\Cref{sec:balancing-sbd}.
    Any change $\Phi(\Tc')-\Phi(\Tc)$ is charged as $\OO_c(\Phi(\Tc')-\Phi(\Tc))$ in the amortized time.
    Note that, since $\Tc$ is downwards well-linked and every node $t\in V(T)\setminus\{r\}$ satisfies $\wl(\Lc[t])\leq c$ in $G$, it follows that $\adhsize(\Tc)\leq c$.

    In all operations other than $\init$, we will update $\Tc$ by a sequence $\mc S$ of basic rotations, which we save for use later in the proof.
    The sizes $\|\mc S\|_\Tc$ of these sequences will have the same amortized upper bound as the running times.
    
    Furthermore, each operation will guarantee that if the subtree rooted at a node $t \in \chd(r)$ corresponding to a hyperedge $e \in E(\torso(r))$ is modified, then $\mc C$ will contain the operation $\deleteHyperedge(e)$.
    This will be clear from our operation implementations, when considering the guarantee from~\Cref{claim:isolate}.
    
    We note that the algorithm of \Cref{claim:isolate} runs in $\OO_c(\log |G|)$ amortized time.
    Specifically, it transform $\Tc$ into a new superbranch decomposition $\Tc'$ such that $\Phi(\Tc')\leq \Phi(\Tc)+ \OO_c(\log |G|)$ in time $\OO_c(\log | G| + \max\{\Phi(\Tc)-\Phi(\Tc'),0\})$. By accounting for the change in potential, we get this amortized running time.

\paragraph{$\addVertex(v)$.} 
    First, we transform $G$ into a graph $G'$ by adding $v$ to the vertex set.
    We transform $\Tc$ into a superbranch decomposition $\Tc'$ of $\Hc(G')$ rooted at $r'$ by inserting $e_v$ as a child of $r$ using the rotation of \Cref{lem:insert-leaf} with $X=\emptyset$ in $\OO(|V(e)|)=\OO(1)$ time. 
    This corresponds to a sequence $\Sc$ with $\|\Sc\|_\Tc\leq O(1)$. 
    The torso of $r'$ can be obtained from $\torso(r)$ by inserting a hyperedge $e$ with $V(e)=\emptyset$ corresponding to the empty adhesion of $e_v$.
    This can be done through a sequence $\mc C$ of basic hypergraph operations of size $\OO(1)$.
    Since $v$ is an isolated vertex, the resulting decomposition $\Tc'$ remains downwards well-linked, with unchanged well-linked number. 
    Then, since every node of $T'$ other than the root has the same number of children as in $T$, it follows that $\Tc'$ is $c$-good.

\paragraph{$\deleteVertex(v)$.}  
    First, we transform $G$ into a graph $G'$ by deleting $v$ from the vertex set.
    We apply~\Cref{claim:isolate} to $e_v$ which transforms $\Tc$ into a $c$-good superbranch decomposition $\Tc_1$ of $\Hc(G)$ rooted at $r_1$ with $\Lc^{-1}(e_v) \in \chd(r_1)$ in $\OO_c(\log|G|)$ amortized time, through a sequence $\Sc_1$ of basic rotations of size bounded by $\OO_c(\log|G|)$.
    The lemma also gives a sequence of hypergraph operations $\mc C_1$ of size $\OO_c(1)$ transforming $\torso(r)$ into $\torso(r_1)$.
    Next, we transform $\Tc_1$ into a superbranch decomposition $\Tc_2$ of $\Hc(G')$ rooted at $r_2$ by deleting $e_v$ as a child of $r_1$ using the rotation of \Cref{lem:delete-leaf} in $\OO(\adhsize(\Tc_1)^2) =\OO_c(1)$ time, through a (single-rotation) sequence $\Sc_2$ of size $\|\Sc_2\|\leq \OO_c(1)$.
    Concatenating the sequences $\Sc_1$ and $\Sc_2$, we get a final sequence $\Sc$ of size $\OO_c(\log|G|)$.
    The torso of $r_2$ can be obtained from $\torso(r_1)$ by removing the hyperedge corresponding to the empty adhesion of $e_v$, through a sequence $\mc C_2$ of basic hypergraph operations of size $\OO(1)$.
    Concatenating $\mc C_1$ and $\mc C_2$ we get a final sequence of size $\OO_c(1)$ transforming $\torso(r)$ into $\torso(r_2)$.
    Since $v$ is an isolated vertex, the resulting decomposition $\Tc_2$ remains downwards well-linked, with unchanged well-linked number. 
    Then, since every node of $T_2$ other than the root has the same number of children as in $T$, it follows that $\Tc'$ is $c$-good.

\paragraph{$\addEdge(uv)$.}
    First, we transform $G$ into a graph $G'$ by adding $uv$ to the edge set.
    We apply \Cref{claim:isolate} to both hyperedges $e_u$ and $e_v$, which transforms $\Tc$ into a $c$-good superbranch decomposition $\Tc_1$ of $\Hc(G)$ rooted at $r_1$ with $\Lc^{-1}\{e_u,e_v\}\subseteq \chd(r)$ in $\OO_c(\log |G|)$ amortized time.
    Let $\Sc_1$ and $\mc C_1$ be the returned sequences of basic rotations and hypergraph operations, respectively.
    The amortized size of $\|\Sc_1\|_\Tc$ is $\OO_c(\log |G|)$ and $\|\mc C_1\|\leq \OO_c(1)$.

    Then, we transform $\Tc_1$ into a superbranch decomposition $\Tc_2$ of $\Hc(G')$ rooted at $r_2$ by inserting $e_{uv}$ as a child of $r_1$ using the rotation of \Cref{lem:insert-leaf} with $X=\Lc_1^{-1}\{e_u,e_v\}$ in $\OO(\rank(\Hc(G)) \cdot |V(e)|)=\OO_c(1)$ time. 
    This corresponds to a sequence $\Sc_2$ with $\|\Sc_2\|_{\Tc_1}\leq \OO_c(1)$. 
    By the lemma, we also have that $\torso(r_2)$ can be obtained from $\torso(r_1)$ via a sequence $\mc C_2$ of basic hypergraph operations of size $\OO_c(1)$.
    We can thus let $\Sc$ be the concatenation of $\Sc_1$ and $\Sc_2$, and $\mc C$ be the concatenation of $\mc C_1$ and $\mc C_2$.

    Note that $\Tc_2$ is trivially $c$-good at $e_{uv}$.
    Furthermore, $\Tc_2$ is $c$-good at every other $t \in \chd(r_2)\setminus\{e_{uv}\}$, as the insertion basic rotation leaves the subtree rooted at $t$, as well as the boundary of every subset of $\Lc_2[t]$ unchanged, since $e_u$ and $e_v$ were already root children.
    
\paragraph{$\deleteEdge(uv)$.}
    First, we transform $G$ into a graph $G'$ by deleting $uv$ from the edge set.
    We apply \Cref{claim:isolate} to both hyperedges $e_u, e_v, e_{uv}$, which transforms $\Tc$ into a $c$-good superbranch decomposition $\Tc_1$ rooted at $r_1$ with $\Lc^{-1}\{e_u,e_v,e_{uv}\}\subseteq \chd(r)$ in $\OO_c(\log |G|)$ amortized time.
    Let $\Sc_1$ and $\mc C_1$ be the returned sequences of basic rotations and hypergraph operations, respectively.
    The amortized size of $\|\Sc_1\|_\Tc$ is $\OO_c(\log |G|)$ and $\|\mc C_1\|\leq \OO_c(1)$.

    Then, we transform $\Tc_1$ into a superbranch decomposition of $\Tc_2$ of $\Hc(G')$ rooted at $r_2$ by deleting $e_{uv}$ as a child of $r$ using the rotation of \Cref{lem:delete-leaf} in $\OO(\adhsize(\Tc_1)^2)=\OO_c(1)$ time. 
    This corresponds to a sequence $\Sc_2$ with $\|\Sc_2\|_{\Tc_1}\leq \OO_c(1)$. 
    By the lemma, we have that $\torso(r_2)$ can be obtained from $\torso(r_1)$ via a sequence $\mc C_2$ of basic hypergraph operations of size $\OO_c(1)$. 
    We can thus let $\Sc$ be the concatenation of $\Sc_1$ and $\Sc_2$, and $\mc C$ be the concatenation of $\mc C_1$ and $\mc C_2$.

    By the same argument as for $\addEdge(uv)$, we get that $\Tc_2$ is $c$-good.
    
\paragraph{$\merge(A)$.}
    By performing a split rotation using \Cref{lem:split} with $B_A$, we transform $\Tc$ into a superbranch decomposition $\Tc_1=(T_1,\Lc_1)$ rooted at $r_1$, where the children of $A$ become children of a new node $t$ that is a child of $r_1$.
    This corresponds to a sequence $\Sc_1$ with $\|\Sc_1\|_\Tc\leq \OO_c(1)$. By the lemma, we have that $\torso(r_1)$ can be obtained from $\torso(r)$ via a sequence $\mc C_1$ of basic hypergraph operations of size $\OO_c(1)$. Since $\Tc$ is downwards well-linked and $B_A\triangleright\Tc$ is well-linked in $G$, it follows that $\Tc_1$ is downwards well-linked. 
    Because $|A|\le 2^{2c}+1$, we have $\Delta((T_1)_t)\le 2^{2c}+1$. 
    Together with $\wl(B_A\triangleright\Tc)\leq c$, this shows that $\Tc_1$ is $c$-semigood.
    Note that a split rotation cannot increase the potential, i.e. $\Phi(\Tc_1)\leq \Phi(\Tc)$.

    Observe that the only node in $V(T_1)\setminus\{r_1\}$ that can be $2^{2c+1}$-unbalanced is $t$. By applying \Cref{lem:c-semigood-to-good} with $t$ and prefix $\{t\}$, we transform $\Tc_1$ into a $c$-good superbranch decomposition $\Tc_2$ via a sequence $\Sc_2$ such that $r_1\notin V_{\Tc_1}(\Sc_2)$ and $\Phi(\Tc_2)\leq \Phi(\Tc_1)$. 
    Thus, the root torso is unchanged. 
    Furthermore, the amortized running time of this algorithm is $\OO_c(1)$, which is also an upper bound for $\| \Sc_2\|_{\Tc_1}$. 
    We let $\Sc$ be the concatenation of $\Sc_1$ and $\Sc_2$, and $\mc C = \mc C_1$.
    By the above, $\Phi(\Tc_2)\leq \Phi(\Tc)$, and the total amortized running time is $\OO_c(1)$, which is also an upper bound on $\Sc$.

\paragraph{Protrusion decomposition.}
    Next, we show how to maintain a corresponding protrusion decomposition.
    We will construct $\tilde \Tc$ from $\Tc$ in time $\OO_c(\|\Sc\|_\Tc)$, giving the claimed amortized time bounds. 
    As a first step, we let $\tilde T = T$ and $\bag(t) = V(\torso(t))$ for all $t \in V(\tilde T)$.
    Note that this already is a tree decomposition corresponding to $\Tc$ with bag size $|\bag(t)| \leq \OO_c(1)$ for every non-root node $t \in V(\tilde T)\setminus \{r\}$, and $|\bag(r)| \leq \Delta(r) c$ for the root.

    We now make the protrusions of $\tilde \Tc$ binary.
    For each non-root node $t \in V(\tilde T) \setminus \{r\}$ with $d > 2$ children $c_1, \ldots, c_d$ and some parent $p$, we remove $t$ and its incident edges from $\tilde T$, and add new nodes $t_1,\ldots,t_{d-1}$, together with edges $t_it_{i+1}$ for $i \in \{1,\ldots,d-2\}$, edges $t_ic_i$ for $i \in \{1,\ldots,d-1\}$, and edges $pt_1$ and $t_{d-1}c_d$.
    The new nodes are given bags $\bag(t_i) = \bag(t)$.
    We say that the $t_i \in V(\tilde T)$ \emph{correspond} to $t \in V(T)$.
    We observe that $(\tilde T$, $\bag)$ can be maintained locally for all non-root nodes, in the sense that if $\Tc$ is transformed to $\Tc'$ by a sequence of basic rotations $\Sc$, then the only parts of $\tilde \Tc$ that need to be recomputed are the nodes corresponding to the nodes in $\trace_{\Tc'}(\Sc)$.
    This can be done in time $\OO_c(|\trace_{\Tc'}(\Sc)|) \leq \OO_c(\|\Sc\|_{\Tc})$.
    For the root $r$, we cannot recompute the entire $\bag(r)$ within the claimed time bounds.
    Instead, we can use the sequence $\mc C$ of basic hypergraph operations computed for each operation above.
    For every hypergraph operation $\deleteVertex(v)$ or $\addVertex(v)$ in $\mc C$, we update the balanced binary search tree representing $\bag(r)$, by deleting or inserting the corresponding vertex $v$.
    The running time of this update is $\OO_c(\log|G|)$.

    We know that $\tilde \Tc$ is normal by~\Cref{lem:corresponding-is-normal}.
    For $\tilde \Tc$ to be an annotated normal $(\Delta(r) c, \OO_c(1))$-protrusion decomposition corresponding to $\Tc$, the only thing left is to maintain the $\edges$ function.
    Let us consider the process of making $\tilde T$ binary, i.e., when we split $t$ into nodes $t_1,\dots t_{d-1}$.
    Since $\bag(t_i)=\bag(t)$ for all $i\in \{1,d-1\}$, we have that $\edges(t_i)=\emptyset$ for all $i\in \{2,d-1\}$ and $\edges_{\tilde T}(t_1)=\edges_T(t)$.
    Furthermore, note that each time we split a single high-degree node into a chain, this does not change the value of $\edges$ on the untouched nodes.
    Thus, to compute the $\edges$ function, it suffices to compute it for nodes in $T$.

    For this, we will maintain the function $\EL \colon V(T) \to 2^{E(G)}$ where, for each node $t \in V(T)$, $\EL(t)$ is the set of all edges $uv\in E(G)$ with $u, v \in V(\torso(t))$ and $e_{uv} \in \Lc[t]$.
    The function is represented as follows:
    for each non-root node $t\in V(T)\setminus \{r\}$, the set $\EL(t)$ is stored as a linked list; and
    for the root $r$, the set $\EL(r)$ is stored as a balanced binary search tree $\mc B$ containing all edges $e\in \EL(r)$.
    Observe that for any internal node $t\in \Vint(T)$ with children $c_1,\dots, c_d$, we have 
    \begin{align}
        \EL(t) = \bigcup_{i} \{e\in \EL(c_i) \mid V(e)\subseteq \adh(tc_i)\}. \label{eq:EL}
    \end{align} 
    The only parts of $\tilde \Tc$ where $\EL$ must be recomputed are the nodes in $\trace_{\Tc'}(\Sc)$, since no torso outside the trace changes during our operations.
    We start by updating $\EL(t)$ for each non-root node $t\in \trace_{\Tc'}(\Sc)\setminus\{r\}$ in a bottom-up manner. 
    Since $|\EL(t)|\leq \binom{|\torso(t)|}{2} \leq \OO_c(1)$, and given $\EL(c)$ for all children $c$ of $t$, we can recompute $\EL(t)$ in time $\OO_c(1)$ using \Cref{eq:EL}.
    This gives a total running time of $\OO_c(|\trace_{\Tc'}(\Sc)|) \leq \OO_c(\|\Sc\|_{\Tc})$.
    
    We now update $\EL(r)$ represented by the binary balanced search tree $\mc B$. 
    For any $e\in E(\torso(r))$ that also appears in the torso after applying the sequence $\Sc$ (i.e. where the sequence $\mc C$ does not contain $\deleteEdge(e)$), we have ensured above, that the corresponding subtree is unchanged. 
    Since the adhesion is the same, it contributes the same edges to $\EL_{\Tc'}(r)$ in \Cref{eq:EL}.
    Hence, only the children inserted or deleted by the sequence $\mc C$ of basic hypergraph operations must be considered.
    As $\|\mc C\|\leq \OO_c(1)$, there are at most $\OO_c(1)$ insertions or deletions, and we can maintain the balanced binary search tree $\mc B$ in time $\OO_c(\log|G|)$.

    Now, we compute $\edges(t)$ for each $t\neq r$ as the set of edges in $\EL(t)$ which are not also in $\EL(p)$, where $p$ is the parent of $t$.
    If $p$ is not the root, this takes $\OO_c(1)$ time, and otherwise, we have to query $\mc B$ resulting in a running time of $\OO_c(\log|G|)$.
    Since at most $\OO_c(1)$ nodes in $\trace_{\Tc'}(\Sc)$ have the root as a parent, the total running time for this step is $\OO_c(\log|G|+ \|\Sc\|_\Tc)$, which remains within the claimed bounds.
    Note that $\edges(r)=\EL(r)$, so the balanced binary search tree $\mc B$ serves as the representation of $\edges(r)$. 
    Each time we insert or delete an edge in $\mc B$, we can append it to a list, which is returned. This list has size $\OO_c(1)$.

\paragraph{Run of an automaton.}
    Finally, we prove that we can maintain a run of an automaton.
    Recall from earlier in this proof, that when $\Tc$ is modified by a sequence of basic rotations $\Sc$, we rebuild a prefix of $\tilde \Tc$ of size at most $\OO_c(\|\Sc\|_\Tc)$.
    This can be done in a bottom-up fashion, so that when we build a non-root node $t$, we can also compute it's automaton state $\run_\autom^{\tilde \Tc}(t)$ incurring an additional factor $\tau$ on the running times.

    If the automaton state at the node $t \in \chd(r)$ corresponding to an edge $e \in E(\torso(r))$ is modified by an operation, this happens because the subtree of the superbranch decomposition rooted at $t$ was modified.
    If so, the sequence of hypergraph operations $\mc C$ will contain the hypergraph operation $\deleteHyperedge(e)$, as guaranteed above.
\end{proof}

%% file: main-data-structure.tex
\section{Assembling the main data structure}\label{sec:main_data_structure}

In this section, we finally assemble our main data structure that maintains a protrusion decomposition of a dynamic topological-minor-free graph $G$, thus proving \Cref{theo:main}. The data structure from \Cref{lem:dynamic_prot} already performs all the required operations for the insertion and deletion of edges and vertices and maintains a superbranch decomposition and a corresponding protrusion decomposition, where the subtrees below the root satisfy all our requirements. However, this data structure allows the root degree, and thus the first parameter of the corresponding protrusion decomposition, to be arbitrarily large, and instead provides a $\merge$ operation, which allows the user to decrease the root degree themselves -- if they know a suitable set of ``mergeable'' root-children.

In \Cref{sec:existence}, for topological-minor-free graphs with a small treewidth-$\cTwMod$-modulator, we proved the existence of such a set of ``mergeable'' root-children, identified by their corresponding hyperedges in the torso of the root, if the degree of the root $r$ of the superbranch decomposition is high enough. That is, we showed that there exists a set $B \subseteq E(\torso(r))$ of bounded size such that $B \expand \Tc$ has small boundary and internal treewidth and all internal components of $B$ have the same boundary as $B$. Then, in \Cref{sec:local_search}, we constructed a data structure that, when applied to $\torso(r)$ and provided with an oracle that decides whether the internal treewidth of such a set $B \expand \Tc$ is ``small'', returns such a set $B$ (if it exists).

Essentially, there are two major things left to be done in this section. For one, we need to construct such an internal treewidth oracle. 
Secondly, the $\merge$ operation from \Cref{lem:dynamic_prot} is restricted to well-linked sets of root-children to keep the superbranch decomposition downwards well-linked. Thus, we need to find a well-linked subset of this ``mergeable'' set $B$. 

We start with the internal treewidth oracle.
To obtain this, we use a tree decomposition automaton for computing internal treewidth, stated in \Cref{lem:internal_treewidth_automaton}.
The proof of this lemma is given in \Cref{sec:automata}.
It is based on the observation that the Bodlaender-Kloks dynamic programming for treewidth~\cite{Bodlaender_Kloks_1996} can be seen as an tree decomposition automaton (used also e.g. by~\cite{Korhonen_Majewski_Nadara_Pilipczuk_Sokołowski_2023}), hence the name $\mc{IBK}$.

\begin{restatable}{lemma}{internalTreewidthAutomaton}\label{lem:internal_treewidth_automaton}
    For every pair of integers $k \leq \ell$, there is a tree decomposition automaton $\mc{IBK}_{k,\ell}$ of width $\ell$ with the following property: For any boundaried graph $G$ and its annotated boundaried tree decomposition $\Tc = (T,\bag,\edges)$ of width at most $\ell$, $\mc{IBK}_{k,\ell}$ accepts $\Tc$ if and only if the internal treewidth $\itw(G)$ of $G$ is at most $k$. The state space of $\mc{IBK}_{k,\ell}$ is of size $\OO_{k,\ell}(1)$ and can be computed in time $\OO_{k,\ell}(1)$. The evaluation time of $\mc{IBK}_{k,\ell}$ is $\OO_{k,\ell}(1)$ as well.
\end{restatable}

Using \Cref{lem:internal_treewidth_automaton} we can construct the internal treewidth oracle.

\begin{lemma}\label{lem:internal_tw_oracle}
    Let $G$ be a graph, $\cNewSbdProt$ an integer, and $\Tc' = (T', \Lc')$ a downwards well-linked superbranch decomposition of $\Hc(G)$ with adhesion size $\cSbdProt$ and root $r$.
    Let further $\Tc = (T,\bag,\edges)$ be an annotated normal $(k,\cFinalProt)$-protrusion decomposition with root $r$ that corresponds to $\Tc$, and $\run^{\Tc}_{\autom}$ the run of the automaton $\autom \coloneq \mc{IBK}_{\cNewSbdProt,(\cFinalProt+1) \cdot \cUpperBoundSize}$ from \Cref{lem:internal_treewidth_automaton} on $\Tc$.
    
    Suppose that the representations of $\Tc'$ and $\Tc$ are stored, and for each $t \in \chd(r)$ we can query the state $\run^{\Tc}_{\autom}(t)$ in $\OO_{\cNewSbdProt,\cFinalProt,\cUpperBoundSize}(1)$ time.
    Then, given a set $S \subseteq E(\torso(r))$ of size $|S| \leq \cUpperBoundSize$, we can in time $\OO_{\cSbdProt,\cFinalProt,\cNewSbdProt,\cUpperBoundSize}(\log \Delta(r))$ decide whether $\itw(S \expand \Tc') \leq \cNewSbdProt$.
\end{lemma}
\begin{proof}
    Let $S = \{e_1, \dots, e_q\} \subseteq E(\torso(r))$ be a set of hyperedges and for $1 \leq i \leq q$, let $t_i \in \chd(r)$ be the root-child that corresponds to $e_i$, i.e., $\{e_i\} \expand \Tc' = \Lc'[t_i]$, and thus $V(e_i) = \bd(\Lc'[t_i])$. 
    Since $\Tc$ corresponds to $\Tc'$, for each $t \in \chd(r)$, we have $V(\Lc'[t]) = \bigcup_{t' \in \desc(t)} \bag(t')$.
    Let $G_S$ be the boundaried graph obtained from $G[\bigcup_{t_i} V(\Lc'[t_i])]$ by setting $\bd(S)$ as the boundary.
    We observe that $\itw(G_S) = \itw(S \expand \Tc')$, so our goal is to determine whether $\itw(G_S) \le \cNewSbdProt$.
    We do this by computing the root-state of $\autom$ on a boundaried tree decomposition of $G_S$.

    Without time constraints, this could be done simply by combining the subtrees $\Tc_i$ of $\Tc$ rooted at the nodes $t_i$ to form an annotated boundaried tree decomposition of $G_S$ and running $\autom$ on it.
    However, building this entire tree decomposition would exceed the runtime constraints.
    But as we already know the states of $\autom$ on the roots of the subtrees $\Tc_i$, it is sufficient to construct a prefix $\Tc_S$ of this tree decomposition that contains the root $t_i$ from each $\Tc_i$ but nothing below.
    To get the same result when running the automaton $\autom$ only on this prefix, we initialize the nodes $t_i$ with their state $\run_{\autom}^{\Tc}(t_i)$, which we can query in $\OO_{\cNewSbdProt,\cFinalProt,\cUpperBoundSize}(1)$ time for each $t_i$.

    In detail, to construct $\Tc_S = (T_S,\bag_S,\edges_S)$, we first take the nodes $t_i$ for $1 \leq i \leq q$ and add them to $V(T_S)$.
    Then, we insert $q$ additional nodes $s_1,\dots,s_{q-1}$ and $r_S$, where $r_S$ will be the root of $\Tc_S$. For $i = 2,\dots,q-1$, we add the edges $s_i s_{i-1}$ and $s_i t_{i+1}$ to $T_S$. Additionally, we add the edges $s_1 t_1$, $s_1 t_2$, and $r_S s_{q-1}$ to $\Tc_S$.
    We define the $\bag_S$ function as follows in a bottom-up way:
    \begin{itemize}
        \item First, we set $\bag_S(t_i) \coloneq \bag(t_i)$ for every $1 \leq i \leq q$.
        \item Then, we set $\bag_S(s_1) \coloneq V(e_1) \cup V(e_2)$.
        \item Next, we recursively set $\bag_S(s_i) \coloneq \bag_S(s_{i-1}) \cup V(e_{i+1})$ for $2 \leq i \leq q-1$.
        \item Lastly, we set $\bag_S(r_S) \coloneq \bd(S \expand \Tc') = \bd(S) \subseteq \bag_S(s_{q-1})$.
    \end{itemize}
    For the last step, by \Cref{lem:compute-boundary}, $\bd(S)$ can be computed in time $\OO(|S|\cdot\cSbdProt) = \OO_{\cSbdProt,\cUpperBoundSize}(1)$.
    We observe that the width of $\Tc_S$ is at most $(\cFinalProt+1) \cdot |S|-1 \leq \cFinalProt \cdot \cUpperBoundSize$, so the automaton $\autom$ of width $(\cFinalProt+1) \cdot \cUpperBoundSize$ can indeed be applied to $\Tc_S$.

    Recall that the $\edges_S$ function is uniquely determined by $G_S$ and the pair $(T_S,\bag_S)$.
    Since $\bag_S(t_i)=\bag(t_i)$, and the adhesion of $t_i$ to its parent $p$ of $T_S$ satisfy $\adh_S(t_i p) = \adh(t_i r)= V(e_i)$, it follows that $\edges_S(t_i)=\edges(t_i)$.
    All other bags of $\Tc_S$ except $\bag_S(s_{q-1})$ and $\bag_S(r_S)$ are subsets of their parents, so they have empty $\edges_S$.
    For $\bag(s_{q-1})$, note that all its vertices are in $\bag(r)$, implying that $\edges_S(s_{q-1})\subseteq\edges(r)$.
    Hence, we can compute $\edges_S(s_{q-1})$ by finding all edges $uv$ with $u,v \in \bag_S(s_{q-1}) \setminus \bag_S(r_S)$ and $uv\in \edges(r)$. For this, we query the balanced binary search tree representing $\edges(r)$ for each pair $u,v\in \bag_S(s_{q-1})  \setminus \bag_S(r_S)$ in $(|S|+\alpha)^{\OO(1)}\cdot \OO(\log(\Delta(r)\cdot \alpha) = \OO_{\alpha,\sigma}(\log \Delta(r))$ time.    
    In the same manner, we can compute $\edges_S(r_S)$ within the same time bound.
    
    We observe that that $\Tc_S$ can indeed be extended to an annotated boundaried tree decomposition $\Tc_S^*$ of $G_S$ by attaching the subtrees $\Tc_i$ to $t_i$ for each $1 \leq i \leq q$.
    Note that for each $t_i$, we have $\run_{\autom}^{\Tc}(t_i) = \run_{\autom}^{\Tc^*_S}(t_i)$.
    Therefore, to determine the state of $\autom$ on the node $r_S$ of $\Tc_S^*$, it suffices to initialize the states of $t_i$ to be $\run_{\autom}^{\Tc}(t_i)$, and compute the states of the rest of the nodes of $\Tc_S$ in a bottom-up manner with the transitions of $\autom$.
    As $\Tc_S$ has $\OO(\cUpperBoundSize)$ nodes and the evaluation time of $\autom$ is $\OO_{\cNewSbdProt,\cFinalProt,\cUpperBoundSize}(1)$, this runs in $\OO_{\cNewSbdProt,\cFinalProt,\cUpperBoundSize}(1)$ time.
    Whether $\itw(G_S) \le \cFinalProt$ can then be determined by whether the state of $r_S$ is accepting.
\end{proof}

From now on, we can assume that we have an oracle that decides for every small set $B \subseteq E(\torso(r))$ whether the internal treewidth of $B \expand \Tc$ is small. Thus, we can give this oracle to the data structure from \Cref{lem:dynamic_local_search}, which will then help us find a ``mergeable'' set $B$ of hyperedges in $\torso(r)$ whenever needed.

Secondly, we need to show that we can always find a well-linked subset of this set $B$.
To do this, \Cref{lem:partition_well_linked} provides us with an algorithm that partitions a set $B$ of hyperedges of a hypergraph $G$ into well-linked subsets. The algorithm is from~\cite[Lemma 5.3]{Korhonen_2025}, although they have a runtime dependency on the size $|G|$ of the hypergraph $G$ instead of only the size $|B|$ of $B$ in their paper. The reason for this is that they use a version of the following algorithm, also with a runtime dependency on $|G|$ instead of $|B|$, as a subprocedure (see~\cite[Lemma 3.1]{Korhonen_2025}). However, as shown in \cite[Lemma 7.1]{Korhonen_2024}, this can be made to run in $2^{\OO(\lambda(B))} \cdot \rank(G)^2 \cdot |B|$ time.

\begin{lemma}[\cite{Korhonen_2024}]\label{lem:well_linked_witness}
    Let $G$ be a hypergraph of rank $r$ whose representation is already stored. There is an algorithm that, given a set $B \subseteq E(G)$, in time $2^{\OO(\lambda(B))} \cdot r^2 \cdot |B|$ either
    \begin{itemize}
        \item correctly concludes that $B$ is well-linked or
        \item returns a bipartition $(B_1,B_2)$ of $B$ so that $\lambda(B_i) < \lambda(B)$ for both $i \in [2]$.
    \end{itemize}
\end{lemma}

Now, if we plug this improved version into the algorithm from~\cite[Lemma 5.3]{Korhonen_2025}, we obtain an algorithm that in time $2^{O(\lambda(B))} \cdot \rank(G)^2 \cdot |B|$ partitions a set $B$ of hyperedges into well-linked sets.

\begin{lemma}[Based on {\cite[Lemma 5.3]{Korhonen_2025}}]\label{lem:partition_well_linked}
    Let $G$ be a hypergraph of rank $r$ whose representation is already stored. There is an algorithm that, given a set $B \subseteq E(G)$, in time $2^{O(\lambda(B))} \cdot r^2 \cdot |B|$ returns a partition $\mathfrak{B}$ of $B$ into at most $|\mathfrak{B}| \leq 2^{\lambda(B)}$ sets, so that each $X\in \mathfrak{B}$ is well-linked in $G$.
\end{lemma}
\begin{proof}
    We maintain a partition $\mathfrak{B}$ of $B$, initially setting $\mathfrak{B} = \{B\}$. For each part $X \in \mathfrak{B}$, we repeatedly apply the algorithm from \Cref{lem:well_linked_witness} to test whether $X$ is well-linked, and if not, replace $X$ by the two sets $X_1,X_2$ returned by the algorithm, where $(X_1,X_2)$ is a bipartition of $X$ with $\lambda(X_i) < \lambda(X)$ for both $i \in [2]$.

    Throughout this process, we always have $\sum\limits_{X \in \mathfrak{B}} 2^{\lambda(X)} \leq 2^{\lambda(B)}$, while the size of $\mathfrak{B}$ increases strictly in each step. Therefore, it must terminate within $2^{\lambda(B)}$ iterations, with $|\mathfrak{B}| \leq 2^{\lambda(B)}$. As the algorithm from \Cref{lem:well_linked_witness} runs in time $2^{\OO(\lambda(X))} \cdot r^2 \cdot |X|$, and $\lambda(X) \leq \lambda(B)$ for each part $X \in \mathfrak{B}$, the total running time is at most $2^{\OO(\lambda(B))} \cdot r^2 \cdot |B|$.
\end{proof}

Finally, we are ready to prove \Cref{theo:main}. For this, we only need to combine our data structures from \Cref{sec:local_search} and \Cref{sec:splay} together with the results from this section and \Cref{sec:existence} and adjust them with the right parameters.
We prove the following stronger and more technical version of \Cref{theo:main}.

\begin{theorem}\label{thm:technical_main}
There is a data structure, that is initialized with a graph $H$, an integer $\cTwMod$, and an empty graph $G$, and supports updating $G$ under the assumption that it remains $H$-topological-minor-free.
The data structure maintains $G$, $\Hc(G)$, and an annotated normal $(\OO_{H,\cTwMod}(\optTwMod(G)), \OO_{H,\cTwMod}(1))$-protrusion decomposition $\Tc = (T,\bag,\edges)$ of $G$ rooted at node $r$ together with the hypergraph $\torso(r)$.
The data structure supports the following operations.
    \begin{itemize}
        \item $\init(H, \cTwMod)$: We initialize the data structure with $H$ and $\cTwMod$, and an empty graph $G$. Afterwards, $\torso(r)$ is an empty hypergraph. Runs in $\OO_{H,\cTwMod}(1)$  time.
        \item $\addVertex(v)$: Given a new vertex $v \notin V(G)$, add $v$ to $G$.
        \item $\deleteVertex(v)$: Given an isolated vertex $v \in V(G)$, remove $v$ from $G$.
        \item $\addEdge(e)$: Given a new edge $e \in \binom{V(G)}{2} \setminus E(G)$, add $e$ to $G$.
        \item $\deleteEdge(e)$: Given an edge $e \in E(G)$, remove $e$ from $G$.
    \end{itemize}
    Each operation except for $\init$ runs in $\OO_{H,\cTwMod}(\log |G|)$ amortized time.
    For each operation, the changes to $\torso(r)$ can be described as a sequence $\mc C$ of hypergraph operation of size $\OO_{H,\cTwMod}(1)$, which is returned.
    Changes to $\edges(r)$ can be described as a list of insertions and deletions of size at most $\OO_{H,\cTwMod}(1)$, which is returned as well.

    Moreover, if upon initialization the data structure is provided a tree decomposition automaton $\autom$ with evaluation time $\tau$, then a run of $\autom$ on each protrusion is maintained, incurring an additional factor $\tau$ on the running times.
    If the automaton state at the node corresponding to an edge $e \in E(\torso(r))$ is modified by an operation, the corresponding sequence $\mc C$ will contain the operation $\deleteHyperedge(e)$.
\end{theorem}
\begin{proof}
    We begin by outlining the structure of the proof.
    We use the data structure from \Cref{lem:dynamic_prot} with $c=\OO_{\cTwMod,H}(1)$ to maintain the graph $G$, the support hypergraph $\Hc(G)$,  a downwards well-linked superbranch decomposition $\Tc'$ of $\Hc(G)$ rooted at a node $r'$, and an annotated normal $(\Delta(r')\cdot \OO_{\cTwMod,H}(1), \OO_{\cTwMod,H}(1))$-protrusion decomposition $\Tc$ corresponding to $\Tc'$ under the required operations. 
    We now also want to keep the degree of the root $r'$ of the superbranch decomposition $\Tc'$ bounded, that is, $\Delta(r') \leq \OO_{H,\cTwMod}(\optTwMod(G))$, implying that $\Tc$ is a  $(\OO_{H,\cTwMod}(\optTwMod(G)), \OO_{H,\cTwMod}(1))$-protrusion decomposition as required. To that end, we let the data structure from \Cref{lem:dynamic_local_search} maintain a representation of the hypergraph $\torso(r')$ -- the same $\torso(r')$ that is also maintained as part of $\Tc'$ by the data structure from \Cref{lem:dynamic_prot}. This will help us find a set of ``mergeable'' root-children (see below for a formal definition of ``mergeable''), on which we can apply the $\merge$ operation from \Cref{lem:dynamic_prot} to decrease the root degree whenever possible.

    Let $\cNewSbdProt = \cNewSbdProt(H,\cTwMod)$ be the integer from \Cref{mergeable-children}, and $b \le \OO_{H,\cTwMod}(1)$ the factor hidden by the $\OO_{H,\cTwMod}$-notation in the upper bound $|B| \le \OO_{H,\cTwMod}(\cSbdProt)$ of \Cref{mergeable-children}.
    We use the data structure of \Cref{lem:dynamic_prot} with $c=\cSbdProt \coloneq \max(6\cNewSbdProt + 3, b)$. It maintains a downwards well-linked superbranch decomposition $\Tc' = (T',\Lc')$ rooted at $r'$ with adhesion size $\cSbdProt$, as well as an annotated normal $(\Delta(r') \cdot \cSbdProt, \OO_{\cSbdProt}(1))$-protrusion decomposition $\Tc = (T,\bag,\edges)$ that corresponds to $\Tc'$.
    Note that the adhesion size bound implies that the rank of the hypergraph $\torso(r')$ is at most $\cSbdProt$.

    Let $\ell = \OO_{\alpha}(1)$ be the upper bound for the width of the protrusions of $\Tc$ maintained by \Cref{lem:dynamic_prot}.
    At the initialization of \Cref{lem:dynamic_prot}, we provide it also the automaton $\mc{IBK}_{\cNewSbdProt, \ell \cdot \alpha^2}$ from \Cref{lem:internal_treewidth_automaton} and the automaton $\autom$ that was given to this lemma.
    Note that \Cref{lem:dynamic_prot} can be generalized from maintaining just one automaton to maintaining multiple by considering the automaton formed as the cartesian product of them.
    Maintaining these automata causes an extra $\OO_{H,\cTwMod}(1) + \tau$ factor in the running time.

    Let $\cSbdRootDeg = \cSbdRootDeg(H,\cTwMod,\cSbdProt)$ be the integer from \Cref{mergeable-children} such that the following holds: If $\Delta(r') \geq \cSbdRootDeg \cdot \optTwMod(G)$, then there exists a set $B \subseteq E(\torso(r'))$ with 
    \begin{itemize}
        \item $\lambda(B) \leq \cNewSbdProt$,
        \item $\itw(B \expand \Tc) \leq \cNewSbdProt$,
        \item $2^{\cNewSbdProt+2} \leq |B| \leq b \cdot \cSbdProt \leq \cSbdProt^2$, and
        \item for every internal component $B'$ of $B$, we have $\bd(B') = \bd(B)$.
    \end{itemize}
    We call a set $B \subseteq \torso(r')$ of hyperedges \emph{semi-mergeable} if $\lambda(B) \le \cNewSbdProt$, $\itw(B \expand \Tc') \leq \cNewSbdProt$, and for every internal component $B'$ of $B$ it holds that $\bd(B') = \bd(B)$. These properties are closely related to what we actually want to be able to apply the $\merge$ operation, as we will show in \Cref{claim:semi_mergeable_to_mergeable}. Specifically, we call a set $B \subseteq E(\torso(r'))$ \emph{mergeable} if it is well-linked and $\wl(B \expand \Tc) \leq \cSbdProt$. Furthermore, a set $C \subseteq \chd(r')$ of root-children corresponding to such a set $B$ is also called \emph{(semi-)mergeable}. We now show that it is indeed a small step from being semi-mergeable to being mergeable.

    \begin{claim}\label{claim:semi_mergeable_to_mergeable}
        There exists an algorithm that, given a semi-mergeable set $B$ of size $|B| \geq 2^{\cNewSbdProt+1}$, finds a mergeable subset $B' \subseteq B$ of size $|B'| \geq 2$ in time $2^{\OO(\lambda(B))} \cdot \cSbdProt^2 \cdot |B|$.
    \end{claim}
    \begin{claimproof}
        By \Cref{lem:partition_well_linked}, there exists an algorithm that in time $2^{\OO(\lambda(B))} \cdot \cSbdProt^2 \cdot |B|$ finds a partition $\mathfrak{B}$ of $B$ into at most $|\mathfrak{B}| \leq 2^{\lambda(B)} \leq 2^{\cNewSbdProt}$ sets so that each $X \in \mathfrak{B}$ is well-linked in $\torso(r')$. By the pigeonhole principle, there is a well-linked set $B' \in \mathfrak{B}$ of size $|B'| \geq |B| / 2^{\cNewSbdProt} \geq 2$. By \Cref{lem:well_linked_torso}, $B' \expand \Tc'$ is well-linked in $\Hc(G)$ as well. Moreover, by \Cref{lem:well_linked_number_treewidth} we have $\wl(B' \expand \Tc') \leq 3 \cdot (\tw(B' \expand \Tc') + 1) \leq 3(\tw(B \expand \Tc') + 1) \leq 3 \cdot (\itw(B \expand \Tc') + \lambda(B \expand \Tc') + 1) \leq 6\cNewSbdProt + 3 \le \cSbdProt$, where $\tw(B' \expand \Tc') \leq \tw(B \expand \Tc')$ holds since $B' \expand \Tc' \subseteq B \expand \Tc$.
    \end{claimproof}

    Our goal is to keep $\Delta(r') < \cSbdRootDeg \cdot \optTwMod(G)$ at all times (without knowing the exact value of $\optTwMod(G)$). To that end, whenever possible, we want to find a semi-mergeable set of hyperedges in $\torso(r')$, turn it into a mergeable set using \Cref{claim:semi_mergeable_to_mergeable}, and then apply the $\merge$ operation from \Cref{lem:dynamic_prot} on the corresponding set of root-children.    
    To find such a set of semi-mergeable hyperedges in $\torso(r')$, we use the data structure from \Cref{lem:dynamic_local_search} with the parameters $s_1 \coloneq 2^{\cNewSbdProt+2}$, $s_2 \coloneq \cSbdProt^2$, and $k \coloneq \cNewSbdProt$ and let it maintain a representation of $\torso(r')$.

    To this end, we provide the data structure with a $\cSbdProt^2$-bounded oracle $\oracle$ so that for every set $S \subseteq E(\torso(r'))$ of size $|S| \leq \cSbdProt^2$ the oracle is satisfied if and only if $\itw(S \expand \Tc') \leq \cNewSbdProt$.
    This oracle is realized by the algorithm from \Cref{lem:internal_tw_oracle} together with the maintained automaton $\mc{IBK}_{\cNewSbdProt, \ell \cdot \alpha^2}$.
    Given a set $S \subseteq E(\torso(r'))$ of size $|S| \leq \cSbdProt^2$, the algorithm from \Cref{lem:internal_tw_oracle} decides whether $\itw(S \expand \Tc') \leq \cNewSbdProt$ in time $\OO_{\cSbdProt,\cNewSbdProt}(\log \Delta(r')) \le \OO_{H,\cTwMod}(\log |G|)$.
    Now, the $\query$ operation of \Cref{lem:dynamic_local_search} returns a semi-mergeable set $B \subseteq E(\torso(r'))$ of size $2^{\cNewSbdProt+1} \leq |B| \leq \cSbdProt^2$ or concludes that no semi-mergeable set of size $2^{\cNewSbdProt+2} \leq |B| \leq \cSbdProt^2$ exists.
    Furthermore, the operations of \Cref{lem:dynamic_local_search} run in time $\OO_{\cSbdProt}(\log |E(\torso(r'))|) \le \OO_{H,\cTwMod}(\log |G|)$.

    We are now ready to describe the operations of our data structure. 
    For the initialization, we compute the integers $\cNewSbdProt$ and $\cSbdProt$ given $H$ and $\cTwMod$.
    Then, we initialize the data structure from \Cref{lem:dynamic_prot} with the integer $\cSbdProt$ in $\OO(1)$ time.
    Note that $\torso(r')$ is an empty hypergraph at the beginning since $G$ is an empty graph. Thus, we can also initialize the data structure from \Cref{lem:dynamic_local_search}, which will keep a representation of $\torso(r')$, in $\OO(1)$ time.

    For an update operation (insert/delete, vertex/edge), we apply the respective operation from \Cref{lem:dynamic_prot} to update our representations of $G$, $\Hc(G)$, $\Tc$, and $\Tc'$ in $\OO_{\cSbdProt}(\log |G|)$ amortized time. This returns a sequence of operations of size at most $\OO_{\cSbdProt}(1)$, which describes how the hypergraph $\torso(r')$ before the operation can be turned into the hypergraph $\torso(r')$ after the operation.
    For every basic hypergraph operation in this sequence, we apply the respective operation from \Cref{lem:dynamic_local_search} to update the representation of $\torso(r')$, each in at most $\OO_{H,\cTwMod}(\log |G|)$ time, so in total this takes $\OO_{H,\cTwMod}(\log |G|)$ amortized time.
    At the end, we return forward this sequence of operations, and also the changes to the set $\edges(r)$.

    To restore our invariant $\Delta(r') < \cSbdRootDeg \cdot \optTwMod(G)$ we do the following. After every update to $G$, we repeatedly apply the $\query$ operation from \Cref{lem:dynamic_local_search} to find a semi-mergeable set $B \subseteq E(\torso(r'))$ of size $2^{\cNewSbdProt+1} \leq |B| \leq \cSbdProt^2$. Then, we use \Cref{claim:semi_mergeable_to_mergeable} to translate $B$ into a mergeable set $B'$ of size $2 \leq |B'| \leq \cSbdProt^2$, and finally apply the $\merge$ operation on the corresponding set of root-children to decrease $\Delta(r')$.
    Let $\cUBTorsoOps$ denote the maximum difference between the root degree $\Delta(r')$ before and after the current operation and observe that $\cUBTorsoOps \le \OO_{\cSbdProt}(1)$.
    We repeat these steps at most $\cUBTorsoOps + 2\cSbdRootDeg$ times, or until the $\query$ operation returns that no semi-mergeable set $B$ of size $2^{\cNewSbdProt+2} \leq |B| \leq \cSbdProt^2$ exists.
    
    By \Cref{mergeable-children}, if no such $B$ exists, we know that $\Delta(r') < \cSbdRootDeg \cdot \optTwMod(G)$. Furthermore, we need to argue that $\cUBTorsoOps + 2\cSbdRootDeg$ applications of the $\merge$ operation are enough to keep $\Delta(r') < \cSbdRootDeg \cdot \optTwMod(G)$, even if the data structure from \Cref{lem:dynamic_local_search} could find more semi-mergeable sets. For this, we first argue that the value of $\optTwMod(G)$ changes by at most two during each update.

    \begin{claim}\label{claim:tw_mod_increase}
        Let $G$ be a graph, $v \in V(G)$ a vertex, and $e \in E(G)$ an edge. For every integer $\cTwMod$, we have
        \begin{itemize}
            \item $\optTwMod(G \setminus \{v\}) \leq \optTwMod(G) \leq \optTwMod(G \setminus \{v\}) + 1$ and
            \item $\optTwMod(G \setminus \{e\}) \leq \optTwMod(G) \leq \optTwMod(G \setminus \{e\}) + 2$.
        \end{itemize}
    \end{claim}
    \begin{claimproof}
        Recall that for a graph $G$, $\optTwMod(G)$ denotes the size of a smallest treewidth-$\cTwMod$-modulator of $G$, i.e., the size of a smallest set $X \subseteq V(G)$ with $\tw(G \setminus X) \leq \cTwMod$. It is clear from the definition that deleting a vertex or an edge cannot increase the size of the smallest treewidth modulator, so $\optTwMod(G \setminus \{v\}) \leq \optTwMod(G)$ and $\optTwMod(G \setminus \{e\}) \leq \optTwMod(G)$. Furthermore, if $X$ is a treewidth-$\cTwMod$-modulator of $G \setminus \{v\}$ for a vertex $v \in V(G)$ (or of $G \setminus \{e\}$ for an edge $e = uv \in E(G)$), then $X \cup \{v\}$ (or $X \cup \{u,v\}$) is a treewidth-$\cTwMod$-modulator of $G$, and thus $\optTwMod(G) \leq \optTwMod(G \setminus \{v\}) + 1$ (and $\optTwMod(G) \leq \optTwMod(G \setminus \{e\}) + 2$). 
    \end{claimproof}

    It follows that the target value $\cSbdRootDeg \cdot \optTwMod(G)$ for the degree $\Delta(r')$ changes by at most $2 \cSbdRootDeg$ during each update. Thus, after one update and before applying the $\merge$ operation, we have $\Delta(r') < \cSbdRootDeg \cdot (\optTwMod(G) + 2) + \cUBTorsoOps$, so we need to decrease $\Delta(r')$ by at most $\cUBTorsoOps + 2\cSbdRootDeg$. As we only apply the $\merge$ operation to sets $B'$ with $|B'| \geq 2$, each application decreases $\Delta(r')$ by at least one. Thus, $\cUBTorsoOps + 2\cSbdRootDeg \le \OO_{H,\cTwMod}(1)$ applications are enough to restore the invariant $\Delta(r') < \cSbdRootDeg \cdot \optTwMod(G)$.
    
    As argued before, the updates themselves can be performed in $\OO_{\cSbdProt}(\log |G|)$ amortized time, so the total amortized running time of each update operation is $\OO_{H,\cTwMod}(\log |G|)$.
    Then, the $\query$ operation from \Cref{lem:dynamic_local_search} runs in $\OO_{H,\cTwMod}(\log |G|)$ time.
    Lastly, each of the $\OO_{H,\cTwMod}(1)$ applications of $\merge(C_{B'})$, where $C_{B'}$ is the set of root-children corresponding to the mergeable set $B'$, runs in $\OO_{\cSbdProt}(1)$ amortized time.
    Thus, the total amortized running time of each update operation is $\OO_{H,\cTwMod}(\log |G|)$.
\end{proof}

%% file: kernelization.tex
\section{Dynamic kernelization}\label{sec:kernelization}

In this section, we finally show how our data structure can be used to maintain kernels for various problems on sparse graphs. Specifically, in \Cref{sec:protrusion_replacement}, we give a tree decomposition automaton for the protrusion replacement, based on the techniques introduced by~\cite{Bodlaender_Fomin_Lokshtanov_Penninkx_Saurabh_Thilikos_2016}. Then, in \Cref{sec:kernelization}, we combine this automaton with the data structure from \Cref{theo:main} to obtain a dynamic kernelization algorithm for all treewidth-bounding problems with FII on $\cmso$-definable topological-minor-free graph classes, thus proving \Cref{thm:kernelization}.

We note to the reader that in this section we will have constants that depend on a non-computable way from other constants, typically from the problem $\Pi$, the graph class $\mc G$, and the treewidth-modulator constant $\cTwMod$.
In this case, we do not use the $\OO$-notation with subscript, as it implies a computable dependence, but instead treat $\Pi$, $\mc G$, and $\cTwMod$ as true constants in the sense that numbers depending on them can be hidden by $\OO(1)$.

\subsection{The protrusion replacement automaton}\label{sec:protrusion_replacement}

Before constructing the protrusion replacement automaton, we first need to define the graphs by which the protrusions will be replaced. Specifically, for every boundaried graph, we define a smaller ``equivalent'' graph. For this, recall the definition of $\equiv_\Pi$, that is, for two boundaried graphs $G_1,G_2 \in \mc F$, we have $G_1 \equiv_\Pi G_2$ if and only if they have the same label set, i.e.,  $\labelSet(G_1) = \labelSet(G_2)$, and there exists a transposition constant $\Delta = \Delta_{\Pi}(G_1,G_2) \in \mathbb{Z}$ such that
\begin{equation}\label{eq:FII}
    \forall (F,k) \in \mc F \times \mathbb{Z}: (F \oplusu G_1,k) \in \Pi \Leftrightarrow (F \oplusu G_2, k + \Delta) \in \Pi.
\end{equation}

It is not clear if such a value $\Delta_{\Pi}(G_1,G_2)$ is uniquely defined, and it in fact is not unique in all cases.
However, we will fix a unique definition for $\Delta_{\Pi}(G_1,G_2)$ as follows.

We say a boundaried graph $G$ is \emph{monotone} if either $(F \oplusu G,k) \in \Pi$ for every $(F,k) \in \mc F \times \mathbb{Z}$ (\emph{positive monotone}) or $(F \oplusu G,k) \notin \Pi$ for every $(F,k) \in \mc F \times \mathbb{Z}$ (\emph{negative monotone}). We observe that the set of all positive monotone graphs and the set of all negative monotone graphs with the same label set form their own equivalence classes under $\equiv_\Pi$, respectively. For two boundaried graphs $G_1,G_2$ with the same label set that are both either positive or negative monotone, we set $\Delta_\Pi(G_1,G_2) \coloneq 0$. For all non-monotone graphs $G_1,G_2$, we now show that the transposition constant $\Delta_\Pi(G_1,G_2)$ is uniquely determined by \Cref{eq:FII}.

\begin{lemma}\label{lem:transposition_unique}
    Let $\Pi$ be an  parameterized graph problem and let $G_1, G_2 \in \mc F$ be two non-monotone boundaried graphs with $G_1 \equiv_\Pi G_2$. Then, there is exactly one constant $\Delta = \Delta_\Pi(G_1,G_2)$ such that for every $F \in \mc F$ and every $k \in \mathbb{Z}$ it holds that $(F \oplusu G_1,k) \in \Pi \Leftrightarrow (F \oplusu G_2, k + \Delta) \in \Pi$.
\end{lemma}
\begin{proof}
    Without loss of generality (by flipping the $\in \Pi$ relation if needed), we assume that $\Pi$ is a parameterized graph problem with $(G,k) \notin \Pi$ for every graph $G$ and $k < 0$.
    Let us fix a boundaried graph $F \in \mc F$ such that there is an integer $k$ with $(F \oplusu G_1,k) \in \Pi$.  This graph $F$ exists since $G_1$ is non-monotone. Note that $G_1 \equiv_\Pi G_2$ implies that there is an integer $k'$ such that $(F \oplusu G_2,k') \in \Pi$.
    For $i \in [2]$, let $k_i^*$ denote the smallest integer such that $(F \oplusu G_i, k_i^*) \in \Pi$. Since $(G,k) \notin \Pi$ for $k < 0$, these integer $k_i^*$ exist and are non-negative. We now show that $\Delta \coloneq k_2^* - k_1^*$ is the only integer that satisfies \Cref{eq:FII}.

    Suppose for the sake of contradiction, there is an integer $\Delta' \neq \Delta$ that satisfies \Cref{eq:FII}. That is, for every $k \in \mathbb{Z}$, we have $(F \oplusu G_1,k) \in \Pi \Leftrightarrow (F \oplusu G_2, k + \Delta') \in \Pi$, or equivalently, for every $k \in \mathbb{Z}$, we have $(F \oplusu G_1, k - \Delta') \in \Pi \Leftrightarrow (F \oplusu G_2, k) \in \Pi$. First, suppose $\Delta' < \Delta$. Since $(F \oplusu G_1,k_1^*) \in \Pi$, it follows that $(F \oplusu G_2, k_1^* + \Delta') \in \Pi$, where $k_1^* + \Delta' < k_1^* + \Delta = k_2^*$, which contradicts $k_2^*$ being the smallest integer $k$ with $(F \oplusu G_2,k) \in \Pi$. Now, suppose $\Delta' > \Delta$. Then, $(F \oplusu G_2,k_2^*) \in \Pi$ implies that $(F \oplusu G_1, k_2^* - \Delta')$, which is again a contradiction to the minimality of $k_1^*$ since $k_2^* - \Delta' < k_2^* - \Delta = k_1^*$.

    Thus, there is no $\Delta' \neq \Delta$ that satisfies \Cref{eq:FII}. Since $G_1 \equiv_\Pi G_2$, there exists an integer that satisfies \Cref{eq:FII}, so this must be $\Delta$.
\end{proof}

We further observe that for any two boundaried graphs $G_1, G_2 \in \mc F$ with $G_1 \equiv_\Pi G_2$, we have $\Delta_\Pi(G_1,G_2) = - \Delta_\Pi(G_2,G_1)$. Moreover, we show that the function $\Delta_\Pi$ satisfies the following \emph{transitivity} condition.

\begin{lemma}\label{lem:transposition_transitive}
    Let $G_1,G_2,G_3 \in \mc F$ be three boundaried graphs that are equivalent under $\equiv_\Pi$. Then, $\Delta_\Pi(G_1,G_3) = \Delta_\Pi(G_1,G_2) + \Delta_\Pi(G_2,G_3)$.
\end{lemma}
\begin{proof}
    If $G_1$, $G_2$, and $G_3$ are monotone, this lemma trivially holds, so assume that they are non-monotone.
    By \Cref{eq:FII}, for every boundaried graph $F \in \mc F$ and every integer $k \in \mathbb{Z}$, we have the following:
    \[(F \oplusu G_1,k) \in \Pi \Leftrightarrow (F \oplusu G_3,k + \Delta_\Pi(G_1,G_3)) \in \Pi,\]
    and further
    \begin{align*}
    (F \oplusu G_1,k) \in \Pi &\Leftrightarrow (F \oplusu G_2, k + \Delta_\Pi(G_1,G_2)) \in \Pi\\
    &\Leftrightarrow (F \oplusu G_3, k + \Delta_\Pi(G_1,G_2) + \Delta_\Pi(G_2,G_3)) \in \Pi
    \end{align*}
    Due to the uniqueness of $\Delta_\Pi$, as shown in \Cref{lem:transposition_unique}, it follows that $\Delta_\Pi(G_1,G_3) = \Delta_\Pi(G_1,G_2) + \Delta_\Pi(G_2,G_3)$.
\end{proof}

For a graph class $\mc G$ and two boundaried graphs $G_1,G_2 \in \mc F$, we further define $G_1 \equiv_{\mc G} G_2$ if and only if $\labelSet(G_1) = \labelSet(G_2)$ and for every $F \in \mc F$, $F \oplusu G_1 \in \mc G$ if and only if $F \oplusu G_2 \in \mc G$.
It is easy to see that $\equiv_{\mc G}$ is an equivalence relation.
For $\cmso$-definable graph classes, it was proven in~\cite{Bodlaender_Fomin_Lokshtanov_Penninkx_Saurabh_Thilikos_2016} (based on~\cite{Courcelle_1990,DBLP:journals/jal/ArnborgLS91,DBLP:journals/algorithmica/BoriePT92,abrahamson1993finite,Downey_Fellows_1999}) that it has a finite number of equivalence classes.

\begin{lemma}[{\cite[Lemma~3.2]{Bodlaender_Fomin_Lokshtanov_Penninkx_Saurabh_Thilikos_2016}}]\label{lem:cmso_finite_equivalence_classes}
    Let $\mc G$ be a $\cmso$-definable graph class and let $I \subseteq \ZI$ be a finite set. Then, the number of equivalence classes of $\equiv_{\mc G}$ that are subsets of $\mc F_I$ is finite.
\end{lemma}

For a parameterized graph problem $\Pi$, a graph class $\mc G$, and two boundaried graphs $G_1,G_2 \in \mc F$, we say $G_1 \equiv_{\Pi,\mc G} G_2$ if and only if $G_1 \equiv_\Pi G_2$ and $G_1 \equiv_{\mc G} G_2$. We observe that $\equiv_{\Pi,\mc G}$ is again an equivalence relation and the number of equivalence classes of $\equiv_{\Pi,\mc G}$ is at most the number of equivalence classes of $\equiv_{\Pi}$ times the number of equivalence classes of $\equiv_{\mc G}$.
Thus, if $\Pi$ has FII and $\mc G$ is $\cmso$-definable, the number of equivalence classes of $\equiv_{\Pi,\mc G}$ that are subsets of $\mc F_I$ for a finite $I \subseteq \ZI$ is finite.

Now, for every equivalence class of $\equiv_{\Pi,\mc G}$, we want to fix a ``good'' representative, with which we can replace all subgraphs that are in the same equivalence class. 
Consider a equivalence class $\mc C$ of $\equiv_{\Pi,\mc G}$. A graph $G \in \mc C$ is called a \emph{progressive representative} of $\mc C$ if for any graph $G' \in \mc C$ it holds that $\Delta_{\Pi}(G', G) \leq 0$, or equivalently, $\Delta_\Pi(G,G') \geq 0$. For the equivalence relation $\equiv_\Pi$ it is known that such a progressive representative always exists for every equivalence class, regardless of whether $\Pi$ has FII~\cite{Bodlaender_Fomin_Lokshtanov_Penninkx_Saurabh_Thilikos_2016,kernelization-book}.

\begin{lemma}[{\cite[Lemma 16.11]{kernelization-book}}]\label{lem:existence_progressive_representatives}
    Let $\Pi$ be a parameterized graph problem. Then, each equivalence class of $\equiv_\Pi$ has a progressive representative.
\end{lemma}

We now have to prove the same thing for the equivalence relation $\equiv_{\Pi,\mc G}$. Since the proof is almost the same as for~\cite[Lemma 16.11]{kernelization-book}, we move it to \Cref{sec:missingproofs}.

\begin{restatable}{lemma}{progRepr}
    Let $\mc G$ be a graph class and let $\Pi$ be a parameterized graph problem. Then, each equivalence class of $\equiv_{\Pi,\mc G}$ has a progressive representative.
\end{restatable}

We remark that there might be more than one progressive representative for each equivalence class, however, we can assume that we fix one of them.
For this, recall that by definition, all graphs within the same equivalence class of $\equiv_\Pi$ (and thus, $\equiv_{\Pi,\mc G}$) have the same label set. Thus, every equivalence class of $\equiv_\Pi$ is a subset of $\mc F_I$ for some finite $I \subseteq \ZI$, where $\mc F_I$ is the class of all boundaried graphs with label set $I$. So, for a fixed $I \subseteq \ZI$, we can find one progressive representative for each equivalence class that is a subset of $\mc F_I$.

That is, for each $I \subseteq \ZI$, we define $\mc S_I$ to be a set containing exactly one smallest (by number of vertices) progressive representative of each equivalence class of $\equiv_{\Pi,\mc G}$ that is a subset of $\mc F_I$. We also define $S_{\subseteq I} \coloneq \bigcup_{I' \subseteq I} \mc S_{I'}$ and $c_{\Pi,\mc G}(t) \coloneq \max\{|V(G)| \colon G \in \mc S_{\subseteq [t]}\}$.
Also, let $\varphi_{\Pi,\mc G,I} \colon \mc F_{\subseteq I} \to \mc S_{\subseteq I}$ be the function that maps every boundaried graph $G \in \mc F_{\subseteq I}$ to the unique smallest progressive representative $\varphi_{\Pi,\mc G,I}(G) \in \mc S_{\subseteq I}$ of $G$'s equivalence class under $\equiv_{\Pi,\mc G}$.

Note that the value of $c_{\Pi,\mc G}(t)$ only depends on $\Pi$, $\mc G$, and $t$, and that $c_{\Pi,\mc G}(t) \geq t$. Recall that, by definition of $\varphi_{\Pi,\mc G,[t]}$ and $\Delta_\Pi$, it holds for every graph $G \in \mc F_{\subseteq[t]}$ and every pair $(F,k) \in \mc F \times \mathbb{Z}$ that $(F \oplusu G,k) \in \Pi$ if and only if $(F \oplusu \varphi_{\Pi,\mc G,[t]}(G), k + \Delta_\Pi(G,\varphi_{\Pi,\mc G,[t]}(G)) \in \Pi$.

Before constructing the protrusion replacement automaton, we need to take care of one last technicality. Note that the replacement function $\varphi_{\Pi,\mc G,[t]}$ replaces boundaried graphs that, by definition, have a fixed labeling. In particular, the replacement is different for different boundaried graphs with the same underlying non-boundaried graph and boundary set, but different labelings. This is a problem since we cannot maintain a global labeling for our dynamic protrusion decomposition, but need to recompute labelings on the fly. To do this, we now show that first changing the labeling and then replacing is equivalent to first replacing and then changing the labeling.

\begin{lemma}\label{lem:permute_labels}
    Let $\Pi$ be a parameterized graph problem and let $t$ be an integer.
    Let $X = (G_X,B,\Lambda)$ be a $t$-boundaried graph, and let $\Lambda' \colon B \to [t]$ be another labeling. Let $X' = (G_X,B,\Lambda')$, $Y \coloneq \varphi_{\Pi,\mc G,[t]}(X) = (G_X',B,\Lambda)$, and $Y' = (G_X',B,\Lambda')$. 
    Then, we have that $X' \equiv_\Pi Y'$ and $\Delta_\Pi(X',Y') = \Delta_\Pi(X,Y)$.
\end{lemma}
\begin{proof}
    Let $(F',k) \in \mc F \times \mathbb{Z}$ be a fixed, but arbitrary pair. We show that $(F' \oplusu X',k) \in \Pi$ if and only if $(F' \oplusu Y', k + \Delta_\Pi(X,Y)) \in \Pi$. Note that for every such boundaried graph $F' \in \mc F$, there is a graph $\tilde{F'} \in \mc F$ with $\labelSet(\tilde{F'}) = \labelSet(X) (= \labelSet(Y))$ such that $\tilde{F'} \oplusu X' = F' \oplusu X'$ and $\tilde{F'} \oplusu Y' = F' \oplusu Y'$. This boundaried graph $\tilde{F'}$ can be obtained from $F'$ by removing boundary vertices with labels in $\labelSet(F') \setminus \labelSet(X')$ from the boundary of $F'$ and inserting isolated vertices for every label in $\labelSet(X) \setminus \labelSet(F')$.
    Thus, without loss of generality, we can assume that $\labelSet(F') = \labelSet(X') (= \labelSet(Y'))$, and moreover, $F' = (G_F,B,\Lambda')$ for some graph $G_F$.
    We consider the boundaried graph $F = (G_F,B,\Lambda)$ and observe that $F \oplusu X = F' \oplusu X'$ and $F \oplusu Y = F' \oplusu Y'$.
    Thus, we have the following relations:
    \begin{align*}
        &(F' \oplusu X',k) \in \Pi\\
        \Leftrightarrow &(F \oplusu X, k) \in \Pi &  (F' \oplusu X' = F \oplusu X)\\
        \Leftrightarrow &(F \oplusu Y, k + \Delta_\Pi(X,Y)) \in \Pi & (Y = \varphi_{\Pi,\mc G,[t]}(X))\\
        \Leftrightarrow & (F' \oplusu Y', k + \Delta_\Pi(X,Y)) \in \Pi & (F \oplusu Y = F' \oplusu Y').
    \end{align*}
    It follows that $X' \equiv_\Pi Y'$ and $\Delta_\Pi(X',Y') = \Delta_\Pi(X,Y)$.
\end{proof}

Our goal is now to provide a tree decomposition automaton that, given a graph $G$ together with an annotated tree decomposition $\Tc = (T,\bag,\edges)$, for every node $x \in V(T)$ computes the replacement $\varphi_{\Pi,\mc G,[t]}(X)$, where $X = (G_x,\bag(x),\Lambda)$ and $\Lambda$ is some labeling, which we fix in the following.

We assume without loss of generality that the set of vertices $V(G)$ of our graph is totally ordered by a relation $\prec$ that can be computed in $\OO(1)$ time, e.g., by indexing them by integers.
Let $B \subseteq V(G)$ be a set of size $|B| \leq t$. An injective labeling $\Lambda \colon B \to [t]$ is called \emph{in-order} if for any two vertices $u,v \in B$, $u \prec v$ implies that $\Lambda(u) < \Lambda(v)$, and $\max_{v \in B} \Lambda(v) = |B|$. Note that there is exactly one injective in-order labeling for every such set $B \subseteq V(G)$, which we denote by $\Lambda_B^*$. Then, for a node $x \in V(T)$, we say that $(G_x,\bag(x),\Lambda_{\bag(x)}^*)$ is the boundaried graph that \emph{corresponds to} $x$. We are now ready to construct the protrusion replacement automaton that, given a graph $G$ together with a tree decomposition $\Tc$ of width $t-1$, for every node $x \in V(T)$ with corresponding boundaried graph $X = (G_x,\bag(x),\Lambda_{\bag(x)}^*)$ computes the progressive representative $\varphi_{\Pi,\mc G,[t]}(X)$ and the transposition constant $\Delta_\Pi(X,\varphi(X))$.


\begin{lemma}\label{lem:protrusion_replacement_automaton}
    Let $\mc G$ be a $\cmso$-definable graph class, $\Pi$ a parameterized graph problem that has FII, and $t \in \mathbb{Z}_{\ge 1}$.
    There exists a tree decomposition automaton $\autom$ of width $t-1$ such that the following holds:
    Let $G$ be a graph together with an annotated tree decomposition $\Tc = (T,\bag,\edges)$ of width $t-1$. Then, given the run $\run_{\autom}^{\Tc}(x)$ of $\autom$ on a node $x \in V(T)$, we can determine the progressive representative $\varphi_{\Pi,\mc G,[t]}(X)$ together with the transposition constant $\Delta_\Pi(X,\varphi_{\Pi,\mc G,[t]}(X))$ in $\OO(1)$ time, where $X = (G_x,\bag(x),\Lambda_B^*)$ is the boundaried graph that corresponds to $x$. The evaluation time $\tau$ of $\autom$ is a constant depending on $\mc G$, $\Pi$ and $t$.
\end{lemma}
\begin{proof}
    We start by arguing that it suffices to construct an automaton $\mc A = (Q, \iota, \emptyset, \delta)$ that assumes that the annotated tree decomposition is nice.
    To remove this assumption, we construct a modified transition map $\delta'$ as follows.
    Consider a node $x\in V(T)$ with two children $y_1$ and $y_2$ (if $x$ has one child, this is handled analogously).
    We conceptually replace the edge $xy_i$ by a path that successively introduces and forgets one node at a time to transform $\bag(y_i)$ into $\bag(x)$ for both $i\in \{1,2\}$.
    This makes it ``locally nice'' and we can use the transition map $\delta$ on these paths:
    Given the state of the children, $\run_{\autom}^{\Tc}(y_1)$ and $\run_{\autom}^{\Tc}(y_2)$, we apply $\delta$ along the nodes of the path, until the resulting run $\run_{\autom}^{\Tc}(x)$ is obtained. 
    Since the length of each path is at most $|\bag(x)|\leq t$, the evaluation time $\tau$ still only depends on $\mc G$, $\Pi$, and $t$.
    
    We will now describe a dynamic programming scheme that, given a graph $G$ with an annotated nice tree decomposition $\Tc = (T,\bag,\edges)$ of width $t-1$, for every node $x \in V(T)$ with corresponding boundaried graph $X = (G_x,\bag(x),\Lambda_{\bag(x)}^*)$, determines the tuple $(\varphi_{\Pi,\mc G,[t]}(X),\Delta_\Pi(G_x,\varphi_{\Pi,\mc G,[t]}(X))$.
    In the following, we drop the subscripts and set $\varphi(G) \coloneq \varphi_{\Pi,\mc G,[t]}(G)$ and $\Delta(G) \coloneq \Delta_\Pi(G,\varphi(G))$ for every boundaried graph $G \in \mc F_{\subseteq [t]}$.
    The dynamic programming algorithm can be implemented as a tree decomposition automaton $\autom = (Q,\emptyset,\iota, \delta)$, whose states correspond to pairs $(G,\Delta) \in \mc S_{\subseteq [t]} \times \mathbb{Z}$, such that for every node $x \in V(T)$ the run $\run_{\autom}(x)$ corresponds to the pair $(\varphi(X),\Delta(X))$.
    Specifically, we start with the base case of the dynamic program, i.e., the leaf nodes of the tree decomposition (corresponding to the initial mapping $\iota$ in the automaton). Then, for the three cases of a non-leaf node $x$ in a nice tree decomposition -- forget, introduce, and join -- we describe how we can determine the pair $(\varphi(X),\Delta(X))$ from the state(s) of the child(ren) (corresponding to the transition mapping $\delta$ of $\autom$).
    For this, we can assume that for every boundaried graph $G \in \mc F_{\subseteq [t]}$ on at most $2 \cdot c_{\Pi,\mc G}(t)$ vertices, we can determine $\varphi(G)$ together with the value $\Delta(G)$ by hardcoding this mapping in our algorithm's source code, similar to the protrusion replacement algorithm from Fomin, Lokshtanov, Saurabh, and Zehavi (see Lemma 16.18 in~\cite{kernelization-book}).

    First, recall that for a node $x \in V(T)$, the graph $G_x$ is defined as $G_x = (V_x,E_x)$, where $V_x = \bigcup_{y \in \desc(x)} \bag(y)$ and $E_x = \bigcup_{y \in \desc(x)} \edges(y)$. Further recall that $\edges(x)$ contains exactly the edges of $G$, for which $x$ is the shallowest node that contains both endpoints, so $\edges(x) = \emptyset$, unless the parent of $x$ is a forget node. In particular, for a node $x$ with children $y$ and $z$ (where $y = z$ if $x$ has only one child), the question whether $E_x \supsetneq E_y \cup E_z$ depends solely on the parent of $x$. For the sake of simplification, we want $E_x = E_y \cup E_z$ for every forget, introduce, and join node $x$ with children $y$ and $z$ (again, $y = z$ if $x$ is a forget or introduce node). To achieve this, we introduce a new type of node, called \emph{edges node}, where every edges node $x$ has exactly one child $y$ with $\bag(x) = \bag(y)$. From now on, we assume that the only child of every forget node is an edges node. Then, we have $\edges(x) \neq \emptyset$ only if $x$ is an edges node. This assumption does not change our requirements on $\Tc$, since processing an edges node before every forget node can also be implemented as part of the forget node procedure. We now describe how we process these edges nodes before taking care of the standard nodes: leaf, forget, introduce, and join.

    \paragraph{Edges node:} Let $x \in V(T)$ be an edges node with child $y$. Note that $\bag(x) = \bag(y)$, so let $X = (G_x,\bag(x),\Lambda_{\bag(x)}^*)$ and $Y = (G_y,\bag(x),\Lambda_{\bag(x)}^*)$ be the boundaried graph that corresponds to $x$ and $y$. Let further $G_B = (\bag(x),\edges(x))$, and $B = (G_B,\bag(x),\Lambda_{\bag(x)}^*)$. Then, $X = B \oplusb Y$. We now show that $\varphi(X) = \varphi(B \oplusb \varphi(Y))$ and $\Delta(X) = \Delta(Y) + \Delta(B \oplusu \varphi(Y))$. From our dynamic program, we know the pair $(\varphi(Y),\Delta(Y))$. Note that $|V(B \oplusb \varphi(Y))| = |V(\varphi(Y))| \leq c_{\Pi,\mc G}(t)$, so we can explicitly construct the boundaried graph $B \oplusb \varphi(Y)$ and then look up the progressive representative $\varphi(B \oplusb \varphi(Y))$ in our hardcoded mapping.

    \begin{claim}
        $\varphi(X) = \varphi(B \oplusb \varphi(Y))$ and $\Delta(X) = \Delta(Y) + \Delta(B \oplusb \varphi(Y))$.
    \end{claim}
    \begin{claimproof}
        Let $(F,k) \in \mc F \times \mathbb{Z}$ be an arbitrary pair. We apply the definitions of $\varphi$ and $\Delta$ and use the associativity of $\oplusb$ to obtain the following:
        \begin{align*}
            &(F \oplusb X,k) \in \Pi\\
            \Leftrightarrow &(F \oplusb (B \oplusb Y),k) \in \Pi & (X = B \oplusb Y)\\
            \Leftrightarrow &((F \oplusb B) \oplusb Y,k) \in \Pi & \text{(associativity)}\\
            \Leftrightarrow &((F \oplusb B) \oplusb \varphi(Y),k + \Delta(Y)) \in \Pi & (F \oplusb B \in \mc F, \text{ definitions of } \varphi,\Delta).
        \end{align*}
        It follows that $X \equiv_\Pi B \oplusb \varphi(Y)$, so $\varphi(X) = \varphi(B \oplusb \varphi(Y))$, and further that $\Delta(X,B \oplusb \varphi(Y)) = \Delta(Y)$. Applying the transitivity of $\Delta_\Pi$, which we showed in \Cref{lem:transposition_transitive}, we obtain that $\Delta(X) = \Delta_\Pi(X,\varphi(X)) = \Delta_\Pi(X,B \oplusb \varphi(Y)) + \Delta_\Pi(B \oplusb \varphi(Y),\varphi(X)) = \Delta(Y) + \Delta(B \oplusb \varphi(Y))$.
    \end{claimproof}

    From now on we can assume for every non-edges node $x$ with children $y$ and $z$ (where $y = z$ if $x$ has only one child), we have $\edges(x) = \emptyset$, and thus $E(G_x) = E(G_y) \cup E(G_z)$.
    
    \paragraph{Leaf node:} Let $x \in V(T)$ be a leaf node and let $X = (G_x,\bag(x),\Lambda)$ be the corresponding boundaried graph. Then, $|V(G_x)| \leq t \leq c_{\Pi(t)}$, and thus we can find the pair $(\varphi(X),\Delta(X))$ by looking it up in our hardcoded mapping.
    
    \paragraph{Forget node:} Let $x \in V(T)$ be a forget node with child $y$ and let $v$ be the forgotten vertex. Let $X = (G_x,\bag(x),\Lambda_{\bag(x)}^*)$ and $Y = (G_y,\bag(y),\Lambda_{\bag(y)}^*)$ be the boundaried graphs that correspond to $x$ and $y$, respectively. Note that we have $\bag(y) = \bag(x) \cup \{v\}$.
    We define the injective labeling $\Lambda^* \colon \bag(y) \to [t]$ as the extension of $\Lambda_{\bag(x)}^*$ to $\bag(y)$, i.e.,
    \begin{equation*}
        \Lambda^*(u) = \begin{cases}
            \Lambda_{\bag(x)}^*(u) & \text{if } u \in \bag(x)\\
            \max\limits_{u' \in \bag(x)} \Lambda_{\bag(x)}^*(u') + 1 & \text{if } u = v
        \end{cases}.
    \end{equation*}
    Let then $Y' = (G_y,\bag(y),\Lambda^*)$, $\tilde Y = \varphi(Y) = (G_y',\bag(y),\Lambda_{\bag(y)}^*)$, $\tilde Y' = (G_y',\bag(y),\Lambda^*)$, and $\tilde Y'' = (G_y', \bag(x),\Lambda_{\bag(x)}^*)$. Note that by \Cref{lem:permute_labels}, we have $Y' \equiv_\Pi \tilde Y'$ and $\Delta_\Pi(Y',\tilde Y') = \Delta_\Pi(Y,\tilde Y) = \Delta(Y)$.

    From our dynamic program, we know the boundaried graph $\tilde Y = \varphi(Y)$ together with the integer value $\Delta(Y)$. In order to show that we can determine $\varphi(X)$ and $\Delta(X)$ directly from $(\varphi(Y),\Delta(Y))$, we first show that $\varphi(X) = \varphi(\tilde Y'')$ and $\Delta(X) = \Delta(Y) + \Delta(\tilde Y'')$. Then, since $G_y' \in \mc S_{\subseteq [t]}$, and thus, $|V(G_y')| \leq c_{\Pi,\mc G}(t)$, we can look up the replacement $\varphi(\tilde Y'') = \varphi(X)$ together with the transposition constant $\Delta(\tilde Y'')$ in our hardcoded mapping. Adding the known value $\Delta(Y)$ to this transposition constant, we obtain $\Delta(X)$.


    \begin{claim}
        $\varphi(X) = \varphi(\tilde Y'')$ and $\Delta(X) = \Delta(Y) + \Delta(\tilde Y'')$.
    \end{claim}
    \begin{claimproof}
        Let $(F,k) \in \mc F \times \mathbb{Z}$ be a fixed, but arbitrary pair. We show that $(F \oplusu X, k) \in \Pi$ if and only if $(F \oplusu \tilde Y'', k + \Delta(Y)) \in \Pi$. If $\labelSet(F) \neq \labelSet(X)$ ($= \labelSet(\tilde Y'')$), we are done, so we can assume that $F = (G_F,\bag(x),\Lambda_{\bag(x)}^*)$ for some graph $G_F$. We consider the boundaried graph $F^* = (G_F \cup \{v\},\bag(y),\Lambda^*)$ and observe that $F \oplusu X = F^* \oplusu Y'$. For this, note that due to the edges nodes, we have $G_x = G_y$, and all edges that are incident to $v$, both in $F \oplusu X$ and $F^* \oplusu Y'$, are in $E(G_x) = E(G_y)$. Similarly, we also have $F \oplusu \tilde Y'' = F^* \oplusu \tilde Y'$. It follows that
        \begin{align*}
            &(F \oplusu X,k) \in \Pi\\
            \Leftrightarrow &(F^* \oplusu Y',k) \in \Pi & (F \oplusu X = F^* \oplusu Y')\\
            \Leftrightarrow &(F^* \oplusu \tilde Y', k + \Delta(Y)) \in \Pi & (Y' \equiv_\Pi \tilde Y' \text{ and } \Delta_\Pi(Y',\tilde Y') = \Delta(Y))\\
            \Leftrightarrow &(F \oplusu \tilde Y'', k + \Delta(Y)) \in \Pi & (F^* \oplusu \tilde Y' = F \oplusu \tilde Y'').
        \end{align*}
        It follows that $X \equiv_\Pi \tilde Y''$, so $\varphi(X) = \varphi(\tilde Y'')$, and $\Delta_\Pi(X,\tilde Y'') = \Delta(Y)$. Due to the transitivity of $\Delta_\Pi$ (see \Cref{lem:transposition_transitive}), this implies that $\Delta(X) = \Delta_\Pi(X,\varphi(X)) = \Delta_\Pi(X,\tilde Y'') + \Delta_\Pi(\tilde Y'', \varphi(X)) = \Delta(Y) + \Delta(\tilde Y'')$.
    
    \end{claimproof}
    
    \paragraph{Introduce node:} Let $x \in V(T)$ be an introduce node with child $y$ and let $v$ be the introduced vertex. Let $X = (G_x,\bag(x),\Lambda_{\bag(x)}^*)$ and $Y = (G_y,\bag(y),\Lambda_{\bag(y)}^*)$ be the boundaried graphs that correspond to $x$ and $y$, respectively. Note that we have $\bag(y) = \bag(x) \setminus \{v\}$.
    Conversely to the forget node, we define the injective labeling $\Lambda^* \colon \bag(y) \to [t]$ as the restriction of $\Lambda_{\bag(x)}^*$ to $\bag(y)$, i.e., $\Lambda^* = \Lambda_{\bag(x)}^*\restriction_{\bag(y)}$.
    Let $Y' = (G_y,\bag(y),\Lambda^*)$, $\tilde Y = \varphi(Y) = (G_y',\bag(y),\Lambda_{\bag(y)}^*)$, $\tilde Y' = (G_y',\bag(y),\Lambda^*)$, and $\tilde Y'' = (G_y' \cup \{v\}, \bag(x),\Lambda_{\bag(x)}^*)$. Note that by \Cref{lem:permute_labels}, we have $Y' \equiv_\Pi \tilde Y'$ and $\Delta_\Pi(Y',\tilde Y') = \Delta_\Pi(Y,\tilde Y) = \Delta(Y)$.

    Again, we show that $\varphi(X) = \varphi(\tilde Y'')$ and $\Delta(X) = \Delta(Y) + \Delta(\tilde Y'')$. Knowing the pair $(\varphi(Y),\Delta(Y))$, and in particular the graph $G_y'$, where $|V(G_y' \cup \{v\})| \leq c_{\Pi,\mc G}(t) + 1$, we can look up the replacement $\varphi(\tilde Y'')$ together with the transposition constant $\Delta(\tilde Y'')$.

    \begin{claim}
        $\varphi(X) = \varphi(\tilde Y'')$ and $\Delta(X) = \Delta(Y) + \Delta(\tilde Y'')$.
    \end{claim}
    \begin{claimproof}
        Let $(F,k) \in \mc F \times \mathbb{Z}$ be a fixed, but arbitrary pair. We show that $(F \oplusu X, k) \in \Pi$ if and only if $(F \oplusu \tilde Y'', k + \Delta(Y)) \in \Pi$. If $\labelSet(F) \neq \labelSet(X)$ ($= \labelSet(\tilde Y'')$), we are done, so we can assume that $F = (G_F,\bag(x),\Lambda_{\bag(x)}^*)$ for some graph $G_F$. We consider the boundaried graph $F^* = (G_F,\bag(y),\Lambda^*)$ and observe that $F \oplusu X = F^* \oplusu Y'$. For this, note that there are no edges incident to $v$ in $G_x$, since $v$ is the introduced vertex and $\edges(x) = \emptyset$. Similarly, since $v$ is isolated in $\tilde Y''$, we also have that $F \oplusu \tilde Y'' = F^* \oplusu \tilde Y'$. It follows that
        \begin{align*}
            &(F \oplusu X,k) \in \Pi\\
            \Leftrightarrow &(F^* \oplusu Y',k) \in \Pi & (F \oplusu X = F^* \oplusu Y')\\
            \Leftrightarrow &(F^* \oplusu \tilde Y', k + \Delta(Y)) \in \Pi & (Y' \equiv_\Pi \tilde Y' \text{ and } \Delta_\Pi(Y',\tilde Y') = \Delta(Y))\\
            \Leftrightarrow &(F \oplusu \tilde Y'', k + \Delta(Y)) \in \Pi & (F^* \oplusu \tilde Y' = F \oplusu \tilde Y'').
        \end{align*}
        It follows that $X \equiv_\Pi \tilde Y''$, so $\varphi(X) = \varphi(\tilde Y'')$, and further $\Delta_\Pi(X,\tilde Y'') = \Delta(Y)$. Due to the transitivity of $\Delta_\Pi$ (see \Cref{lem:transposition_transitive}), this implies that $\Delta(X) = \Delta_\Pi(X,\varphi(X)) = \Delta_\Pi(X,\tilde Y'') + \Delta_\Pi(\tilde Y'', \varphi(X)) = \Delta(Y) + \Delta(\tilde Y'')$.
    \end{claimproof}

    \paragraph{Join node:} Let $x \in V(T)$ be a join node with children $y$ and $z$ and let $X = (G_x,\bag(x),\Lambda_x)$, $Y = (G_y,\bag(y),\Lambda_y)$, and $Z = (G_z,\bag(z),\Lambda_z)$ be the corresponding boundaried graphs. Then, $X = Y \oplusb Z$. Again, from our dynamic program, we know the pairs $(\varphi(Y),\Delta(Y))$ and $(\varphi(Z),\Delta(Z))$, and we need to determine $(\varphi(X),\Delta(X))$. For this, we show that $\varphi(X) = \varphi(\varphi(Y) \oplusb \varphi(Z))$ and $\Delta(X) = \Delta(Y) + \Delta(Z) + \Delta(\varphi(Y) \oplusb \varphi(Z))$. Note that $\varphi(Y)$ and $\varphi(Z)$ have at most $c_{\Pi,\mc G}(t)$ vertices each, and thus $\varphi(Y) \oplusb \varphi(Z)$ has at most $2 \cdot c_{\Pi,\mc G}(t)$ vertices. Hence, we can look up the graph $\varphi(\varphi(Y) \oplusb \varphi(Z))$ together with the value $\Delta(\varphi(Y) \oplusb \varphi(Z))$ in our hardcoded mapping, and thus determine $(\varphi(X),\Delta(X))$ from the state of the two children $y$ and $z$ of $x$.
    \begin{claim}
        $\varphi(X) = \varphi(\varphi(Y) \oplusb \varphi(Z))$ and $\Delta(X) = \Delta(Y) + \Delta(Z) + \Delta(\varphi(Y) \oplusb \varphi(Z))$.
    \end{claim}
    \begin{claimproof}
        Let $(F,k) \in \mc F \times \mathbb{Z}$ be an arbitrary pair. Similarly to the edges node, we apply the definitions of $\varphi$ and $\Delta$ and use the associativity and commutativity of $\oplusb$ to obtain the following equivalences.
        
        \begin{align*}
            &(F \oplusb X, k) \in \Pi\\
            \Leftrightarrow &(F \oplusb (Y \oplusb Z), k) \in \Pi && (X = Y \oplusb Z)\\
            \Leftrightarrow &((F \oplusb Y) \oplusb Z,k) \in \Pi && \text{(associativity)}\\
            \Leftrightarrow &((F \oplusb Y) \oplusb \varphi(Z), k + \Delta(Z)) \in \Pi && (F \oplusb Y \in \mc F \text{, definitions of } \varphi, \Delta)\\
            \Leftrightarrow &((F \oplusb \varphi(Z)) \oplusb Y, k + \Delta(Z)) \in \Pi && \text{(associativity, commutativity)}\\
            \Leftrightarrow &((F \oplusb \varphi(Z)) \oplusb \varphi(Y), k + \Delta(Z) + \Delta(Y)) \in \Pi && (F \oplusb \varphi(Z) \in \mc F, \text{ definitions of } \varphi, \Delta)\\
            \Leftrightarrow &(F \oplusb (\varphi(Y) \oplusb \varphi(Z)), k + \Delta(Y) + \Delta(Z)) \in \Pi && \text{(associativity, commutativity)}
        \end{align*}
        
        It follows that $X$ and $\varphi(Y) \oplusb \varphi(Z)$ are equivalent under $\equiv_\Pi$, so $\varphi(X) = \varphi(\varphi(Y) \oplusb \varphi(Z))$. Moreover, it follows that $\Delta_\Pi(X,\varphi(Y) \oplusb \varphi(Z)) = \Delta(Y) + \Delta(Z)$. Together with \Cref{lem:transposition_transitive}, we have $\Delta(X) = \Delta_\Pi(X,\varphi(X)) = \Delta_\Pi(X,\varphi(Y) \oplusb \varphi(Z)) + \Delta_\Pi(\varphi(Y) \oplusb \varphi(Z), \varphi(X)) = \Delta(Y) + \Delta(Z) + \Delta(\varphi(Y) \oplusb \varphi(Z))$.
    \end{claimproof}
    This completes the proof.
\end{proof}

\subsection{The dynamic kernelization data structure}\label{sec:kernelization_data_structure}

In this section, we finally describe our dynamic kernelization data structure by combining our data structure from \Cref{thm:technical_main} with our protrusion replacement automaton from \Cref{lem:protrusion_replacement_automaton}. First, we describe how the protrusion replacement automaton can be used to obtain a kernel in the static case. Then, we show that we can maintain this kernel in the dynamic case by using our dynamic protrusion decomposition data structure from \Cref{thm:technical_main}.

Let $\mc G$ be a $\cmso$-definable graph class and $\Pi$ a parameterized graph problem.
Let $G \in \mc G$ be a graph together with a normal $(p,t)$-protrusion decomposition $\Tc = (T,\bag)$ with root $r$. Let $\chd(r) = \{c_1,\dots,c_q\}$ be the set of children of $r$.
For $1 \leq i \leq q$, let $G_i = (G_{c_i},\bag(c_i),\Lambda_{\bag(c_i)}^*)$ be the corresponding boundaried graph. We observe that $G = G[\bag(r)] \cup \bigcup_{i = 1}^q G_{c_i}$. By the definition of protrusion decompositions, we have $V(G_i) \cap V(G_j) \subseteq \bag(r)$ for distinct $i$ and $j$. Moreover, due to the $\edges$ function, there is no edge $uv \in E(G_i)$ with $u,v \in \bag(r)$ (for $1 \leq i \leq q$). Thus, every edge of $G$ is attributed to exactly one $G_i$ or to $G[\bag(r)]$ and when replacing one boundaried subgraph $G_i$, the other boundaried subgraphs $G_j$ for $j \neq i$ are unaffected. We further remark that since $\Tc$ is normal, we have $\bag(r) \subseteq \bigcup_{i = 1}^q \bag(c_i)$.

We define the graph $K = K(G,\Tc)$ as the graph obtained by replacing each $G_i$ with its progressive representative $\varphi_{\Pi,\mc G,[t+1]}(G_i)$ for every $1 \leq i \leq q$.
That is, we replace the boundaried subgraphs $G_i$ one by one, producing a sequence $G = G^{(0)}, G^{(1)},\dots,G^{(q)} = K$ of graphs, where $G^{(i)}$ is the graph obtained from $G^{(i-1)}$ by replacing $G_i$ with $\varphi_{\Pi,\mc G,[t+1]}(G_i)$.
In particular, if we assume $G^{(i-1)} = H_{i-1} \oplusu G_i$ for some boundaried graph $H_{i-1}$, then $G^{(i)} = H_{i-1} \oplusu \varphi_{\Pi,\mc G,[t+1]}(G_i)$. As mentioned before, this replacement does not affect the other subgraphs $G_j$, so in particular, we have $G_{c_j} \subseteq G^{(i)}$ for $0 \leq i < j \leq q$, and this sequence of replacements is well-defined. We further remark that the graph $K$ is independent of the order $c_1,\dots,c_q$ of the root-children. We additionally define $\Delta = \Delta(G,\Tc) = \sum_{i = 1}^q \Delta_\Pi(G_i,\varphi_{\Pi,\mc G,[t+1]}(G_i))$.   Now, we first show that $(K,\Delta)$ is indeed a kernel, $K \in \mc G$, $\Delta \leq 0$, and $|V(K)| \leq \OO(p)$. Afterwards, we show how we can maintain $K(G,\Tc)$ and $\Delta(G,\Tc)$ efficiently for a dynamic graph $G$ with the dynamic protrusion decomposition $\Tc$ from \Cref{thm:technical_main}.

\begin{lemma}\label{lem:kernel}
    Let $\mc G$ be a $\cmso$-definable graph class and $\Pi$ be a parameterized graph problem. Let $G \in \mc G$ be a graph together with a normal $(p,t)$-protrusion decomposition $\Tc = (T,\bag)$. Let $K = K(G,\Tc)$ and $\Delta = \Delta(G,\Tc)$ be as defined above. Then, $(K,\Delta)$ is a kernel for $\Pi$ on $G$, $K \in \mc G$, $\Delta \leq 0$, and $|V(K)| \leq \OO(p)$. The hidden factors in the $\OO$-notation depend only on $\Pi$, $\mc G$, and $t$.
\end{lemma}
\begin{proof}
    We start by showing that $|V(K)| \leq \OO(p)$ and $\Delta \leq 0$.
    For this, let $r$ be the root of the $(p,t)$-protrusion decomposition.
    It holds that $|\bag(r)| \leq p$ and $r$ has at most $q \leq p$ children. In $K$, each of these children corresponds to a boundaried graph $\varphi_{\Pi,\mc G,[t+1]}(G_i)$ for $1 \leq i \leq q$ on at most $c_\Pi(t+1)$ vertices. Thus, we have $|V(K)| \leq p + q \cdot c_\Pi(t+1) \le \OO(p)$, where the hidden factors in the $\OO$-notation depend only on $\Pi$, $\mc G$, and $t$.
    Moreover, by the definition of progressive representatives, we have $\Delta_\Pi(G_i,\varphi_{\Pi,\mc G,[t+1]}(G_i)) \leq 0$ for every $1 \leq i \leq q$, and thus $\Delta \leq 0$.

    We now continue with showing that $K \in \mc G$.
    
    \begin{claim}
        $K \in \mc G$.
    \end{claim}
    \begin{claimproof}
        By the definition of $\equiv_{\mc G}$, and thus $\varphi_{\Pi,\mc G,[t+1]}$, we have for every boundaried graph $F \in \mc F$ and every $1 \leq i \leq q$, $F \oplusu G_i \in \mc G$ if and only if $F \oplusu \varphi_{\Pi,\mc G,[t+1]}(G_i) \in \mc G$. In particular, we have $G^{(i-1)} = H_{i-1} \oplusu G_i \in \mc G$ if and only if $H_{i-1} \oplusu \varphi_{\Pi,\mc G,[t+1]}(G_i) = G^{(i)} \in \mc G$ for every $1 \leq i \leq q$. Thus, we have
        \[G = G^{(0)} \in \mc G \Leftrightarrow G^{(1)} \in \mc G \Leftrightarrow G^{(2)} \in \mc G \Leftrightarrow \dots \Leftrightarrow G^{(q)} = K \in \mc G.\]
        Since by assumption $G \in \mc G$, it follows that $K \in \mc G$.
    \end{claimproof}

Finally, we show that $(K,\Delta)$ is a kernel for $\Pi$ on $G$.

    \begin{claim}
        For every $k \in \mathbb{Z}$, we have $(G,k) \in \Pi$ if and only if $(K,k+\Delta) \in \Pi$.
    \end{claim}
    \begin{claimproof}
        By definition of $\varphi_{\Pi,\mc G,[t+1]}$ and $\Delta_\Pi$, we have for every boundaried graph $F \in \mc F$, every integer $k\in \mathbb{Z}$, and every $1 \leq i \leq q$, $(F \oplusu G_i,k) \in \Pi$ if and only if $(F \oplusu \varphi_{\Pi,\mc G,[t+1]}(G_i), k + \Delta_\Pi(G_i,\varphi_{\Pi,\mc G,[t+1]}(G_i)) \in \Pi$.
        Thus, it follows that for every integer $k$ and every $1 \leq i \leq q$, $(G^{(i-1)},k) = (H_{i-1} \oplusu G_i,k) \in \Pi$ if and only if $(H_{i-1} \oplusu \varphi_{\Pi,\mc G,[t+1]}(G_i), k + \Delta_\Pi(G_i,\varphi_{\Pi,\mc G,[t+1]}(G_i)) = (G^{(i)},k+\Delta_\Pi(G_i,\varphi_{\Pi,\mc G,[t+1]})) \in \Pi$.
    In total, we have for every integer $k$,
    \begin{align*}
        (G,k) \in \Pi
        \Leftrightarrow& (G^{(0)},k) \in \Pi 
        \Leftrightarrow (G^{(1)},k + \Delta_{\Pi}(G_1,\varphi_{\Pi,\mc G,[t+1]}(G_1)) \in \Pi\\ \Leftrightarrow& (G^{(2)},k + \Delta_{\Pi}(G_1,\varphi_{\Pi,\mc G,[t+1]}(G_1) + \Delta_{\Pi}(G_2,\varphi_{\Pi,\mc G,[t+1]}(G_2)) \in \Pi\\
        \Leftrightarrow& \dots
        \Leftrightarrow (G^{(q)}, k + \sum\limits_{i = 1}^q \Delta_\Pi(G_i,\varphi_{\Pi,\mc G,[t+1]}(G_i)) \in \Pi
        \Leftrightarrow (K,\Delta) \in \Pi.
    \end{align*}
\end{claimproof}
This completes the proof.
\end{proof}

Finally, we are ready to construct the kernelization data structure by combining our data structure from \Cref{thm:technical_main} with our protrusion replacement automaton from \Cref{lem:protrusion_replacement_automaton} to maintain the kernel $(K,\Delta)$ for $\Pi$ and $\mc G$ on a dynamic graph $G$.
We first give a kernelization algorithm with the size of the smallest treewidth-$\cTwMod$-modulator as parameter. Then, we use the treewidth-boundedness to lift the parameter to the solution size.

\begin{lemma}\label{lem:kernelization_tw_mod}
    Let $\mc G$ be a $\cmso$-definable graph class that excludes a topological minor, $\Pi$ a parameterized graph problem that has FII, and $\cTwMod$ an integer.

    There exists a data structure that, for a dynamic graph $G \in \mc G$, maintains a kernel $(K,\Delta)$ for $\Pi$ on $G$ so that $K \in \mc G$, $\Delta \le 0$, and $|K| \le \OO(\optTwMod(G))$.
     The data structure supports the following operations: 
    \begin{itemize}
        \item $\init(G)$: Initialize the data structure with a graph $G \in \mc G$. Outputs $(K,\Delta)$. Runs in $\OO(|G| \log |G|)$ amortized time.
        \item $\addVertex(v)$: Given a new vertex $v \notin V(G)$, add $v$ to $G$.
        \item $\deleteVertex(v)$: Given an isolated vertex $v \in V(G)$, remove $v$ from $G$. 
        \item $\addEdge(e)$: Given a new edge $e \in \binom{V(G)}{2} \setminus E(G)$, add $e$ to $G$.
        \item $\deleteEdge(e)$: Given an edge $e \in E(G)$, remove $e$ from $G$.
    \end{itemize}
    Each update takes $\OO(\log |G|)$ amortized time.
    For each operation, the changes to $K$ can be described as a sequence of (hyper)graph operations of size $\OO(1)$, which is returned.
    The change to $\Delta$ is also returned (and can be arbitrary).
    The data structure works under the promise that $G \in \mc G$ at all times.
\end{lemma}
\begin{proof}
    Let $H$ be a graph so that all graphs in $\mc G$ are $H$-topological-minor-free.
    We use the data structure of \Cref{thm:technical_main} initialized with the parameters $H$ and $\cTwMod$.
    It maintains an annotated normal $(\OO_{H,\cTwMod}(\optTwMod(G),\OO_{H,\cTwMod}(1))$-protrusion decomposition $\Tc = (T,\bag,\edges)$ of $G$.
    Let $\cFinalProt = \OO_{H,\cTwMod}(1)$ be a constant that depends on $H$ and $\cTwMod$ so that the second parameter of the protrusion decomposition $\Tc$, i.e., the width of the protrusions, is upper bounded by $\cFinalProt$.
    In addition, if $r$ is the root of $\Tc$, the hypergraph $\torso(r)$ is maintained, and after every update operation to $G$, the changes to $\torso(r)$ can be described as a sequence of basic hypergraph operations of size $\OO_{H,\cTwMod}(1)$, which is returned. 
    Moreover, the set of all edges that have been removed from $\edges(r)$ or inserted into $\edges(r)$ is also returned and has size $\OO_{H,\cTwMod}(1)$, as well.

    Our goal is to maintain the kernel $(K,\Delta)$, where $K = K(G,\Tc) \in \mc G$ and $\Delta = \Delta(G,\Tc)$.
    By \Cref{lem:kernel} it is indeed a kernel, has $\Delta \le 0$, $K \in \mc G$, and $|V(K)| \leq \OO(\optTwMod(G))$.
    Since $K \in \mc G$, $K$ excludes $H$ as a topological-minor and thus is sparse (see \Cref{lem:topminorfreesparse}), so also $|K| \le \OO(\optTwMod(G))$.
    
    We now describe how we maintain $K(G,\Tc)$ and $\Delta(G,\Tc)$.
    By \Cref{lem:protrusion_replacement_automaton}, there exists a tree decomposition automaton $\autom$ of width $\cFinalProt$ that, given a node $x \in V(T) \setminus \{r\}$ determines the pair $(\varphi_{\Pi,\mc G,[\cFinalProt+1]}(X),\Delta(X,\varphi_{\Pi,\mc G,[\cFinalProt+1]}(X))$, where $X = (G_x,\bag(x),\Lambda_{\bag(x)}^*)$ is the boundaried graph that corresponds to $x$. By \Cref{thm:technical_main}, we can maintain the run $\run_{\autom}^{\Tc}(x)$ for every non-root node $x \in V(T)$ at the cost of an additional $\tau(\cFinalProt)$ factor on the running times, where $\tau(\cFinalProt)$ is the evaluation time of $\autom$ that depends only on $\Pi$, $\mc G$ and $\cFinalProt$. For the sake of convenience, we drop the subscripts in the following and set $\varphi(X) \coloneq \varphi_{\Pi,\mc G,[\cFinalProt+1]}(X)$ and $\Delta(X) \coloneq \Delta(X,\varphi_{\Pi,\mc G,[\cFinalProt+1]}(X))$ for every boundaried graph $X$.

    Then, in particular, for every $c_i \in \chd(r)$, we can in constant time determine the progressive representative $\varphi(G_i)$ together with the transposition constant $\Delta(G_i)$ by querying the run of the automaton $\autom$. 
    Then, whenever there is a ``change'' in the run $\run_{\autom}^{\Tc}(c_i)$ of $\autom$ on a root-child $c_i$, we update the replacement $\varphi(G_i)$ in $K$.
    To keep track of which vertices and edges in $K$ belong to which (boundaried) subgraph $G_i$, we additionally maintain a balanced binary search tree $\bbst$ that contains every root-child $c_i \in \chd(r)$ together with a list of all edges in $E(\varphi(G_i))$ and all vertices in $\inter(\varphi(G_i))$.

    We now describe in detail the operations of our data structure and how the kernel $(K,\Delta)$ is maintained.
    First, let us note that we can implement the initialization operation by first initializing the data structure of \Cref{thm:technical_main} with an empty graph, and then inserting edges and vertices one at the time, because if $G$ is $H$-topological-minor-free, then also every subgraph of $G$ is $H$-topological-minor-free.
    We do not need to worry about the fact that intermittently the graph may be outside of $\mc G$, because we output the kernel only at the end of all of these operations, when $G \in \mc G$ is guaranteed.

    For the other operations (vertex/edge insertion/deletion), we update the protrusion decomposition $\Tc$ with the respective operation from \Cref{thm:technical_main} in $\OO_{H,\cTwMod}(\log |G|)$ time. Then, we remove every edge that has been removed from $\edges(r)$ from $K$. Given the set of size $\OO_{H,\cTwMod}(1)$ of all such edges, this can be done in $\OO_{H,\cTwMod}(1)$ time via $\OO_{H,\cTwMod}(1)$ basic graph operations to $K$. Then, we consider the root-children $c_i \in \chd(r)$. By \Cref{thm:technical_main}, every change of the run $\run_{\autom}^{\Tc}(c_i)$ at a root-child $c_i$ corresponds to a change in $\torso(r)$. Thus, we do not need to reconsider the root-children $c_i \in \chd(r)$, where the corresponding hyperedge $e_{c_i}$ is in $\torso(r)$ before and after the update.

    Let $G$ and $G'$ denote the graph, $\Tc = (T,\bag,\edges)$ and $\Tc' = (T',\bag',\edges')$ the protrusion decomposition and $R$ and $R'$ the hypergraph $\torso(r)$ before and after the update, respectively.
    Then, by \Cref{thm:technical_main}, the changes from $R$ to $R'$ can be described as a sequence of operations $\mc C$ of size $\OO_{H,\cTwMod}(1)$ which is returned by the data structure.
    Suppose, this sequence $\mc C$ produces the following sequence of hypergraphs: $R = R^{(0)}, R^{(1)}, \dots, R^{(\zeta)} = R'$, where $\zeta = |\mc C| \le \OO_{H,\cTwMod}(1)$. We update the kernel $(K,\Delta)$ step-by-step along $\mc C$ using $\OO(1)$ basic graph operations for each operation in $\mc C$. In the end, these sequences of graph operations are concatenated to obtain the sequence of operations between the graph $K$ before and after the update, which is then returned. In the following, we describe in detail how the kernel is updated for the different types of basic hypergraph operations in $\mc C$.

    First, for every $\addVertex$ operation in $\mc C$, we also add the new vertex $v$ to $K$. Clearly, this requires only one basic graph operation, namely $\addVertex(v)$, and can be done in $\OO(1)$ time.

    We then consider the $\addHyperedge$ operations in $\mc C$. Due to the minimality of $\mc C$, if a hyperedge $e$ is added during $\mc C$, then $e$ cannot be deleted later within $\mc C$, so $e \in E(R')$. Thus, the inserted hyperedge $e$ corresponds to a root-child $c_i \in V(T') \setminus V(T)$. To update the kernel $(K,\Delta)$ we glue the graph $\varphi(G_i')$ to $K$ and the integer $\Delta(G_i')$ to $\Delta$. Both $\varphi(G_i')$ and $\Delta(G_i')$ can be retrieved from the run of our automaton $\autom$ in constant time. To add $\varphi(G_i')$ to $K$, we iterate over $\varphi(G_i')$ and apply the $\addVertex(v)$ and $\addEdge(e)$ operation for each $v \in \inter(\varphi(G_i'))$ and each edge $e \in E(\varphi(G_i'))$ to $K$. We also insert the same vertices and edges into a list and add this list to $\bbst$ with key $c_i$. Since the size of $\varphi(G_i')$ depends only on $\mc G$, $\Pi$, and $\cFinalProt = \cFinalProt(H,\cTwMod)$, this can be done in time $\OO(\log |G|)$ and the created sequence of operations to $K$ is in $\OO(1)$, where the hidden factors in the $\OO$-notation depend only on $\mc G$, $\Pi$, and $\cTwMod$.

    Next, we consider the $\deleteHyperedge$ operations in $\mc C$. Again, due to the minimality of $\mc C$, if a hyperedge $e$ is deleted during $\mc C$, then $e$ was not inserted before during the same sequence $\mc C$, so $e \in E(R)$. Thus, the hyperedge $e$ corresponds to a root-child $c_i \in V(T) \setminus V(T')$. To update $(K,\Delta)$, we remove $\varphi(G_i)$ from $K$ and subtract $\Delta(G_i)$ from $\Delta$. To do this, we need the balanced binary search tree $\bbst$. We first look up the entry for $c_i$ in $\bbst$ and obtain a list of all edges in $\varphi(G_i)$ and all vertices in $\inter(\varphi(G_i))$. We remove every such vertex and edge from $K$ by using the basic graph operations $\deleteEdge(e)$ and $\deleteVertex(v)$. Finally, we delete the entry with key $c_i$ from $\bbst$.
    Again, the size of $\varphi(G_i)$ depends only on $\mc G$, $\Pi$, and $\cFinalProt$. So, this can be done in time $\OO(\log |G|)$ and the created sequence of operations to $K$ is in $\OO(1)$, where the hidden factors in the $\OO$-notation depend only on $\mc G$, $\Pi$, and $\cTwMod$.

    Then, for every $\deleteVertex$ operation in $\mc C$, we delete the same vertex $v$ from $K$, again using only one basic graph operation, namely $\deleteVertex(v)$, which takes $\OO(1)$ time. To do this, we need to argue that $v$ is an isolated vertex in the current graph $K$. For this, note that due to the minimality of $\mc C$, every edge in the current $K$ belongs to at least one $\varphi(G_i)$ or $\varphi(G'_i)$, where $c_i \in V(T) \cup V(T')$. Furthermore, each such $\varphi(G_i)$ or $\varphi(G'_i)$ corresponds to a hyperedge $e_t \in E(R) \cup E(R')$. Now, suppose that the current $\deleteVertex$ operation is the $j^{\text{th}}$ operation in $\mc C$, i.e., it transforms $R^{(j-1)}$ into $R^{(j)}$. Then, $v$ is an isolated vertex in $R^{(j-1)}$, i.e., there is no hyperedge $e_t \in E(R^{(j-1)})$ with $v \in V(e_t)$, so every such hyperedge has been deleted in one of the previous steps together with the corresponding subgraph $\varphi(G_i)$ from $K$. Thus, $v$ is isolated in $K$ and can be removed via the $\deleteVertex$ basic graph operation.

    Finally, we need to add all the edges that have been added to $\edges(r)$ to $K$. This can again be done in $\OO_{H,\cTwMod}(1)$ time via $\OO_{H,\cTwMod}(1)$ basic graph operations given the list of all such edges.
\end{proof}

For linearly treewidth-bounding problems, we observe that we can instead parameterize by $\opt_{\Pi}(G)$.

\begin{lemma}
\label{lem:kernelization_tw_mod:v2}
Consider the setting of \Cref{lem:kernelization_tw_mod}.
If $\Pi$ is linearly treewidth-bounding on $\mc G$, then the same lemma holds, but without taking $\cTwMod$ as a parameter and instead guaranteeing that $|K| \le \OO(\opt_{\Pi}(G))$.
\end{lemma}
\begin{proof}
Because $\Pi$ is linearly treewidth-bounding, there exists a constant $\cTwMod$ so that $\optTwMod(G) \le \OO(\opt_{\Pi}(G))$.
Therefore, by choosing the $\cTwMod$ in \Cref{lem:kernelization_tw_mod} to be this constant, we get that $|K| \le \OO(\opt_{\Pi}(G))$.
\end{proof}

We then observe that \Cref{lem:kernelization_tw_mod,lem:kernelization_tw_mod:v2} form a (more precise) version of \Cref{thm:kernelization} (note that $\Delta$ should be negated), and thus this completes the proof of \Cref{thm:kernelization}.

%% file: conclusion.tex
\section{Conclusions}
\label{sec:concl}
We gave a dynamic algorithm for maintaining an approximately optimal protrusion decomposition of a dynamic topological-minor-free graph, and applied it for dynamic kernelization.
Let us discuss the applications of our results, future work, and open problems.

\paragraph{Direct applications.}
We list some concrete problems to which our dynamic kernelization meta-theorem \Cref{thm:kernelization} applies directly.

First, let us consider the problems where \Cref{thm:kernelization} applies for any topological-minor-free graph class $\mc G$.
In this setting, most of the natural problems that are known to satisfy the required conditions are vertex-deletion problems to graph classes $\mc Q$, for which it is known that graphs in $\mc G \cap \mc Q$ have bounded treewidth.
Such problems are trivially linearly treewidth-bounding, so \Cref{thm:kernelization} applies whenever they have FII.
Examples with FII include \textsc{Vertex Cover}, \textsc{Feedback Vertex Set}, more generally \textsc{Treewidth-$\cTwMod$-Deletion} for any $\cTwMod$, even more generally \textsc{$\mc F$-Minor-Free-Deletion} for any finite set of connected graphs $\mc F$ that contains a planar graph, \textsc{Cluster Vertex Deletion}, \textsc{Chordal Vertex Deletion}, \textsc{Interval Vertex Deletion}, \textsc{Proper Interval Vertex Deletion}, \textsc{Cograph Vertex Deletion}, \textsc{Connected Vertex Cover}, \textsc{Connected Cograph Vertex Deletion}, and \textsc{Connected Cluster Vertex Deletion}~\cite{Bodlaender_Fomin_Lokshtanov_Penninkx_Saurabh_Thilikos_2016,Kim_Langer_Paul_Reidl_Rossmanith_Sau_Sikdar_2015}.
Another problem which is not strictly speaking a vertex-deletion problem, but for which \Cref{thm:kernelization} applies in this setting is \textsc{Edge Dominating Set}~\cite{Kim_Langer_Paul_Reidl_Rossmanith_Sau_Sikdar_2015}.

Let us then consider problems where \Cref{thm:kernelization} applies for any minor-free graph class.
This of course includes everything mentioned above, but additionally includes ``minor-bidimensional'' problems that are linearly treewidth-bounding by non-trivial arguments~\cite{Fomin_Lokshtanov_Saurabh_2018,Fomin_Lokshtanov_Saurabh_Thilikos_2020}.
Such problems include the problem \textsc{$\mc F$-Minor-Packing} for any finite set of connected graphs $\mc F$ that contains a planar graph, which asks to pack a maximum number of vertex-disjoint minor-models of graphs from $\mc F$.
Concretely, this contains for example the problem \textsc{Cycle Packing}.

We then consider applications of \Cref{thm:kernelization} to apex-minor-free graph classes, which include planar graphs and graphs of bounded genus.
We cover all the problems above, but also ``contraction-bidimensional'' problems~\cite{Fomin_Lokshtanov_Saurabh_2018,Fomin_Lokshtanov_Saurabh_Thilikos_2020}.
Such problems include \textsc{Dominating Set}, \textsc{Induced Matching}, \textsc{$r$-Dominating Set} for any $r$, \textsc{Connected Dominating Set}, and \textsc{$r$-Scattered Set} for any $r$~\cite{Fomin_Lokshtanov_Saurabh_Thilikos_2020}.

\paragraph{Potential indirect and future applications.}
In \Cref{thm:kernelization}, we applied the data structure of \Cref{theo:main} to obtain dynamic linear kernels for problems that directly satisfy certain conditions.
However, protrusion decompositions have been applied in kernelization even more broadly.
As one example, Bodlaender et al.~\cite{Bodlaender_Fomin_Lokshtanov_Penninkx_Saurabh_Thilikos_2016} obtain polynomial (but superlinear) kernels on graphs of bounded genus for problems that do not have FII but are instead expressible in $\cmso$ in a certain way, for example, \textsc{Independent Dominating Set}.
We conjecture that our data structure could be modified to obtain dynamic versions of these results, but with additional $\OO(\poly(\opt(G)))$ factors in the update time and the number of changes to the kernel.

Another problem for which \Cref{thm:kernelization} does not apply but for which~\cite{Bodlaender_Fomin_Lokshtanov_Penninkx_Saurabh_Thilikos_2016} obtain a linear kernel on graphs of bounded genus is \textsc{Triangle Packing}, or more generally, the problem of packing connected subgraphs from any fixed family of subgraphs.
It has FII, but in order to relate it to treewidth, one must first apply a preprocessing routine that deletes all vertices not contained in any triangle.
It would be interesting to see if this preprocessing routine could be implemented dynamically, so that \Cref{thm:kernelization} could be applied also to this problem.

Besides kernelization, protrusion decompositions have been applied to design parameterized and approximation algorithms, for example, in~\cite{Fomin_Lokshtanov_Misra_Saurabh_2012} (see \Cref{sec:intro} for more references).
Such algorithms typically employ complex techniques in addition to protrusion replacement, and therefore it is not clear at all whether the data structure of \Cref{theo:main} (or its potential extensions) could help to design dynamic versions of them.
As a benchmark problem in this direction, we ask whether \textsc{Feedback Vertex Set} (on general graphs) admits a dynamic constant-factor approximation algorithm, for example, with sublinear (amortized) update time.

\paragraph{Potential improvements and future work.}
As already discussed in the introduction, the lower bound of~\cite{DBLP:journals/siamcomp/PatrascuD06} implies that the logarithmic update time of the data structure of \Cref{theo:main} is tight, and furthermore it is tight even for some cases of \Cref{thm:kernelization}, such as \textsc{Cycle Packing} on planar graphs.
It would be interesting to know if the logarithmic update time for dynamic kernelization is tight even for problems with more local flavor, such as \textsc{Dominating Set} on planar graphs.

While the logarithmic update time seems natural for problems associated with treewidth-modulators, perhaps it could be improved to constant update time for problems associated with treedepth-modulators, since graphs of bounded treedepth admit dynamic algorithms with constant update time~\cite{Dvořák_Kupec_Tůma_2014,DBLP:conf/soda/ChenCDFHNPPSWZ21}.

The fact that the update time of \Cref{theo:main} is amortized instead of worst-case comes from the techniques of~\cite{Korhonen_2025}, and we believe that if the result of~\cite{Korhonen_2025} would be improved to worst-case instead of amortized, then \Cref{theo:main} most likely also could be improved.

The data structure of \Cref{theo:main} is restricted to topological-minor-free graph classes.
This is a somewhat natural barrier because of the following \Cref{thm:topoltight}, whose proof is presented in \Cref{subsec:missingproofsconcl}.
A graph class $\mc G$ is \emph{subgraph-closed} if for all $G \in \mc G$, every subgraph of $G$ is also contained in $\mc G$.

\begin{restatable}{proposition}{topoltight}
\label{thm:topoltight}
Let $\mc G$ be a graph class that is subgraph-closed.
If $\mc G$ does not exclude a topological minor, then there is no function $f(\cTwMod)$ and a constant $c < 2$ so that for every graph $G \in \mc G$ and every $\cTwMod$, $G$ admits an $(f(\cTwMod) \cdot \optTwMod(G)^c, f(\cTwMod))$-protrusion decomposition.
\end{restatable}

\Cref{thm:topoltight} tells that topological-minor-free graph classes are the most general subgraph-closed classes that admit a linear relation between treewidth-modulators and protrusion decompositions.
However, one could still imagine a version of \Cref{theo:main} that directly maintains an approximately-optimal protrusion decomposition, regardless of the size of the minimum treewidth-modulator.
Such a version of \Cref{theo:main} could be useful for dynamic versions of the applications of protrusion decompositions beyond topological-minor-free graph classes, e.g.~\cite{Fomin_Lokshtanov_Misra_Saurabh_2012}.

Currently, the dynamic kernelization algorithms of \Cref{thm:kernelization} provide no method of lifting a solution in the kernel $K$ into a solution in the original graph.
It would be interesting if such a method could be implemented, for example, lifting a dominating set of $K$ of size $\opt(K)$ into a dominating set of $G$ of size $\opt(G)$ in time $\OO(\opt(G) \cdot \log n)$.

%% file: appendix.tex
\section{Missing proofs}
\label{sec:missingproofs}

\subsection{Proofs missing from \Cref{sec:preliminaries}}

We prove \Cref{lem:well_linked_torso} using the following lemma from~\cite{Korhonen_2025}. We denote by $e_\perp$ the hyperedge with $V(e_\perp) = \emptyset$. For a hypergraph $G$ with a hyperedge $e \in E(G)$, an $e$-rooted superbranch decomposition of $G$, where the root is the leaf $\Lc^{-1}(e)$.

\begin{lemma}[{\cite[Lemma 4.3]{Korhonen_2025}}]\label{lem:well_linked_trans}
    Let $G$ be a hypergraph, $e_\perp \in E(G)$, and $\Tc = (T,\Lc)$ an $e_\perp$-rooted superbranch decomposition of $G$. Let also $t \in V_\inter(T)$ be a node with parent $p$, so that $\Lc[c]$ is well-linked for every child $c$ of $t$. Let $e_p \in E(\torso(t))$ be the hyperedge of $\torso(t)$ that corresponds to $p$. Then, a set $A \subseteq E(\torso(t)) \setminus \{e_p\}$ is well-linked in $\torso(t)$ if and only if $A \expand \Tc$ is well-linked in $G$.
\end{lemma}

Now, \Cref{lem:well_linked_torso} is an easy generalization of \Cref{lem:well_linked_trans}.

\wellLinkedTrans*
\begin{proof}
    Let $G' = G \cup \{e_\perp\}$ and let $\Tc' = (T',\Lc')$ be the rooted superbranch decomposition of $G'$ obtained from $\Tc$ by adding a node $r'$ to $T$ as a leaf-child of $r$, and setting $\Lc'(r') = e_\perp$. We first note that a set $B \subseteq E(G)$ is well-linked in $G$ if and only if it is well-linked in $G'$. For a node $t \in V_\inter(T')$, we denote by $\torso'(t)$ the torso of $t$ in $\Tc'$. We then further observe that for every node $t \in V_\inter(T') \setminus \{r\} = \Vint(T) \setminus \{r\}$, $\torso(t) = \torso'(t)$, while $\torso'(r) = \torso(r) \cup e_\perp$. Therefore, for every node $t \in \Vint(T)$, a set $A \subseteq E(\torso(t))$ is well-linked in $\torso(t)$ if and only if it is well-linked in $\torso'(t)$. Applying \Cref{lem:well_linked_trans}, the lemma follows.
\end{proof}

\subsection{Proofs missing from \Cref{sec:kernelization}}

\progRepr*
\begin{proof}
    Let $\mc{C}$ be an equivalence class of $\equiv_{\Pi,\mc G}$.
    Without loss of generality (by flipping the $\in \Pi$ relation if needed), we assume that $(G,k) \notin \Pi$ for every graph $G$ and $k < 0$.
    If $\mc C$ is a class of monotone graphs, i.e., for every boundaried graph $G \in \mc C$, every $F \in \mc F$, and every $k \in \mathbb{Z}$, $(F \oplusu G,k) \notin \Pi$, then $\Delta_\Pi(G_1,G_2) = 0$ for any two graphs $G_1,G_2 \in \mc C$, and thus every graph in $\mc C$ is a progressive representative.
    
    Otherwise, there exist $G_0 \in \mc C$, $F_0 \in \mc F$, and a (non-negative) integer $k_0$ such that $(F_0 \oplusu G_0,k_0) \in \Pi$. Among all those triples, we choose $(G_0,F_0,k_0)$ such that $k_0$ is minimized. We now show that $G_0$ is a progressive representative of $\mc C$. For this, let $G \in \mc C$ and suppose that $\Delta_\Pi(G_0,G) < 0$. For every $(F,k) \in \mc F \times \mathbb{Z}$, we have $(F \oplusu G_0,k) \in \Pi$ if and only if $(F \oplusu G,k + \Delta_\Pi(G_0,G)) \in \Pi$. Thus, since $(F_0 \oplusu G_0,k_0) \in \Pi$, we have $(F_0 \oplusu G,k_0 + \Delta(G_0,G)) \in \Pi$, which is a contradiction to the minimality of $(G_0,F_0,k_0)$. It follows that $\Delta_\Pi(G_0,G)\geq 0$, and since $G$ was chosen arbitrarily from $\mc C$, $G_0$ is a progressive representative of $\mc C$.
\end{proof}

\subsection{Proofs missing from \Cref{sec:concl}}
\label{subsec:missingproofsconcl}
We re-state \Cref{thm:topoltight} and prove it.

\topoltight*
\begin{proof}
Let $\mc G$ be a graph that is subgraph-closed and does not exclude a topological minor.
For integers $r,t \ge 1$, a graph is an \emph{$(\ge r)$-subdivided $t$-clique} if it is obtained from the $t$-clique by replacing each edge by a path of $\ge r$ edges.
We call the $t$ ``non-subdivision'' vertices of such a graph the \emph{junction vertices}.
We observe that because graphs in $\mc G$ contain all graphs as topological minors, they particularly contain an $(\ge 2)$-subdivided $t$-clique as a topological minor for every $t \ge 1$, which by subgraph-closedness implies that $\mc G$ contains an $(\ge 2)$-subdivided $t$-clique for every $t \ge 1$.

For the sake of contradiction, suppose that there is a function $f(\cTwMod)$ and a constant $c < 2$, so that for every graph $G \in \mc G$ and every $\cTwMod$, $G$ admits an $(f(\cTwMod) \cdot \optTwMod(G)^c, f(\cTwMod))$-protrusion decomposition.
Let $G \in \mc G$ be a $(\ge 2)$-subdivided $t$-clique for a large enough $t$ that we will choose later.
By the assumption, it admits an $(f(1) \cdot t^c, f(1))$-protrusion-decomposition $(T,\bag)$, rooted at a node $r \in V(T)$.
For a child $c$ of $r$, denote by $\Tc_c = (T_c, \bag_c)$ the tree decomposition rooted at $c$, and by $G_c$ the subgraph of $G$ induced by the bags of $\Tc_c$, and note that $\tw(G_c) \le f(1)-1$.

A \emph{subdivision path} of $G$ is a maximal path in $G$ containing no junction vertices.
In particular, $G$ has exactly $\binom{t}{2}$ subdivision paths, which are pairwise disjoint.

\begin{claim}
For each $c$, $G_c$ contains vertices from at most $\binom{4 \cdot f(1)}{2} + 2 \cdot f(1)$ subdivision paths.
\end{claim}
\begin{claimproof}
Let $P$ be a subdivision path of $G$.
We say that $G_c$ \emph{partially contains} $P$ if some vertices of $P$ are in $G_c$ and some are not, and \emph{fully contains} $P$ if all vertices of $P$ are in $G_c$.
We observe that if $G_c$ partially contains $P$, then at least one vertex of $P$ must be in $\bag(c)$, and therefore $G_c$ can partially contain at most $f(1)$ subdivision paths.

Suppose that $G_c$ fully contains more than $\binom{4 \cdot f(1)}{2} + f(1)$ subdivision paths.
If follows that $G_c \setminus \bag(c)$ fully contains more than $\binom{4 \cdot f(1)}{2}$ subdivision paths, which implies that $G_c$ contains more than $4 \cdot f(1)$ junction vertices, which implies that $G_c \setminus \bag(c)$ contains more than $3 \cdot f(1)$ junction vertices.
Denote these junction vertices by $J$.

Each subdivision path between a pair of junction vertices in $J$ must be partially or fully contained in $G_c$, so there is a subset $J' \subseteq J$ of size $|J'| > f(1)$ so that all subdivision paths between vertices in $J'$ are fully contained in $G_c$.
However, then $G_c$ contains an $(f(1)+1)$-clique as a (topological) minor, implying that $\tw(G_c) \ge f(1)$, which is a contradiction.
\end{claimproof}

Now, the graphs $G_c$ over all $c$ can contain vertices from at most $f(1) \cdot t^c \cdot \left(\binom{4 \cdot f(1)}{2} + 2 \cdot f(1)\right)$ subdivision paths, and the root-bag $\bag(r)$ can contain vertices from at most $f(1) \cdot t^c$ subdivision paths.
By choosing $t$ large enough depending on $f(1)$ and $c$, we get that \[\binom{t}{2} > f(1) \cdot t^c \cdot \left(\binom{4 \cdot f(1)}{2} + 2 \cdot f(1) + 1\right),\] so there would be a subdivision path that appears nowhere in the protrusion decomposition $(T,\bag)$, which is a contradiction.
\end{proof}

\section{Tree decomposition automata}
\label{sec:automata}

In this section, following~\cite{Korhonen_Majewski_Nadara_Pilipczuk_Sokołowski_2023,Korhonen_2025}, we formally define tree decomposition automata. We assume that the vertices of all graphs that we process come from a countable, totally ordered universe $\Omega$ (e.g., $\Omega = \N$)

\begin{definition}[\cite{Korhonen_Majewski_Nadara_Pilipczuk_Sokołowski_2023}]\label{def:automaton}
    A \emph{(deterministic) tree decomposition automaton of width $\ell$} is a tuple $\autom = (Q,F,\iota,\delta)$, where
    \begin{itemize}
        \item $Q$ is a (possibly infinite) set of \emph{states}, which we assume to contain the \enquote{null state} $\perp$,
        \item $F \subseteq Q$ is a set of \emph{accepting states},
        \item $\iota$ is an \emph{initial mapping} that assigns to each boundaried graph $G$ with at most $\ell + 1$ vertices a state $\iota(G) \in Q$,
        \item $\delta \colon 2^\Omega \times 2^\Omega \times 2^\Omega \times 2^{\binom{\Omega}{2}} \times Q \times Q \to Q$ is a \emph{transition mapping} that describes the transitions.
    \end{itemize}
\end{definition}

\begin{definition}[\cite{Korhonen_Majewski_Nadara_Pilipczuk_Sokołowski_2023}]\label{def:run}
    Let $\autom = (Q,F,\iota,\delta)$ be a tree decomposition automaton and $\Tc = (T,\bag,\edges)$ be an annotated tree decomposition of a graph $G$. The \emph{run} $\run^{\Tc}_{\autom} \colon V(T) \to Q$ of $\autom$ on $\Tc$ is the unique labeling satisfying the following:
    \begin{itemize}
        \item For every leaf $l$ of $T$: $\run^{\Tc}_{\autom}(l) = \iota(G_l)$.
        \item For every non-leaf node $x \in V(T)$ with one child $y$: \[\run^{\Tc}_{\autom}(x) = \delta(\bag(x),\bag(y),\emptyset,\edges(x),\run^{\Tc}_{\autom}(y),\perp).\] 
        \item For every non-leaf node $x \in V(T)$ with two children $y$ and $z$: \[\run^{\Tc}_{\autom}(x) = \delta(\bag(x),\bag(y),\bag(z),\edges(x),\run^{\Tc}_{\autom}(y),\run^{\Tc}_{\autom}(z)).\]
    \end{itemize}
    A tree decomposition automaton $\autom$ \emph{accepts} a tree decomposition $(T,\bag,\edges)$ with root $r$ if $\run^{\Tc}_{\autom}(r) \in F$.
\end{definition}

\begin{observation}\label{obs:run_subtree}
    Let $\autom$ be a tree decomposition automaton and $\Tc = (T,\bag,\edges)$ be an annotated tree decomposition of a graph $G$. For a given node $t \in V(T)$ with parent $p$, the run $\run_{\autom}^{\Tc}(t)$ depends only on $\Tc\restriction_{V(T_t)}$ and $\adh(tp)$.
\end{observation}

We say that an automaton has \emph{evaluation time} $\tau$ if the functions $\iota$ and $\delta$ can be evaluated in time $\tau$, and, given a state $q \in Q$, it can be decided whether $q \in F$ in time $\tau$ as well.



In~\cite{Korhonen_Majewski_Nadara_Pilipczuk_Sokołowski_2023}, a treewidth automaton was introduced that, given a tree decomposition of width $\ell$ of a graph $G$ together with an integer $k \leq \ell$, decides whether $\tw(G) \leq k$. The algorithm is based on the linear-time dynamic programming algorithm by Bodlaender and Kloks~\cite{Bodlaender_Kloks_1996}, and thus called \emph{Bodlaender-Kloks automaton}.

\begin{lemma}[Lemma A.3 in~\cite{Korhonen_Majewski_Nadara_Pilipczuk_Sokołowski_2023}]\label{lem:treewidth_automaton}
    For every pair of integers $k \leq \ell$, there is a tree decomposition automaton $\mc{BK}_{k,\ell}$ of width $\ell$ with the following property: For any graph $G$ and its annotated tree decomposition $(T,\bag,\edges)$ of width at most $\ell$, $\mc{BK}_{k,\ell}$ accepts $(T,\bag,\edges)$ if and only if the treewidth of $G$ is at most $k$. The state space of $\mc{BK}_{k,\ell}$ is of size $\OO_{k,\ell}(1)$ and can be computed in time $\OO_{k,\ell}(1)$. The evaluation time of $\mc{BK}_{k,\ell}$ is $\OO_{k,\ell}(1)$ as well.
\end{lemma}

We now use this Bodlaender-Kloks automaton to define an internal treewidth automaton, that is, an automaton that, given a boundaried tree decomposition of width $\ell$ of a boundaried graph $G$ and an integer $k$, decides whether $\itw(G) \leq k$. Recall that $\itw(G) = \tw(G \setminus \bd(G))$. The high-level idea of our automaton is to keep a run of the Bodlaender-Kloks automaton for $G \setminus B$ for every set $B \subseteq V(G)$ that could be the boundary of $G$.

\internalTreewidthAutomaton*
\begin{proof}
    Let $\mc{BK}_{k,\ell} = (Q,F,\iota,\delta)$ be the automaton from \Cref{lem:treewidth_automaton}. In the following, we define a tree decomposition automaton $\mc{IBK}_{k,\ell}$ that, given a graph $G$ with its annotated tree decomposition $\Tc = (T,\bag,\edges)$ of width at most $\ell$, decides whether $\itw(G) \leq k$.

    For this, we introduce the following notation. Let $G$ be a graph and $\Tc = (T,\bag,\edges)$ be its annotated tree decomposition. For a set $S \subseteq V(G)$, we denote by $E_S \coloneq \{uv \in E(G) \mid u \in S\}$ the set of edges with at least one endpoint in $S$. Further, given a set $S \subseteq V(G)$, we define the annotated tree decomposition $\Tc_S = (T,\bag_S,\edges_S)$ with $\bag_S(x) = \bag(x) \setminus S$ and $\edges_S(x) = \edges(x) \setminus E_S$ for every node $x \in V(T)$. Note that $\Tc_S$ is an annotated tree decomposition of the graph $G - S$.
    
    Our goal is to define the automaton $\mc{IBK}_{k,\ell} = (Q^{2^{\ell+1}}, Q^{2^{\ell+1} - 1} \times F, \iota^*, \delta^*)$ in such a way that for any node $x \in V(T)$, the run $\run^{\Tc}_{\mc{IBK}_{k,\ell}}(x)$ is a list that contains for every set $S \subseteq \bag(x)$ the run $\run^{\Tc_S}_{\mc{BK}_{k,\ell}}(x)$ of the automaton $\mc{BK}_{k,\ell}$ on the tree decomposition $\Tc_S$ for the same node $x$. Specifically, when defining $\mc{IBK}_{k,\ell}$, we want to keep the following invariant for every graph $G$ and every annotated tree decomposition $\Tc = (T,\bag,\edges)$ of $G$.

    \begin{invariant}\label{inv:itw_automaton}
        For every node $x \in V(T)$, $\run^{\Tc}_{\mc{IBK}_{k,\ell}}(x)$ is a list that contains for every set $S \subseteq \bag(x)$, the run $\run_{\mc{BK}_{k,\ell}}^{\Tc_{S}}(x)$.
    \end{invariant}
    
    So, formally, a state of $\mc{IBK}_{k,\ell}$ is a list of (at most) $2^{\ell+1}$ states of $\mc{BK}_{k,\ell}$, each corresponding to a set $S \subseteq \Omega$ of size at most $\ell + 1$. Given a node $x \in V(T)$, the internal treewidth of the boundaried graph $G_x$, i.e., the treewidth of $G_x \setminus \bag(x)$, is at most $k$ if and only if $\run^{\Tc_{\bag(x)}}_{\mc{BK}_{k,\ell}}(x) \in F$, so our automaton $\mc{IBK}_{k,\ell}$ accepts exactly in this case. Technically, we realize this by defining an order $\prec$ on the sets $S \subseteq \Omega$ so that $|S_1| \leq |S_2|$ implies $S_1 \prec S_2$ for all sets $S_1,S_2 \in \Omega$, and sorting our ``state-lists'' accordingly (and filling up with $\perp$-entries from the left). Then, for every node $x \in V(T)$, the run $\run^{\Tc_{\bag(x)}}_{\mc{BK}_{k,\ell}}(x)$ corresponding to the set $S = \bag(x)$ will be the last entry in the list $\run^{\Tc}_{\mc{IBK}_{k,\ell}}(x)$, so taking $Q^{2^{\ell+1} - 1} \times F$ as the set of accepting states of $\mc{IBK}_{k,\ell}$ does the trick.

    Now, let us define the initial and transition mapping $\iota^*$ and $\delta^*$ of $\mc{IBK}_{k,\ell}$. We start with the initial mapping $\iota^*$. 
    For a graph $G$ with at most $\ell + 1$ vertices, we define $\iota^*(G)$ to be a list containing $\iota(G-S)$ for every set $S \subseteq V(G)$ (ordered with respect to $\prec$).

    \begin{claim}\label{claim:itw_automaton_leaf}
        Every leaf node $x \in V(T)$ satisfies \Cref{inv:itw_automaton}.
    \end{claim}
    \begin{claimproof}
        Let $x \in V(T)$ be a leaf node. Then, $\run^{\Tc}_{\mc{IBK}_{k,\ell}}(x) = \iota^*(G_x)$. By the definition of $\iota^*$, $\iota^*(G_x)$ is a list containing for every subset $S \subseteq V(G_x) = \bag(x)$ the state $\iota(G_x - S) \in Q$. Since $x$ is also a leaf node of $\Tc_{S}$ and by \Cref{def:run}, we have $\iota(G_x - S) = \run_{\mc{BK}_{k,\ell}}^{\Tc_{S}}(x)$, so \Cref{inv:itw_automaton} is satisfied.
    \end{claimproof}

    Next, we define the transition mapping $\delta^*$.
    Let $(X,Y,Z,J,q',q'')$  be a $6$-tuple, where  $X, Y$, and $Z$ are sets of size at most $\ell + 1$, $J \subseteq \binom{\Omega}{2} \setminus \binom{X}{2}$, and $q',q'' \in Q^{2^{\ell+1}}$. Let further $q'$ (and $q''$ if $q'' \neq \perp)$ be a list that contains for every set $S' \subseteq Y$ (or $S' \subseteq Z$) a state $q'_{S'} \in Q$ (or $q''_{S'} \in Q)$. Then, we define $\delta^*(X,Y,Z,J,q',q'')$ to be a list that contains for every set $S \subseteq B$ the state $\delta(X \setminus S, Y \setminus S, Z \setminus S, J \setminus E_S, q'_{S \cap Y}, q''_{S \cap Z})$. We will later discuss how we can determine this list within our running time. But first, let us prove that this definition satisfies \Cref{inv:itw_automaton}.

    \begin{claim}
        Every node $x \in V(T)$ satisfies \Cref{inv:itw_automaton}.
    \end{claim}
    \begin{claimproof}
        We show this claim by induction on the nodes of $T$.
        By \Cref{claim:itw_automaton_leaf}, every leaf node of $T$ satisfies \Cref{inv:itw_automaton}. So, let $x \in V(T)$ be a non-leaf node. We distinguish the two cases whether $x$ has one or two children.

        First, assume that $x$ has only one child $y$. By \Cref{def:run}, we have
        \begin{equation*}
            \run^\Tc_{\mc{IBK}_{k,\ell}}(x) = \delta^*(\bag(x),\bag(y),\emptyset,\edges(x),\run^\Tc_{\mc{IBK}_{k,\ell}}(y),\perp).
        \end{equation*}
        
        By induction, \Cref{inv:itw_automaton} is satisfied for $y$, i.e., $\run^{\Tc}_{\mc{IBK}_{k,\ell}}(y)$ is a list that contains for every $S \subseteq \bag(y)$, the state $\run_{\mc{BK}_{k,\ell}}^{\Tc_{S}}(y)$. Thus, $\delta^*$ is defined on this $6$-tuple and we can apply it to obtain $\run^\Tc_{\mc{IBK}_{k,\ell}}(x)$, which is a list that contains for every $S \subseteq \bag(x)$ the state
        \begin{equation*}
            \delta(\bag(x) \setminus S, \bag(y) \setminus S, \emptyset, \edges(x) \setminus E_S, \run_{\mc{BK}_{k,\ell}}^{\Tc_{S \cap \bag(y)}}(y), \perp)
        \end{equation*}

        If there is a vertex $v \in S \setminus \bag(y)$, then $v \notin V_y$, i.e., $v$ is not in the bag of a node in the subtree of $\Tc$ (or $\Tc_S$ or $\Tc_{S \cap \bag(y)}$) rooted at $y$. It follows that $\Tc_S\restriction_{V_y} = \Tc_{S \cap \bag(y)}\restriction_{V_y}$, so by \Cref{obs:run_subtree} $\run_{\mc{BK}_{k,\ell}}^{\Tc_{S \cap \bag(y)}}(y) = \run_{\mc{BK}_{k,\ell}}^{\Tc_{S}}(y)$. Therefore, we have
        
        \begin{equation*}
            \begin{split}
                &\delta(\bag(x) \setminus S, \bag(y) \setminus S, \emptyset, \edges(x) \setminus E_S, \run_{\mc{BK}_{k,\ell}}^{\Tc_{S \cap \bag(y)}}(y), \perp)\\
                = &\delta(\bag_S(x), \bag_S(y), \emptyset, \edges_S(x), \run_{\mc{BK}_{k,\ell}}^{\Tc_{S}}(y), \perp)\\
                = &\run_{\mc{BK}_{k,\ell}}^{\Tc_{S}}(x),
            \end{split}
        \end{equation*}
        so for every $S \subseteq \bag(x)$, the list $\run^\Tc_{\mc{IBK}_{k,\ell}}(x)$ contains $\run_{\mc{BK}_{k,\ell}}^{\Tc_{S}}(x)$ and \Cref{inv:itw_automaton} is satisfied.

        Now, assume that $x$ has two children $y,z$. By \Cref{def:run}, we have
        \begin{equation*}
            \run^\Tc_{\mc{IBK}_{k,\ell}}(x) = \delta^*(\bag(x),\bag(y),\bag(z),\edges(x),\run^\Tc_{\mc{IBK}_{k,\ell}}(y),\run^\Tc_{\mc{IBK}_{k,\ell}}(z)).
        \end{equation*}
        By induction, \Cref{inv:itw_automaton} is satisfied for $y$ and $z$, i.e., $\run^{\Tc}_{\mc{IBK}_{k,\ell}}(y)$ and $\run^{\Tc}_{\mc{IBK}_{k,\ell}}(z)$ are both lists that contain for every $S \subseteq \bag(y)$ and $S \subseteq \bag(z)$ the state $\run_{\mc{BK}_{k,\ell}}^{\Tc_{S}}(y)$ and $\run_{\mc{BK}_{k,\ell}}^{\Tc_{S}}(z)$, respectively. Thus, $\delta^*$ is again defined on this input and we can apply it to obtain $\run^\Tc_{\mc{IBK}_{k,\ell}}(x)$, which is a list that contains for every $S \subseteq \bag(x)$ the state
        \begin{equation*}
            \begin{split}
                &\delta(\bag(x) \setminus S, \bag(y) \setminus S, \bag(z) \setminus S, \edges(x) \setminus E_S, \run_{\mc{BK}_{k,\ell}}^{\Tc_{S \cap \bag(y)}}(y), \run_{\mc{BK}_{k,\ell}}^{\Tc_{S \cap \bag(z)}}(z))\\
                = &\delta(\bag_S(x), \bag_S(y), \bag_S(z), \edges_S(x), \run_{\mc{BK}_{k,\ell}}^{\Tc_{S}}(y), \run_{\mc{BK}_{k,\ell}}^{\Tc_{S}}(z))\\
                = &\run_{\mc{BK}_{k,\ell}}^{\Tc_{S}}(x),
            \end{split}
        \end{equation*}
        where we used the same arguments as before to show that $\run_{\mc{BK}_{k,\ell}}^{\Tc_{S \cap \bag(y)}}(y) = \run_{\mc{BK}_{k,\ell}}^{\Tc_{S}}(y)$ and $\run_{\mc{BK}_{k,\ell}}^{\Tc_{S \cap \bag(z)}}(z) = \run_{\mc{BK}_{k,\ell}}^{\Tc_{S}}(z)$. Again, this shows that \Cref{inv:itw_automaton} is satisfied, and thus finishes the proof of the claim.
    \end{claimproof}

    This completes the description of the tree decomposition automaton $\mc{IBK}_{k,\ell}$ and, as argued before, proves the correctness of the automaton. Note that the states of $\mc{IBK}_{k,\ell}$ are lists of length $2^{\ell + 1}$ consisting of states from $\mc{BK}_{k,\ell}$. Since the state space of $\mc{BK}_{k,\ell}$ is of size $\OO_{k,\ell}(1)$ and can be computed in time $\OO_{k,\ell}(1)$, the the same holds for $\mc{IBK}_{k,\ell}$. Moreover, the evaluation of $\delta^*$ requires $2^{\ell+1}$ evaluations of $\delta$, where determining the input for each of these evaluations is possible in time $\OO(\ell^2)$. To decide whether a state of $\mc{IBK}_{k,\ell}$ is accepting, we only need to look at the last state in the list and let $\mc{BK}_{k,\ell}$ decide, which takes time $\OO_{k,\ell}(1)$. Thus, the evaluation time of $\mc{IBK}_{k,\ell}$ is $\OO_{k,\ell}(1)$ as well.
\end{proof}